\RequirePackage{fix-cm}
\documentclass[smallextended]{svjour3}       
\smartqed  
\usepackage{amsmath,amssymb}
\usepackage{mathptmx}
\usepackage{latexsym}
\usepackage{amsmath,amssymb, amsfonts}
\usepackage{verbatim}
\usepackage[all]{xy}
\usepackage[dvips]{graphicx}
\usepackage{float}

\numberwithin{equation}{section}

\newcommand{\Vdo}{\underline{V}}
\newcommand{\Vup}{\overline{V}}
\newcommand{\Udo}{\underline{U}}
\newcommand{\Uup}{\overline{U}}
\newcommand{\Me}{\mathcal{M}}
\newcommand{\Se}{\mathcal{S}}
\newcommand{\He}{\mathcal{H}}
\newcommand{\Swe}{\mathcal{S}^{\mathcal{W}}}

\newcommand{\Sup}{\mathcal{S}^{\textrm{up}}}
\newcommand{\Sdo}{\mathcal{S}^{\textrm{do}}}
\newcommand{\se}{{\bf S}}
\newcommand{\seup}{{\bf S}^{\textrm{up}}}
\newcommand{\sedo}{{\bf S}^{\textrm{do}}}

\begin{document}

\title{Trajectory based  models.\\ Evaluation of minmax price bounds.
}

\author{Iv\'an Degano \and Sebasti\'an Ferrando \and  Alfredo Gonzalez}
\authorrunning{I. Degano, S. Ferrando, A. Gonzalez}

\institute{I. Degano, Universidad Nacional de Mar del Plata, CONICET \at
              3350 Funes St., Mar del Plata, Argentina \\
              Tel.: +549-0223-4752426 x 234\\
              \email{ivandegano@mdp.edu.ar}           
           \and
           S. Ferrando, Ryerson University \at
              350 Victoria St., Toronto, Canada\\
              Tel.: 416-979-5000 x 7415\\
              \email{ferrando@ryerson.ca}           
           \and
           A. Gonzalez, Universidad Nacional de Mar del Plata \at
              3350 Funes St., Mar del Plata, Argentina \\
              Tel.: +549-0223-4752426 x 234\\
              \email{algonzal@mdp.edu.ar}           
}

\date{Received: date / Accepted: date}

\maketitle

\begin{abstract}
The paper studies  sub and super-replication price bounds for contingent claims defined on general trajectory based market models. No prior probabilistic or topological assumptions are placed on the trajectory space, trading is assumed to take place at a finite number of occasions but not bounded in number nor necessarily equally spaced in time. For a given option, there exists an interval bounding the set of possible fair prices; such interval exists under more general conditions than the usual no-arbitrage requirement. The paper develops a backward recursive method to evaluate the option bounds; the global minmax optimization, defining the price interval, is reduced to a local minmax optimization via dynamic programming. Trajectory sets are introduced for which existing non-probabilistic  markets models are nested as a particular case. Several examples are presented, the effect of the presence of arbitrage on the price bounds is illustrated.
\end{abstract}

\section{Introduction}

In an incomplete market model, the classical theory shows that, under no arbitrage assumptions, there exists an interval of fair prices. Such an interval is given by the sub and super-replication bounds introduced first in a diffusion setting in \cite{elKaroui} (see \cite{follmer} for a general discrete time formulation). The super-replication price bound of an European contingent claim $Z$ equals the supremum of its expectation over the set of equivalent martingale measures (with an analogous characterization for sub-replication). For a discrete time setting, such dual formulation can be found in \cite{follmer} and \cite{cutland} (the second reference assumes a finite probability space).

It turns out that for a large class of stochastic models the fair price interval degenerates to absolute (i.e. model independent as in (\cite{merton})) no-arbitrage bounds. This is shown in \cite{eberlein} for continuous time and in \cite{carassus1} for discrete time.
These results rely on the assumption of an unbounded range and a non atomic distribution for the increments of the modeling stochastic processes (i.e. the underlying). Reference \cite{carassus2} studies a class of stochastic models for which the fair price interval does not trivialize to absolute bounds.
A popular alternative, in order to reduce the size of the fair price interval, is to allow trading with liquid options in order to better approximate an illiquid derivative. Presumably, this is a way to acknowledge the limitations of the original model proposed for the underlying in order to account for the degrees of freedom influencing the derivative's market price.

There is uncertainty in the modelling of any assumed probabilistic distribution as well as in the specification of the support of the modelling stochastic process.   An example of such uncertainty is the  modelling of crashes (\cite{desmettre}) where, the number, timing and size of a downwards stock change (a {\it crash}) is treated without probabilistic assumptions. An example requiring a set of non equivalent measures is provided by the uncertain volatility model (\cite{avellaneda}). A related development is given by sublinear expectations and their associated stochastic calculus (\cite{peng}).  In order to accommodate such uncertainties, our general setting requires no
prior stochastic assumptions.  Recent and related literature also develops results that weaken, or eliminate entirely, stochastic assumptions; as examples, we mention robust versions of FTAP in \cite{riedel}, \cite{burzoni}, \cite{burzoni2} and \cite{cheridito2}.

The present paper develops computational results of fair price bounds for a large class of non probabilistic models built around a trajectory space. The general framework in discrete time is developed in \cite{deganoI} where detailed justifications and connections with standard stochastic models can be found.  The setting grew as a generalization of a model proposed in \cite{BJN} (see also the book exposition in \cite{rebonato}). A related reference is given by \cite{roorda}.
We show in examples that the resulting fair price intervals are much narrower that the ones given by the absolute bounds and that the task of modeling trajectory sets directly, as opposed to firstly prescribing a probability distribution and then obtaining its support as a by product, is a worthwhile modeling enterprise. Realistic models and preliminary comparison with market data can be found in \cite{fleck}. A basic result in
\cite{deganoI} is the proof of existence of a fair price interval despite the presence of a certain kind of  arbitrage. We show numerically the effect of such arbitrage on the price bounds.

It is natural to inquire about the differences between the fair price intervals for stochastic and trajectory based models.
A main technical difference is that the superhedging inequalities in a stochastic setting are requested to hold a.e., this implies the need to evaluate essential infima and suprema which are, in general, computationally intractable. When non equivalent measures are involved there is the need to use polar sets. Literature providing a general approach to evaluate sub and super-replication bounds in a general discrete time setting is scarce, we are only aware of \cite{carassus1}.
Our setting and results do not require to deal with sets of measure zero and hence complications of that nature are avoided from the outset. We establish general results that allow the evaluation of fair price bounds for a large class of trajectory based models. We restrict our attention to a single tradable asset but expect that the results obtained can be extended to higher dimensions without essential complications. For comparison purposes, we mention the reference \cite{kahale} that also works in a model independent setting, in particular no apriori fixed measure is assumed, and allows for static and dynamic hedging in the super-replication portfolio.

The financial context is of a riskless bond with zero interest rates and one risky asset. We consider a financial discrete market $\Me=\Se^{\mathcal{W}} \times \He$, elements $\{{\bf S}_i= (S_i, W_i, m)\} \in \Swe$ are {\it trajectories}. Coordinates $S_i$ are the values of the tradable underlying while the variable $W_i$, possibly vector valued, represents values of other observable financial variables used to define the trajectory set (the coordinate $m$ is described later). $H = \{H_i\} \in \mathcal{H}$ are functions $H_i: \mathcal{S}^{\mathcal{W}}\rightarrow \mathbb{R}$ representing the  trading strategies. The general class of models included in the formalism allow for certain arbitrage opportunities while, at the same time, providing a fair price interval for options without introducing logical or  practical inconsistencies.

We present effective and rigorous results that allow to evaluate  the super- and sub-replication pricing interval $[\Vdo(S_0,Z,\Me),\Vup(S_0,Z,\Me)]$ given by a minmax optimization in \cite{deganoI} (see Definition \ref{conditionalBounds} in the present paper). The resulting algorithm is a dynamic programming based optimization applicable to general trajectory sets $\Swe$. To efficiently deal with the resulting local minmax optimization, we propose a geometric procedure
consisting in computing the convex hull of a set of future stock values (see Section \ref{sec:convexEnvelope}). This represents the so called {\it convex hull algorithm} (introduced informally in \cite{BJN}) but here made rigurous and extended to a general setting.
In contrast to available methods evaluating the convex hull
(\cite{andrew}, for example) we isolate a relevant sector of the convex hull containing the required solution. Moreover, our approach works for the case of an infinite number of points, its end effect is to reduce the local minmax to a single maximization. This last step is achieved by parametrizing the hedging parameter by a geometric ratio and represents the essence of the convex hull algorithm. The hedging ratio is a discretized version of the delta hedging term appearing in the stochastic setting and gives an optimal hedging. We provide a formal analysis of the procedure in a general setting.

The resulting algorithm allows to evaluate fair price bounds for a realistic class of options and a general class of trajectory sets.  We prove that, for a class of models and  options with convex payoff, the super-replication price is equal to the replication price in a Cox-Ross-Rubinstein model (see \cite{cox}): this result has been already obtained in \cite{roorda} in a non probabilistic fixed time framework and \cite{carassus2} in a probabilistic context. Also, we extend a model from 
\cite{BJN} (see also \cite{rebonato}) by allowing for trajectory dependent quadratic variation. Finally, relevant numerical examples illustrate the viability of the approach and some of the characteristics of the models studied.

The paper is organized as follows, Section
\ref{sec:generalFramework} provides the general framework of the paper and describes notation and relevant results to be used in the remaining of the paper. That section introduces the notion of $0$-neutral and the fair price interval. Section \ref{dynamicMinimaxBounds} establishes, under appropriate
conditions, how to recover the global minmax optimization defining
the bound prices by means of an iterative dynamic programming
procedure. Section \ref{sec:convexEnvelope} describes how the
iterative procedure described in Section \ref{dynamicMinimaxBounds}
can actually be implemented by an efficient, geometric based method which
we call {\it convex hull algorithm}. Section \ref{chartsAndTrajectories} describes some simple models as well as a class of models allowing trajectory dependent values of (sampled) quadratic variation. For this last case, we describe in Section \ref{gridDataStructure} a data structure supporting the implementation of the models. Finally, Section \ref{sec:conclusions} concludes by providing a perspective on the
paper as well as some speculation on possible extensions.
Appendices contain complementary and technical results.

\section{General framework}  \label{sec:generalFramework}

Usually, financial discrete markets fixes a finite partition of the time interval $[0,T]$ where transactions are carried out. The index $i$ refers to time $t_i$ between $0$ and $T$.
For the sake of flexibility and generality, trajectories $\se \in \Swe$,
are of the form $\se = \{(S_i, W_i, m)\}$ where $W_i$ belongs to abstract sets $\Omega_i$ from which we only require to have defined an equality relationship. Such coordinates  are referred to as {\it additional sources of uncertainty} (analogously to the case of an augmented filtration containing the canonical filtration). In financial terms, the quantities $W_i$ are considered to be  observable. This framework allows the investor to rebalance the portfolio as a result of an arbitrary market event, for example, the quadratic variation reached a certain value. So, time  dependent trajectories are mapped to a space which depends on the variable $W$ which, presumably, better jointly constraints the sequence of pairs $(S_i, W_i)$.
As we will price options expiring with a finite time horizon, we need an extra information indicating when a trajectory has reached the final time $T$. We denote this trajectory coordinate by $m \in \mathbb{N}$. The introduction of $m$ is important to the calculation of the fair price interval as the generality of the setting does not necessarily require that the coordinate $W_i$ carries any explicit time information (see Section \ref{chartsAndTrajectories} for several examples). The coordinate $m$ could have been formally incorporated into $W_i$ but, for clarity, we decided not to do so.

We reproduce some needed definitions from \cite{deganoI} which should be consulted for further details.

\begin{definition}[Trajectory Set]  \label{trajectories}
Given the real numbers $s_0$ and $w_0$, a set of (discrete) trajectories $\Swe = \Swe (s_0,w_0)$ is a subset of the
following set
\begin{equation}  \nonumber
\Swe_{\infty}= \Swe_{\infty}(s_0,w_0)= \{{\bf S} = \{{\bf S}_i \equiv (S_i, W_i,m)\}_{i \geq 0}: S_i \in \Sigma_i, W_i \in \Omega_i, m \in \Theta,~ S_0 = s_0, W_0 = w_0\},
\end{equation}
where $\{\Sigma_i\}_{i\ge 1}$ and $\{\Omega\}_{i\ge 1}$ are families of subsets of $\mathbb{R}$ and $\Theta \subset \mathbb{N}$. Elements $\se \in \Swe$ are called trajectories.
\end{definition}

We remark that if ${\bf S}^1 = \{(S_i^1, W_i^1, m^1)\}$ and ${\bf S}^2 = \{(S_i^2, W_i^2, m^2)\}$ are two trajectories, $\se^1_1$ could take place in a different time than $\se^2_1$, although $W^1_1=W^2_1$.

For $\se \in \Swe$  we will use the notation $\Delta_iS \equiv
S_{i+1}-S_i$ for $i\ge 0$ and  define $M:\Swe \rightarrow \mathbb{N}$ to be the projection function over the third coordinate of $\se$, that is $M(\se)=m$. The following conditional spaces will play a key role. Let $k \geq 0$, for $\se  \in \Swe$ such that $M(\se)>k$ set:
\begin{equation} \nonumber
{\Swe}_{(\se,k)}\equiv\{\se' \in \Swe: M(\se')>k \mbox{ and } (S'_i,W'_i)= (S_i,W_i) \; \forall \; 0 \le i \le k \}.
\end{equation}
Notice $\Swe_{(\se,0)} = \Swe$ and that if $\se' \in \Swe_{(\se,k)}$, then $\Swe_{(\se',k)}=\Swe_{(\se,k)}$. Whenever convenient, the tuple $(\se,k)$ will be referred generically as {\it a node}.

A portfolio in our model will be a function over the set of trajectories as in \cite{deganoI}, but we have modified slightly the non anticipative condition to accommodate the variable $m$.

\begin{definition}[Portfolio Set]
A portfolio $H$ is a sequence of (pairs of) functions $H = \{\Phi_i
= (B_i, H_i)\}_{i\geq 0}$ with $B_i, H_i: \Swe \rightarrow
\mathbb{R}$.
\begin{itemize}
\item A portfolio $H$ is said to be \emph{admissible} for $\Swe$ if for each $\se \in \Swe$ there exists a nonnegative integer $N=N_H(\se)$ such that $H_i(\se)=0$ for all $i \ge N_{H}(\se)$ and $N_{H}(\se) \le M(\se)$.
\item A portfolio $H$ is said to be self-financing at $\se \in \Swe$ if for all $i \geq 0$,
\begin{equation} \label{selfFinancing}
H_i(\se) ~S_{i+1} + B_i(\se) = H_{i+1}(\se) S_{i+1} + B_{i+1}(\se).
\end{equation}
\item A portfolio $H$ is  called non-anticipative if for all $\se, \se' \in
\Swe$, satisfying $S'_k = S_k$ and $W'_k=W_k$ for all $0 \leq k \leq i$ with $i < \min\{ N_H(S),N_H(S')\}$, it
then follows that $\Phi_i(\se) = \Phi_i(\se')$.
\end{itemize}
\end{definition}

Given $\se \in \Swe$ and a self-financing portfolio $H$, the portfolio value defined by $V_H(i,\se)=B_i(\se)+H_i(\se)S_i$ is equal to
\[ V_H(i,\se)=V_H(0,S_0) + \sum_{k=0}^{i-1} H_k(\se)\Delta_kS,\]
during the period $[i,i+1)$ for $i=0, \dots, N_H(\se)-1$. Of course, $V_H(0,S_0)=B_0(\se) + H_0(\se)S_0$. Clearly, to specify self-financing portfolios, it is enough to
provide  sequences $H =\{H_i\}$ of non-anticipative functions and an
associated real number $V_0= V_H(0, S_0)$.

As suggested above, the non-anticipative condition is slightly different to Definition 2 in \cite{deganoI}; the new condition is more general and it is useful in the setting of the present work. If ${\bf S}=\{(S_i,W_i, m) \}_{i\ge 0}$ and ${\bf S}'=\{(S'_i,W'_i, m') \}_{i\ge 0}$, we could alter the condition to $i < \min\{N_H({\bf S}),N_H({\bf S}') \}$ to $i < \min\{m,m'\}$. The condition using $N_H$ is more general and it incorporates investors' information. For each strategy the investor chooses a stage to liquidate the portfolio with respect to the trajectory ${\bf S}$ taking into account the information of the market or merely his intuition. This selection could be different for a trajectory ${\bf S}$ which is equal to ${\bf S}'$. As a special case, in Section \ref{dynamicMinimaxBounds}, we will impose $N_H({\se})=m$ for all $H \in \He$.

A trajectory based discrete market (or trajectory market for short)
is defined by $\Me = \Swe \times \He$ where elements $H \in \mathcal{H}$ are admissible, non-anticipative and  self-financing at
each $\se \in \Swe$. The models are discrete in the sense that we index potential portfolio rebalances by integer numbers. Otherwise, stock charts and investment amounts can take values in general subsets of the real numbers and time can flow continuously. The zero portfolio is assumed to belong to $\He$ and we take $N_0=0$.

As indicated, some of the above definitions involve minor
modifications from material in \cite{deganoI} but most
algebraic manipulations in that reference only involve the first coordinate $S_i$ (in the triples $(S_i,W_i,m)$). This remark can be used to show that the results we will rely upon from \cite{deganoI} are valid in the setting of the current paper.

The following notion of discrete bounded market will be needed in several instances later in the paper.

\begin{definition}[$n$-Bounded Market]  \label{nBoundedDefinition}
A market $\Me=\Se^{\mathcal{W}} \times \He$ is called $n$-bounded if there exists a constant $n$ so that:
\begin{equation} \nonumber
~~~ \sup_{{\bf S} \in \Se^{\mathcal{W}}} ~~~M({\bf S}) \leq n.
\end{equation}
\end{definition}
We refer to $\mathcal{M}$ as {\it bounded} when reference to $n$ is immaterial.


We use the following  definition of no-arbitrage market.

\begin{definition}[Arbitrage-Free Market] \label{ArbitrageDefinition}
Given a discrete market $\Me= \Swe \times \He$, we will call $H \in \He$ an arbitrage strategy if:
\begin{itemize}
\item $\forall \se \in \Swe$,  $V_{H}(N_H(\se), \se)\geq V_{H}(0, \se_0)$.
\item $\exists \se^{\ast} \in \Swe$ satisfying $V_{H}(N_H(\se^{\ast}), \se^{\ast})) > V_{H}(0, \se_0)$.
\end{itemize}
We will say  $\Me$ is arbitrage-free if $\He$ contains no arbitrage strategies.
\end{definition}

Let $Z:\Swe \rightarrow \mathbb{R}$ denote a general function, from time to time, we will refer to such function informally as the {\it derivative} or {\it payoff function}. See Appendix \ref{minMaxFunctionsProofs} for general conditions on $Z$ that guarantee finiteness of the quantities introduced below.

\begin{definition}[Conditional Minmax Bounds]
Given a discrete market $\mathcal{M} = \Swe \times \mathcal{H}$, $k \ge 0$ and $\se \in \Swe$ such that $M(\se)>k$. Let $Z$ a function defined on $\Swe$, define
\begin{equation} \label{conditionalBounds}
 \overline{V}_k(\se, Z, \mathcal{M}) \equiv ~\inf_{H \in \mathcal{H}}~\sup_{\se' \in \Swe_{(\se, k)}} [Z(\se') -  \sum_{i=k}^{N_H(\se)-1}H_{i}(\se') \Delta_iS'].
\end{equation}
Also define $\underline{V}_k(\se, Z, \mathcal{M}) = -\overline{V}_k(\se, -Z, \mathcal{M})$. Since $\overline{V}_0(\se, Z, \mathcal{M})$ and $\underline{V}_0(\se, Z, \mathcal{M})$ depend on $\se$ only through $S_0$, we adopt the notation  $\overline{V}(S_0,Z, \mathcal{M})$ and  $\underline{V}(S_0, Z, \mathcal{M})$ respectively. These quantities are called  {\it price bounds}.
\end{definition}

The price bounds can be recast in a more familiar way:
\begin{eqnarray*}
  \overline{V}(S_0,Z, \mathcal{M}) &=& \inf \left\lbrace V_0 : \exists H \in \He , V_0+ \sum_{i=0}^{N_H(\se)-1} H_i(\se)~\Delta_i \se \ge Z(\se),~ \forall \se \in \Swe \right\rbrace \\
  \underline{V}(S_0,Z, \mathcal{M}) &=& \sup \left\lbrace V_0 : \exists H \in \He , V_0 + \sum_{i=0}^{N_H(\se)-1} H_i(\se)~\Delta_i \se \le Z(\se),~ \forall \se \in \Swe \right\rbrace
\end{eqnarray*}

We know from financial stochastic models, that there exists an arbitrage-free price interval for the derivative $Z$ if the market does not contain arbitrage strategies. In our context, the free arbitrage condition is replaced by the notion of a $0$-neutral market that plays a key role.
\begin{definition}[$0$-Neutral Market] \label{def:conditionally0Neutral}
The market is {\it conditionally $0$-neutral} at node $(\se,k)$ if
\begin{equation}  \nonumber
 \overline{V}_k(\se, Z=0, \mathcal{M}) =0.
\end{equation}
For $k=0$, we will just refer to $\mathcal{M}$ as $0$-neutral.
\end{definition}
The  notion of $0$-neutral market, taken from \cite{deganoI}, was originally introduced in \cite{BJN} and was considered equivalent to arbitrage-free in their context. In our general setting, it is only a necessary condition for a discrete market to be arbitrage-free \cite[Corollary 1]{deganoI} while simultaneously allowing for arbitrage opportunities and a well defined theory of option pricing. $0$-neutrality is  key to obtain a well defined fair price interval. Theorem \ref{havingAnIntervalTheorem} is stated for a bounded market and $\mathcal{H}$ assumed closed under addition. This is done to avoid  introducing further notions, the result holds in more generality as can be seen in \cite{deganoI}.

\begin{theorem}[Price Interval] \label{havingAnIntervalTheorem}
Consider a bounded discrete market $\Me= \Swe \times \He$ and a function $Z$ defined on $\Swe$;  fix $\se \in \Swe$ and $k \geq 0$. If $\He$ is closed under addition and $\Swe_{(\se,k)}$ is conditionally $0$-neutral, then
\begin{equation} \nonumber
\underline{V}_k(\se, Z, \Me) \le \Vup_k(\se, Z, \Me).
\end{equation}
In particular $\underline{V}(S_0, Z, \Me) \le \Vup(S_0, Z, \Me).$
\end{theorem}
\begin{proof}
The result follows from the same calculations as in \cite[Theorem 1]{deganoI}.
\qed \end{proof}

Under the assumption that $\Vdo(S_0, Z, \Me) \le \Vup(S_0, Z,\Me)$, we will call $[\Vdo(S_0, Z, \Me), \Vup(S_0,
Z, \Me)]$ the price interval of $Z$ relative to $\Me$. Appendix \ref{minMaxFunctionsProofs},
provides conditions for the boundedness of $\Vup(S_0,Z, \Me)$ and $\Vdo(S_0,Z,\Me)$.

\vspace{.1in}
The notion of attainability is basic in option pricing.

\begin{definition}
  Given a discrete market $\Me=\Swe \times \He$, a function $Z$ is called attainable if there exist $H^Z \in \He$ such that
  \begin{equation} \nonumber 
Z(\se)=   V_{H^Z}(0, s_0)+\sum_{i=0}^{N_{H^Z}(\se)-1} H^Z_i(\se)\Delta_iS, ~~~ \mbox{for all}~~~~ \se \in \Swe.
\end{equation}
\end{definition}

In stochastic frameworks there exists a unique fair price for an attainable option. The following analogue result holds in the present setting.

\begin{corollary} \label{onevalue}
  Consider a discrete market $\Me = \Swe \times \He$ , $\se \in \Se$, $k \ge 0$ and $Z$ a function on $\Swe$ and assume the conditions of Theorem \ref{havingAnIntervalTheorem}. If $Z$ is attainable then $\Vup_k(\se,Z,\Me) =\Vdo_k(\se,Z,\Me)$.
\end{corollary}
\begin{proof}
  The proof is given in \cite[Corollary 6]{deganoI}.
\qed \end{proof}

\subsection{Global, Conditional and Local Concepts}

Given the central role of $0$-neutrality in our framework, it is imperative to find simple to check conditions guaranteeing a market to be $0$-neutral.
Definition \ref{localDefinitions} below introduces two basic concepts towards that goal: a local, and  portfolio independent, analogue on $\Swe$ of the $0$-neutral property of $\Me$ and a strengthening of this notion representing the local analogue of the arbitrage free property.

\begin{definition} [$0$-Neutral \& Arbitrage-Free Nodes] \label{localDefinitions}
Given a trajectory space $\Swe$ and a node $(\se,k)$:
\begin{itemize}
\item $(\se,k)$ is called a \emph{$0$-neutral node} if
\begin{equation} \label{cone}
\sup_{\se' \in \Swe_{(\se, k)}}~~ (S'_{k+1} - S_{k}) \geq 0 \mbox{ and } \inf_{\se' \in \Swe_{(\se, k)}} (S'_{k+1} - S_{k}) \leq 0.
\end{equation}
\item $(\se,k)$ is called an \emph{arbitrage-free node} if
\begin{equation} \label{upDownProperty}
\sup_{\se' \in \Swe_{(\se, k)}} (S'_{k+1} - S_{k}) >0 \mbox{ and } \inf_{\se' \in \Swe_{(\se, k)}} (S'_{k+1} - S_{k}) <   0
\end{equation}
or
\begin{equation} \label{flatNode}
\sup_{\se' \in \Swe_{(\se, k)}} (S'_{k+1} - S_{k}) = \inf_{\se' \in \Swe_{(\se, k)}} (S'_{k+1} - S_{k}) = 0.
\end{equation}
\end{itemize}
$\Swe$ is called \emph{locally $0$-neutral} if (\ref{cone}) holds at each node $(\se, k)$. $\Swe$ is said to be \emph{locally arbitrage-free}  if either (\ref{upDownProperty}) or (\ref{flatNode}) hold at each node  $(\se, k)$. If just (\ref{upDownProperty}) holds at each node, it is said that $\Swe$ satisfies the up-down property. A node that satisfies (\ref{upDownProperty}) will be called an up-down node, and a node satisfying (\ref{flatNode}) will be called a {\it flat node}. A node that is  $0$-neutral but that is not an arbitrage-free node, will be called an {\it arbitrage node}.
\end{definition}

The next Proposition gives local conditions ensuring that a discrete market is conditionally $0$-neutral. As already pointed out, only the first coordinate $S_i$ (in the triples $(S_i, W_i,m)$) appear in most algebraic manipulations; therefore,  the following results from \cite{deganoI} hold in our setting.

\begin{proposition} \label{0-neutral}
Consider a bounded discrete market $\Me=\Swe \times \He$,
\begin{itemize}
  \item If $\Swe$ is locally arbitrage-free, then it is locally $0$-neutral.
  \item If $\Swe$ is locally $0$-neutral, then it is $0$-neutral (as per Definition \ref{def:conditionally0Neutral}).
  \item If $\Swe$ is locally arbitrage-free and $N_H$ is a stopping time(in the sense of Definition \ref{stoppingtime} in Appendix \ref{minMaxFunctionsProofs}) for all $H \in \He$ then $\Me$ is arbitrage free.
\end{itemize}
\end{proposition}
\begin{proof}
  The first item follows immediately from Definition \ref{localDefinitions} above. The next two items are special cases of \cite[Theorem 2]{deganoI} and \cite[Corollary 3]{deganoI}.
\end{proof}

\section{Dynamic Minmax Bounds}  \label{dynamicMinimaxBounds}

Arguably, attempting a direct evaluation of the minmax optimization
required in (\ref{conditionalBounds}) and in related expressions, is a daunting
task. Moreover, the minmax formulation of the problem gives no clues
on how to construct the hedging values $H_i(\se)$, for a given  payoff
$Z$, by means of the unfolding path values $S_0, S_1, S_2, \ldots $

Consider next another pair of numbers namely
$\underline{U}_0(S_0, Z, \Me)$ and $\overline{U}_0(S_0, Z, \Me)$. These numbers are obtained through a dynamic, or iterative, definition each instance involving a local minmax optimization. Using these definitions we provide
conditions  under which the global and the iterated definitions
coincide.

A special case of the iterative construction was introduced
informally in \cite{BJN} (see also \cite{kolokoltsov} and \cite{roorda}) for a specific discrete market model. Here we formalize the validity of the approach in such a way that becomes
available in a more general class of models and at the same time
indicating the differences with the global minmax approach. The
references  \cite{bertsekas} and \cite{bertsekasAndShreves}  provide a dynamic programming version of a global minmax optimization. Our approach differs as we make use of specific
hypothesis present in our  setting.

Markets will be assumed to be bounded and that all portfolios are liquidated on the expiration time $T$; that is, for each $H \in \He$, $N_H(\se)=M(\se)=m$. Further restrictions on $\He$ will be introduced as needed.

The following inductive definition gives the basic dynamic programming formulation to compute $\overline{V}(S_0,Z,\Me)$.

\begin{definition}[Dynamic Bounds] \label{dynamicBounds}
Consider an $n$-bounded, discrete market $\Me$; for a given
function $Z$ defined on $\Swe$, any $\se\in \Swe$, and
$0\le i \le n$ set
\begin{equation}  \label{DynamicBound}
\overline{U}_i(\se, Z,\Me)= \left\lbrace \begin{array}{lcc}
\inf\limits_{H\in\He} \sup\limits_{\se' \in \Swe_{(\se,i)}}
[\overline{U}_{i+1}(\se', Z,\Me) - H_i(\se)
\Delta_iS'] &\mbox{if}& 0\le i< M(\se),\\
Z(\se) &\mbox{if}& i=M(\se),\\
0 &\mbox{if}& i>M(\se).
\end{array}\right.
\end{equation}
Also define $\underline{U}_i(\se, Z, \Me) = -\overline{U}_i(\se,-Z, \Me)$.
\end{definition}

\begin{remark}$\;$
\begin{enumerate}
\item Since $\Udo_0(\se,Z,\Me)$ and $\Uup_0(\se,Z,\Me)$ depend on ${\bf S}$ only through $(S_0,W_0)$, we adopt the notation $\underline{U}_0(S_0, Z, \Me)$ and $\overline{U}_0(S_0, Z, \Me)$, respectively.
\item Note that in Definition \ref{dynamicBounds},  $H_i(\se) = H_i(\se')$ for all $\se' \in \Swe_{(\se,i)}$.
\end{enumerate}
\label{remark}
\end{remark}

The next remark shows that whenever $M$ is a stopping time, in the sense of Definition \ref{stoppingtime} in Appendix \ref{minMaxFunctionsProofs}, the dynamic bounds depend only on the history of the trajectory.

\begin{remark} \label{StoppingTimeEqualityOfDynamicBounds} Assume $M$ is a stopping time and fix $\se\in\Swe$. Let $i \in \mathbb{N}$ and $\se' \in \Swe$ be such that $(S'_j,W'_j)=(S_j.W_j)$ for all $0 \le j \le i$. If $i\ge M(\se)$, then $M(\se)=M(\se')$ and it follows by definition that $\Uup_i(\se,Z,\Me)=\Uup_i(\se',Z,\Me)$. If $i<M(\se)$ since $M$ is stopping time, $\se' \in \Swe$ belongs to $\Swe_{(\se,i)}$. Consequently, $\Swe_{(\se,i)}=\Swe_{(\se',i)}$ and
\begin{eqnarray*} \nonumber
\Uup_i(\se',Z,\Me)&=&\inf_{H\in\He} \sup_{\tilde{\se}\in\Swe_{(\se',i)}}
[\overline{U}_{i+1}(\tilde{\se}, Z,\Me) - H_i(\se') \Delta_i\tilde{S}] = \\
&=& \inf_{H\in\He} \sup_{\tilde{\se}\in\Swe_{(\se,i)}}
[\overline{U}_{i+1}(\tilde{\se}, Z,\Me) - H_i(\se) \Delta_i\tilde{S}]\} = \Uup_i(\se,Z,\Me).
\end{eqnarray*}
Therefore $\Uup_i(\se,Z,\Me)=\Uup_i(\se',Z,\Me)$ for all $\se' \in \Swe$ such that $(S'_j,W'_j)=(S_j.W_j)$ for all $0 \le j \le i$ and $i \ge 0$.
\end{remark}

For any $\se \in \Swe$ and $0\le k < M(\se)$, we let $I^k_{\se}$ to be the set of  portfolio values at node $({\bf S}, k)$, in other words
\begin{equation}\label{PortfoliosRange}
I^k_{\se}\equiv \{H_k(\se): H\in \He\} \subseteq \mathbb{R}.
\end{equation}
Thus, by item $(2)$ in Remark \ref{remark}, we can rewrite the expression in (\ref{DynamicBound}) for $~0\le k < M(\se)$,
\begin{equation}  \label{RangeDynamicBound}
\Uup_k(\se, Z,\Me)= \inf_{u\in I^k_{\se}} \sup_{\se'\in\Swe_{(\se,k)}} [\Uup_{k+1}(\se', Z,\Me) - u ~\Delta_kS'].
\end{equation}

As we mentioned earlier, one of the purpose is to compare the global bound $\Vup(S_0,Z,\Me)$ with the dynamic bound $\Uup_0(S_0,Z,\Me)$. Without any assumptions, we have the following general relationship.

\begin{theorem}  \label{useful}
For any function $Z$ defined on a discrete $n$-bounded market
$\Me = \Swe \times \He$ and $0\le k <n$,  the following
inequality holds:
\begin{equation} \label{firstInequality}
\overline{U}_k(\se, Z, \Me) \leq \overline{V}_k(\se, Z,\Me),
\end{equation}
for all $\se \in \Swe$ such that $M(\se)>k$. Furthermore $\underline{U}_k(\se, Z, \Me) \geq
\underline{V}_k(\se, Z, \Me)$ is also valid.
\end{theorem}
\begin{proof}
We proceed  by backward induction on $k$. For $k=n-1$ and $\se \in \Swe$ with $M(\se)> n-1$, all $\se' \in \Swe_{(\se,k)}$ satisfy $M(\se')=n$. Then, we have from (\ref{conditionalBounds}) and Definition \ref{dynamicBounds} that
\[ \Vup_k(\se, Z, \Me)= \inf_{H \in \He} \sup_{\se' \in \Swe_{(\se,k)}}  [Z(\se')-H_k(\se)~(S'_{k+1}-S_k)] = \Uup_{k}(\se, Z,\Me).\]
Let us now assume that \eqref{firstInequality} holds for $k$ and consider a node $(\se,k-1)$ with $M(\se)>k-1$. Fix $H \in \He$, for all $\se' \in \Swe_{(\se,k-1)}$ with $M(\se)=k$ we have
\begin{equation}\label{eqn:4}
  \Uup_k(\se',Z,\Me)-H_{k-1}(\se)~(S'_{k}-S_{k-1})=Z(\se')- \sum_{i=k-1}^{n-1} H_i(\se)\Delta_{i}S' \le \sup_{\se' \in \Swe_{(\se,k-1)}} [ Z(\se')- \sum_{i=k-1}^{n-1} H_i(\se')\Delta_{i}S]
\end{equation}
since $H_i(\se)=0$ for all $i \ge k$. Consider now $\se' \in \Swe_{(\se,k)}$ with $M(\se')>k$. Then, by inductive hypothesis,
\[ \Uup_{k}(\se', Z,\Me) \le \Vup_k(\se',Z,\mathcal{M}) = \inf_{H \in \He} \sup_{\se'' \in \Swe_{(\se',k)}} [Z(\se'')-\sum_{i=k}^{n-1}H_i(\se'')\Delta_{i}S'' ].\]
Therefore, for $H^* \in \He$,
\begin{eqnarray}
  \nonumber \Uup_k(\se', Z,\Me)-H^*_{k-1}(\se)\Delta_{k-1}S' &\le& -H^*_{k-1}(\se)\Delta_{k-1}S' + \inf_{H \in \He} \sup_{\se'' \in \Swe_{(\se',k)}} [Z(\se'')-\sum_{i=k}^{n-1}H_i(\se'')\Delta_{i}S'' ] \\
  \nonumber &\le& -H^*_{k-1}(\se)\Delta_{k-1}S + \sup_{\se'' \in \Swe_{(\se',k)}} [Z(\se'')-\sum_{i=k}^{n-1}H^*_i(\se'')\Delta_{i}S'' ] \\
  \nonumber &\le& \sup_{\se'' \in \Swe_{(\se',k)}} [Z(\se'')-\sum_{i=k-1}^{n-1}H^*_i(\se'')\Delta_{i}S'' ]\\
  &\le& \sup_{\se' \in \Swe_{(\se,k-1)}} [Z(\se')-\sum_{i=k-1}^{n-1}H^*_i(\se')\Delta_{i}S' ]. \label{eqn:2}
\end{eqnarray}
Finally, from \eqref{eqn:4} and \eqref{eqn:2} it follows that
\[ \sup_{\se' \in \Swe_{(\se,k)}} [\Uup_{k}(\se', Z,\Me)-H^*_{k-1}(\se)\Delta_{k-1}S] \le \sup_{\se' \in \Swe_{(\se,k)}} [Z(\se')-\sum_{i=0}^{n-1}H^*_i(\se')\Delta_{i}S' ] \]
and since $H^* \in \He$ was taken to be arbitrary \eqref{firstInequality} follows.
\qed \end{proof}

The next corollary, being a consequence of Proposition \ref{0-neutral} and Theorem \ref{useful}, represents the dynamic analogue of the $0$-neutral condition,

\begin{corollary}  \label{necessaryDynamicCondition}
Let $\Me = \Swe\times \He$ a discrete $n$-bounded market model and $Z\ge 0$ a function defined on $\Swe$. If $\Swe$ satisfies the local $0$-neutral property, then for any $\se \in \Swe$ and $0\le i \le n$:
\begin{enumerate}
\item  $\overline{U}_i(\se, Z,\Me)\ge 0$.
\item $ \underline{U}_i(\se, Z=0, \Me) =\overline{U}_i(\se, Z=0, \Me) =0$.
\end{enumerate}
\end{corollary}
\begin{proof}
For $(1)$ we proceed by induction backwards, since
\[ \overline{U}_n(\se, Z,\Me)=Z(\se)\ge 0 \mbox{ or } \overline{U}_n(\se, Z,\Me)=0, \]
this is so by definition as for any $\se \in \Swe$, $ M(\se) \leq n$. Assume $\overline{U}_{i+1}(\se, Z,\Me)\ge 0$, for some $0\le i\le n-1$ and any $\se\in \Swe$. For fixed  $\se$, if $i\ge M(\se)$ then $\overline{U}_i(\se, Z,\Me)= 0$ or $\overline{U}_i(\se, Z,\Me)=Z(\se)\ge 0$. On the other hand, if $i< M(\se)$, since $\Swe$ satisfies the local $0$-neutral property at $(\se,i)$, for any $H\in \He$ 
\[~\sup_{\se' \in \Swe_{(\se,i)}}[ - H_i(\se)\Delta_{i}S']\ge 0.\]
Thus
\[ \sup_{\se' \in \Swe_{(\se,i)}}[\overline{U}_{i+1}(\se', Z,\Me) - H_i(\se)
\Delta_iS'] \ge \sup_{\se' \in \Swe_{(\se,i)}}[ - H_i(\se)\Delta_{i}S'] \ge 0 \]
and then $\overline{U}_i(\se, Z,\Me) \ge 0$.

For statement $(2)$ assume first that $M(\se)\le i$, then $\Uup_i(\se,Z=0,\Me)=0$. For $M(\se)>i$, the equality follows from Theorem \ref{useful} and item $(1)$ above since
\[0 \le \overline{U}_i(\se, Z=0,\Me)\le \overline{V}_i(\se, Z=0,\Me)=0,\]
where the last equality follows from Proposition \ref{0-neutral}.
\qed \end{proof}

Continuing the analogy between global and dynamic bounds, we obtain an analogue of Theorem \ref{havingAnIntervalTheorem}. First we need the following lemma.

\begin{lemma} \label{subadd}
  Let $\Me=\Swe \times \He$ an $n$-bounded discrete market and assume $\He$ is closed under addition. Set $\se \in \Swe$ and $0 \le i \le n$ . Assume $Z_1$ and $Z_2$ are real valued functions defined on $\Swe$ then,
  \begin{equation}\label{dynint}
    \Uup_i(\se,Z_1+Z_2,\Me) \le \Uup_i(\se,Z_1,\Me)+\Uup_i(\se,Z_2,\Me).
  \end{equation}
\end{lemma}
\begin{proof}
  We proceed by backward induction; consider first $i=n$, if $M(\se)<n$,
  \[ 0=\Uup_i(\se,Z_1+Z_2,\Me) \le \Uup_i(\se,Z_1,\Me)+\Uup_i(\se,Z_2,\Me)=0+0=0.\]
  If $i=M(\se)$
  \[ Z_1(\se)+Z_2(\se)=\Uup_i(\se,Z_1+Z_2,\Me) \le \Uup_i(\se,Z_1,\Me)+\Uup_i(\se,Z_2,\Me)=Z_1(\se)+Z_2(\se).\]
  Assume \eqref{dynint} holds for some $0 \le i+1 \le n$ and any $\se \in \Swe$. If $i \ge M(\se)$ then, as before, we have
  \[ \Uup_i(\se,Z_1+Z_2,\Me) \le \Uup_i(\se,Z_1,\Me)+\Uup_i(\se,Z_2,\Me).\]
  Let $H^1$ and $H^2$ elements of $\He$ so, $H^1+H^2 \in \He$, then if $i<M(\se)$ we have
  \begin{eqnarray*}
    \Uup_i(\se,Z_1+Z_2,\Me) &\le& \sup_{\se' \in \Swe_{(\se,i)}}[\overline{U}_{i+1}(\se', Z_1+Z_2,\Me) - (H^1_i(\se)+H^2_i(\se))\Delta_iS'] \\
    &\le& \sup_{\se' \in \Swe_{(\se,i)}}[\overline{U}_{i+1}(\se', Z_1,\Me) - H^1_i(\se)\Delta_iS' + \overline{U}_{i+1}(\se', Z_2,\Me)- H^2_i(\se)\Delta_iS']\\
    &\le& \sup_{\se' \in \Swe_{(\se,i)}}[\overline{U}_{i+1}(\se', Z_1,\Me) - H^1_i(\se)\Delta_iS'] + \sup_{\se' \in \Swe_{(\se,i)}}[\overline{U}_{i+1}(\se', Z_2,\Me) - H^2_i(\se)\Delta_iS'].
  \end{eqnarray*}
  Therefore, since $H^1$ and $H^2$ are generic elements of $\He$, it follows that
  \[  \Uup_i(\se,Z_1+Z_2,\Me) \le \Uup_i(\se,Z_1,\Me)+\Uup_i(\se,Z_2,\Me).\]
\qed \end{proof}

\begin{theorem} \label{dyninttheo}
  Consider an $n$-bounded discrete market $\Me=\Swe \times \He$, a function $Z$ defined on $\Swe$ and $\se \in \Swe$ fixed. If $\Swe$ satisfies the local $0$-neutral property and $\He$ is closed under addition, then
  \begin{equation}\label{dynint2}
    \Udo_i(\se,Z,\Me) \le \Uup_i(\se,Z,\Me).
  \end{equation}
\end{theorem}
\begin{proof}
  By Lemma \ref{subadd} with $Z_1=Z$ and $Z_2=-Z$ and Corollary \ref{necessaryDynamicCondition} we have
  \[ 0=\Uup_i(\se,0,\Me) \le \Uup_i(\se,Z,\Me)+\Uup_i(\se,-Z,\Me).\]
  Then
  \[ \Udo_i(\se,Z,\Me)=-\Uup_i(\se,-Z,\Me) \le \Uup_i(\se,Z,\Me).\]
\qed \end{proof}

The next Corollary shows that the dynamic and global bounds coincide for an attainable $Z$.

\begin{corollary}
  Consider an $n$-bounded discrete market $\Me=\Swe \times \He$, $0 \le k < n$ fixed and $\se \in \Swe$ with $M(\se)>k$. Let $Z$ a function on $\Swe$ and assume $\Swe$ is locally $0$-neutral and $\He$ is closed under addition. If $Z$ is attainable with portfolio $H^Z \in \He$ and $-H^Z \in \He$, then
  \[ \Vdo_k(\se,Z,\Me)=\Udo_k(\se,Z,\Me)=\Uup_k(\se,Z,\Me)=\Vup_k(\se,Z,\Me).\]
\end{corollary}
\begin{proof}
  From Theorem \ref{useful} and Theorem \ref{dyninttheo}, it follows that
  \[ \Vdo_k(\se,Z,\Me)\le\Udo_k(\se,Z,\Me)\le\Uup_k(\se,Z,\Me)\le\Vup_k(\se,Z,\Me).\]
  Notice that Corollary \ref{onevalue} is applicable when $Z$ is attainable, thus
   \[ \Vdo_k(\se,Z,\Me)=\Udo_k(\se,Z,\Me)=\Uup_k(\se,Z,\Me)=\Vup_k(\se,Z,\Me).\]
\qed \end{proof}

\subsection{\bf Full Set of Portfolios}\label{sec:fullPortfolios}

We are interested in obtaining the reverse of inequality
(\ref{firstInequality}) when $Z$ is not attainable. To achieve that goal, it will be necessary to
introduce some conditions on the set of portfolios, as well as other
conditions, that imply equality in the inequality
(\ref{firstInequality}) and also lead to an efficient method to
compute the dynamic bounds. Results in \cite{bertsekas} suggest that having all possible portfolios may lead to establishing the desired equality; this motivates the definition of Full set of portfolios.

\begin{definition}
  Let $i \in \mathbb{N}$, a function $h:\Swe \rightarrow \mathbb{R}$ is said to be $i$-non-anticipative if for each $\se,\se' \in \Swe$ satisfying $i < \min\{ M(\se),M(\se')\}$ and $(S_j,W_j)=(S'_j,W'_j)$, for all $0 \le j \le i$, it then follows that  $h(\se)=h(\se')$.
\label{naf}
\end{definition}

\begin{definition}[FULL Set of Portfolios]\label{full} Given a discrete market $\Me = \Swe \times \He$,
consider  $k\ge 0$, $\se \in \Swe$, $j \geq k$ and range set,
\[ I^j_{\Swe_{(\se,k)}}\equiv \{H_j(\se'): H\in \He, \; \se' \in \Swe_{(\se,k)} \}. \]
We will say that $\He$ is {\small\emph {FULL}}, if the set of functions with domain $\Swe_{(\se,k)}$ and range $I^j_{\Swe_{(\se,k)}}$, which are $j$-non-anticipative coincides, for any such $k, \se$ and $j$, with the set of functions $H_j|_{\Swe_{(\se,k)}}: \Swe_{(\se,k)} \rightarrow \mathbb{R}$ where $H \in \mathcal{H}$.
\end{definition}


\noindent
Observe that $\Se_{(\se,k)}=\Se_{(\se',k)}$ for $\se' \in \Swe_{(\se,k)}$, which justifies the notation $I^j_{\Swe_{(\se,k)}}$. A particular, but convenient possibility, is the case when $I^j_{\Swe_{(\se,k)}}\equiv \mathbb{R}$, for any $k\ge 0$, and $\se \in \Swe$.

Theorem \ref{fullDynamicProgramming} below
shows that equality in (\ref{firstInequality}) holds
for a bounded market with a {\small\emph {FULL}}
portfolio set. The latter is natural in the sense that any of the values $H_j(\se)$, taken by a portfolio $H\in\He$ at a rebalancing instance $j$ for some $\se \in \Swe$, should also be taken at any $\se' \in \Swe_{(\se,k)}$ if $j\ge k$. This implies that there exists $H'\in\He$ such that $H'_j(\se')=H_j(\se)$. Actually, any set of portfolios $\He$ can be extended to a set $\overline{\He}$ which is {\small\emph {FULL}} as we explain next.

For $j \geq k$ and $h$ a $j$-non-anticipative function, define
\[ H_i(\se')=\left\lbrace \begin{array}{ll}
                         h(\se') & \textrm{if $\se' \in \Swe_{(\se,k)}$ and $i=j$,}\\
                         0 & \textrm{otherwise,}
                         \end{array} \right.
\]
we show next that $H$ is non-anticipative. Let $\se^1, \se^2 \in \Swe$ such that $S^1_l=S^2_l$ and $W^1_l=W^2_l$ for all $0\le l \le i$ with $i \le \min\{M(\se^1),M(\se^2)\}$. Assume first $i=j$.  It is not possible that, for example, if $\se^2 \notin \Swe_{(\se,k)}$, then $M(\se^2) \le k$ and $i \le k <j$ which is a contradiction. Then, $\se^1,\se^2 \in \Swe_{(\se,k)}$. Since $h$ is $j$-non-anticipative, it follows
\[ H_i(\se^1)=h(\se^1)=h(\se^2)=H_i(\se^2).\]
Finally, the case $i \neq j$ is trivial because $H_i(\se^1)=H_i(\se^2)=0$.


\begin{theorem}  \label{fullDynamicProgramming}
For a general $n$-bounded  market $\Me = \Swe\times
\He$, where $\He$ is {\small FULL}, and for a given
function $Z$ defined on $\Swe$, we have
\begin{equation} \label{fullDynamicProgramingEquality}
\overline{V}(S_0, Z, \Me) = \overline{U}_0(S_0, Z,\Me).
\end{equation}
\end{theorem}
\begin{proof}
Because of Theorem \ref{useful} we only need to prove the inequality,
\begin{equation}
\Vup_0(S_0, Z,\Me) \le \Uup_0(S_0, Z,\Me).
\label{secondinequality}
\end{equation}
We proceed  by induction on $n$. For $n=1$, for all $\se \in \Swe$ we have $M(\se)=1$. Then, from (\ref{conditionalBounds}) and Definition
\ref{dynamicBounds},
\[ \Vup(S_0, Z, \Me)= \inf_{H \in \He} \sup_{\se \in \Swe}  [Z(\se)-H_0(\se)~(S_{1}-S_0)] = \Uup_0(S_0, Z,\Me).\]
Let us now assume that \eqref{secondinequality} holds for every $n$-bounded discrete market model and consider an $(n+1)$-bounded one, $\Me= \Swe \times \He$. Fix $H \in \He$, and let $\se \in \Swe$ such that $M(\se)>1$. We can then apply Lemma \ref{inductionModelproof} and it follows that $\mathcal{\widehat{M}}_1$ is an $n$-bounded market and $\Uup_1(\se,Z,\Me)=\Uup_0(\hat{S}_0,\hat{Z},\mathcal{\widehat{M}}_1)$ where $\mathcal{\widehat{M}}_1$,  $\hat{Z}$, $\hat{S}_0$, are introduced in Definition \ref{inductionModel} (this definition and lemma are located in Appendix \ref{sec:technicalresults}). Then, by the inductive hypothesis,
\[ \Uup_1(\se, Z,\Me)= \Uup_0(\hat{S}_0,\hat{Z},\mathcal{\widehat{M}}_1) \ge \Vup_0(\hat{S}_0,\hat{Z},\mathcal{\widehat{M}}_1) = \inf_{H' \in \He} \sup_{\se' \in \Swe_{(\se,1)}} [Z(\se')-\sum_{i=1}^{(n+1)-1}H'_i(\se')\Delta_{i}S' ].\]
Thus, we can assume that $\Uup_0(S_0, Z,\Me)>-\infty$, and consequently for $\se \in \Swe$, $\Uup_1(\se,
Z,\Me)>-\infty$.
Fix $\varepsilon > 0$, then there exists $H^{\se} \in \He$, such that
\[\sup_{\se' \in \Swe_{(\se,1)}}[Z(\se')-\sum_{i=1}^{(n+1)-1}H^{\se}_i(\se')\Delta_iS'] < \varepsilon + \Uup_1(\se, Z,\Me).\]
Therefore,
\begin{equation}\label{firstNontrivialInequality}
-H_0(\se)\Delta_0S + \sup_{\se'\in \Swe_{(\se,1)}}[Z(\se')-\sum_{i=1}^{(n+1)-1}H^{\se}_i(\se')\Delta_iS']
< \varepsilon -H_0(\se)\Delta_0S + \Uup_1(\se,Z,\Me).
\end{equation}
Since $\He$ is {\small FULL}, there exist
$H^{\varepsilon}\in \He$ such that, $H^{\varepsilon}_0 =H_0$ and for any $\se^* \in \Swe$
\[ H^{\varepsilon}_i(\se^*) = H^{\se}_i(\se^*) \mbox{ if } \se^* \in \Swe_{(\se,1)} \mbox{ and } i\ge 1,\]
the functions $H^{\varepsilon}_i$ are well defined since the family $\{\Swe_{(\se,1)}\}_{\se \in \Swe}$ is a partition of $\Swe$.
From (\ref{firstNontrivialInequality}) it follows that
\begin{eqnarray}
Z(\se)-\sum_{i=0}^{(n+1)-1}H^{\varepsilon}_i(\se)\Delta_iS &<& \varepsilon + \sup_{\se \in \Swe} [-H_0(\se)\Delta_0S + \Uup_1(\se,Z,\Me)],
\label{eqn:7}
\end{eqnarray}
Assume now $\se \in \Swe$ with $M(\se)=1$, then
\[ \Uup_1(\se,Z,\Me)-H_0(\se)~(S_{1}-S_0)=Z(\se)- \sum_{i=0}^{(n+1)-1} H_i(\se)\Delta_{i+1}S, \]
since $H_i(\se)=0$ for all $i \ge 1$. Therefore
\begin{eqnarray}
 Z(\se)- \sum_{i=0}^{(n+1)-1} H_i(\se)\Delta_{i}S &<& \varepsilon + \sup_{\se \in \Swe} [\Uup_1(\se,Z,\Me)-H_0(\se)\Delta_0S].
 \label{eqn:6}
 \end{eqnarray}
 Finally from \eqref{eqn:7} and \eqref{eqn:6} it follows that
 \[\inf_{\tilde{H} \in \He} \sup_{\se \in \Swe} [Z(\se)-\sum_{i=0}^{(n+1)-1}\tilde{H}_i(\se)\Delta_iS] < \varepsilon + \sup_{\se \in \Swe} [\Uup_1(\se,Z,\Me)-H_0(\se)\Delta_0S], \]
 and then
 \begin{eqnarray*}
\overline{V}_0(S_0, Z,\Me) &<& \varepsilon + \inf_{H \in \He}\sup_{\se \in \Swe} [-H_0(\se)\Delta_0S + \Uup_1(\se,Z,\Me)] \\
&<& \varepsilon + \overline{U}_0(S_0, Z,\Me).
\end{eqnarray*}
Since $\varepsilon$ was taken arbitrarily,  \eqref{secondinequality} follows.
\qed \end{proof}

\subsection{\bf u-Complete Set of Portfolios}\label{sec:u-completePortfolios}

We introduce another condition that allows to derive the equality $\Uup_0(S_0,Z,\Me)=\Vup_0(S_0,Z,\Me)$. Most of the proofs and some required new notation for this section are provided in Appendix \ref{sec:u-completeProofs}.

\begin{definition}[$u$-Complete Market] We will say that an $n$-bounded discrete market $\Me$ is \emph{u-complete} with respect to a real function $Z$ defined on $\Swe$, if for any $\se \in \Swe$, and $1\le k < M(\se)$, there exists $H^*\in \He,\;$ verifying
\begin{equation}   \nonumber 
\Uup_k(\se, Z, \Me)= ~ \sup_{\se'\in\Swe_{(\se,k)}}[\Uup_{k+1}(\se', Z, \Me)-H^*_k(\se)\Delta_kS'].
\end{equation}
\end{definition}


\begin{theorem}  \label{u-completeDynamicProgramming}
If $\Me = \Swe\times \He$ is an $n$-bounded
discrete market \emph{u-complete} with respect to a given function $Z$ defined on $\Swe$, then
\begin{equation} \nonumber
\Vup(S_0, Z, \Me) = \Uup_0(S_0, Z, \Me).
\end{equation}
\end{theorem}
\begin{proof} As in the proof
of Theorem \ref{fullDynamicProgramming} the required equality for
$n=1$ is clear, we complete the proof by induction on $n$. Assume $\Me = \Swe \times \He$ is an  $(n+1)$-bounded discrete market which is u-complete, then by Lemma
\ref{u-completeInductionLemma}, item  $2$, in Appendix \ref{sec:u-completeProofs},
$\mathcal{\widetilde{M}}$ is $n$-bounded and u-complete.
Thus, resorting now to item $1$ of  Lemma \ref{u-completeInductionLemma} and the inductive
hypothesis,
\[\overline{U}_0(S_0, Z,\Me) = \overline{U}_0(S_0, \widetilde{Z}, \mathcal{\widetilde{M}}) = \overline{V}(S_0, \widetilde{Z},\mathcal{\widetilde{M}}).\]
By u-completeness, for any
$\se \in \Swe$ there exists $H^* \in \He$ such that
\[ \overline{U}_{n}(\se, Z, \Me) = \sup_{\se' \in \Swe_{(\se,n)}}\{~\overline{U}_{n+1}(\se', Z, \Me)- H^*_n(\se)\Delta_iS'\}.\]
If $M(\se)=n+1$,
\[
\widetilde{Z}(\widetilde{\se}) = \overline{U}_{n}(\se, Z, \Me)= \sup_{\se' \in
\Swe_{(\se,n)}}\{~Z(\se')- H^*_n(\se)\Delta_iS'\}\ge Z(\se)-
H^*_n(\se)\Delta_iS,\]
and if $M(\se)\le n$, $\widetilde{Z}(\widetilde{\se})=Z(\se)-H^*_n(\se)\Delta_iS$, since $H^*_n(\se)=0$. In any
case
\begin{eqnarray*}
\overline{V}(S_0, \widetilde{Z}, \mathcal{\widetilde{M}}) & = & \inf_{H \in \He}~~\sup_{\widetilde{\se} \in \widetilde{\Swe}} ~~[\widetilde{Z}(\widetilde{\se}) - \sum_{i=0}^{n-1}H_i(\widetilde{\se}) \Delta_i\widetilde{S}] \ge \\
& \ge & \inf_{H \in \He}~~\sup_{\se \in \Swe}~~[~Z(\se)- H^*_n(\se)\Delta_nS - \sum_{i=0}^{n-1}H_i(\se) \Delta_iS] \ge \\
& \ge & \inf_{H \in \He}~~\sup_{\se \in \Swe}~~[Z(\se)-
\sum_{i=0}^{n}H_i(\se) \Delta_iS] = \overline{V}(S_0, Z,\Me).
\end{eqnarray*}
The reverse inequality follows from Proposition \ref{useful}.
\qed \end{proof}

Considered together, Propositions \ref{u-completionOfMarkets}  and \ref{usefulThForUcomplete} below provide practical and useful hypothesis for an application of Theorem \ref{u-completeDynamicProgramming} above.

\begin{proposition} \label{u-completionOfMarkets}
Consider an $n$-bounded discrete market $\Me = \Swe \times \He$ and a node $(\se, k)$ with
 $0 \le k < M(\se)$. Assume one of the hypothesis below hold:
 \begin{enumerate}
   \item $I^k_{\se}$ is a compact subset of $\mathbb{R}$,
   \item $\Swe$ satisfies the up-down property (as per Definition \ref{localDefinitions}) at node $(\se,k)$ and  $I^k_{\se}=\mathbb{R}$.
 \end{enumerate}
Then, there exists $u^*\in I^k_{\se}$, verifying that
\begin{equation}\label{u-completionEquation}
\inf_{u \in I^k_{\se}}~~\sup_{\se' \in \Swe_{(\se,k)}}~~[\Uup_{k+1}(\se', Z, \Me) - u~ \Delta_kS'] =
\sup_{\se'\in \Swe_{(\se,k)}} ~~[\Uup_{k+1}(\se', Z, \Me) - u^* \Delta_kS'].
\end{equation}
Moreover, in the case $I^k_{\se}=\mathbb{R}$, there exists $R>0$
such that $|u^*| \le R$.
\end{proposition}
\begin{proof}
Define $G:\mathbb{R} \rightarrow \mathbb{R}$, by
\[G(u)= \sup_{\se' \in \Swe_{(\se,k)}} [\overline{U}_{k+1}(\se', Z, \Me) - u \Delta_kS'],\]
assuming that $\overline{U}_{k+1}(\se', Z, \Me)<\infty$ for all $\se' \in \Swe_{(\se,k)}$. Assume first that hypothesis $1$ above holds, since for any $\se' \in \Swe_{(\se,k)}$, the functions given by
$G_{\se'}(u)=\overline{U}_{k+1}(\se', Z, \Me) - u\Delta_kS'$ are affine, then its supremum $G$ is lower semicontinuous and convex. If $I^k_{\se}$ is compact, by lower semicontinuity, there exists $u^*\in I^k_{\se}$ verifying $\displaystyle G(u^*)= \inf_{u \in I^k_{\se}} G(u)$. The proof for the case when hypothesis $2$ holds is provided in Appendix \ref{sec:u-completeProofs}.
\qed \end{proof}
Notice that for the case when $M$ is a stopping time, $S \in \Se_{(S^*,k)}$ and $k< M(S^*)$ then, the left side of (\ref{u-completionEquation}) is $\Uup_k(S,Z,\Me)$.

\begin{proposition} \label{usefulThForUcomplete}
Assume $\Me =\Swe\times \He$ is an $n$-bounded discrete market and $M$ is a stopping time. Furthermore, assume that for any $\se \in \Swe$ and $0\le k< M(\se)$,
the sets $I^k_{\se}$ and $\Swe$ verify the hypothesis  of Proposition \ref{u-completionOfMarkets}. Define
\[H^*_k:\Swe \rightarrow \bigcup_{\se \in \Swe} I^k_{\se} \mbox{ by } H^*_k(\se')\equiv u^* \mbox{ for any } \se' \in \Swe_{(\se,k)},
\]
where $u^*$ is given by Proposition \ref{u-completionOfMarkets} and $k$ is such that (\ref{u-completionEquation}) holds. Also define
\begin{equation}\label{u-completionPortfolio}
H^* = (H^*_i)_{i\ge 0} \mbox{ where } H^*_i=0 \mbox{ for } i\ge M(\se), \;\; N_{H^*}(\se)=M(\se), \mbox{ and } V_{H^*}(0,s_0)=H^*_0(\se)s_0.
\end{equation}
Then, with $\He^* = \He \cup \{H^*\}$,
$\Me^*= \Swe \times \He^*$ is a \emph{u-complete} discrete market.
\end{proposition}
\begin{proof}
See Appendix \ref{sec:u-completeProofs}.
\qed \end{proof}

\section{Convex Envelope for Dynamic Minimax Bounds} \label{sec:convexEnvelope}

This section presents a rigorous method to calculate the dynamic bounds
$\Uup_i(\se,Z,\Me)$ introduced in the previous section.
In what follows, we will assume that the dynamic bounds are finite,
this, for example, follows by an application of Theorem
\ref{useful} or  under the
assumptions of Proposition \ref{finitenessCondition} in Appendix \ref{minMaxFunctionsProofs}.

We will consider an $n$-bounded discrete market $\Me =
\Swe \times \He$ (as per Definition
\ref{nBoundedDefinition}). For $\se \in \Swe$, and $0<i<M(\se)$ we are
going to give a geometric procedure, originally introduced in
\cite{BJN} for a specific example, in order to compute the dynamic bounds. For an arbitrary, but momentarily fixed, $\se' \in \Swe_{(\se,i)}$, set
\[ \ell(x)= \Uup_{i+1}(\se',Z,\Me)-u_i(S'_{i+1}-x),\]
i.e. the line in the plane, through the point
$(S'_{i+1},\Uup_{i+1}(\se',Z,\Me))$ with slope $u_i$. Thus,
\[ \Uup_{i+1}(\se',Z,\Me)-u_i(S'_{i+1}-S_{i}) \]
is the intersection of $\ell$ with the vertical straight line
$x=S_{i}$. Therefore, for each fixed $u_i \in I^{i}_{\se}$, with
some abuse of language
\[ \sup_{\se' \in \Swe_{(\se,i)}}~ \left\lbrace \Uup_{i+1}(\se',Z,\Me)- u_i(S'_{i+1}-S_i) \right\rbrace\]
is the largest of the ordinates of the points of intersection
between the straight lines $\ell$ and $x=S_i$. Then
$\Uup_i(\se,Z,\Me)$ becomes the lowest value of these largest intersections.

To complete the geometric procedure, assume for $\se \in \Swe$ and $0\le i < M(\se)$ that,
\begin{equation}\label{eqn:sets}
 \Sdo_{(\se,i)}= \left\lbrace \se' \in \Swe_{(\se,i)}: S'_{i+1} \le S_i \right\rbrace \neq \emptyset \mbox{ and } \Sup_{(\se,i)}= \left\lbrace \se' \in \Swe_{(\se,i)}: S'_{i+1} > S_i \right\rbrace \neq \emptyset.
\end{equation}
These sets are nonempty if, for example, the node $(\se,i)$ is $0$-neutral and there exist a trajectory $\se' \in \Swe_{(\se,i)}$ such that $S'_{i+1}=S_i$ or $(\se,i)$ is an up-down node. For $\se^{\textrm{up}} \in \Sup_{(\se,i)}$ and $\sedo \in
\Sdo_{(\se,i)}$ denote by $u_{(\seup,\sedo)}$ the slope of the straight line in the plane through the points
$(S_{i+1}^{\textrm{up}},\Uup_{i+1}(\seup,Z,\Me))$ and
$(S_{i+1}^{\textrm{do}},\Uup_{i+1}(\sedo,Z,\Me))$:
\[ u_{(\seup,\sedo)}= \frac{\Uup_{i+1}(\seup,Z,\Me)-\Uup_{i+1}(\sedo,Z,\Me)}{S_{i+1}^{\textrm{up}}-S^{\textrm{do}}_{i+1}}.\]
Theorem \ref{ConvexHullThm} below will show that
\begin{equation}\label{CHdynamicBound}
L_i(\se,Z,\Me) \equiv \sup_{\substack {\seup \in
\Sup_{(\se,i)} \\ \sedo \in \Sdo_{(\se,i)}}}
[ \Uup_{i+1}(\seup,Z,\Me)-u_{(\seup,\sedo)}\Delta_iS^{\textrm{up}}]= \Uup_i(\se,Z,\Me),
\end{equation}
that is, $\Uup_i(\se,Z,\Me)$, is the largest intersection of the referred lines with the vertical line $x=S_i$.

\begin{remark} \label{obs:2}
\begin{enumerate}
  \item For any $\seup \in \Sup_{(\se,i)}$ and $\sedo \in \Sdo_{(\se,i)}$
\[\Uup_{i+1}(\seup,Z,\Me) - u_{(\seup,\sedo)}\Delta_iS^{\textrm{up}} = \Uup_{i+1}(\sedo,Z,\Me) - u_{(\seup,\sedo)}\Delta_iS^{\textrm{do}}.\]
  \item The sets defined on \eqref{eqn:sets} can also be defined in an alternative way interchanging the strict inequality, namely,
  \[  \Sdo_{(\se,i)}= \left\lbrace \se' \in \Swe_{(\se,i)}: S'_{i+1} < S_i \right\rbrace, ~~\mbox{and}
~~\Sup_{(\se,i)}= \left\lbrace \se' \in \Swe_{(\se,i)}: S'_{i+1} \ge S_i \right\rbrace. \]
\end{enumerate}
\end{remark}

\begin{proposition} \label{prop:2}
Let $\Me = \Swe \times \He$ be an $n$-bounded discrete market. Then, for all $\se \in \Swe$ and $i \in \mathbb{N}$ such that the node $(\se,i)$ is a $0$-neutral node,
\begin{equation}\nonumber
L_i(\se,Z,\Me) \le \Uup_i(\se,Z,\Me). 
\end{equation}
\end{proposition}
\begin{proof}
We consider first the case $L_i(\se,Z,\Me) < \infty$. Let $\delta > 0 $, then there is $\widetilde{\se}^{\textrm{up}} \in \Sup_{(\se,i)}$ and $\widetilde{\se}^{\textrm{do}} \in \Sdo_{(\se,i)}$ such that
\[ L_i(\se,Z,\Me) \le \Uup_{i+1}(\widetilde{\se}^{\textrm{up}},Z,\Me)- u_{(\widetilde{\se}^{\textrm{up}},\widetilde{\se}^{\textrm{do}})}\Delta_i\widetilde{S}^{\textrm{up}} + \delta.\]
For $u \in I^{i}_{\se}$ such that $u \le u_{(\widetilde{\se}^{\textrm{up}},\widetilde{\se}^{\textrm{do}})}$,
\begin{equation} \nonumber
 L_i(\se,Z,\Me) \le \Uup_{i+1}(\widetilde{\se}^{\textrm{up}},Z,\Me)- u\Delta_i\widetilde{S}^{\textrm{up}}+ \delta
 \le \sup_{\se'\in\Swe_{(\se,i)}} [ \Uup_{i+1}(\se',Z,\Me)- u\Delta_iS' ] + \delta.
\end{equation}
On the other hand, if $u > u_{(\widetilde{\se}^{\textrm{up}},\widetilde{\se}^{\textrm{do}})}$,
observing that by Remark \ref{obs:2},
\begin{eqnarray*} \nonumber
L_i(\se,Z,\Me) &\le& \Uup_{i+1}(\widetilde{\se}^{\textrm{do}},Z,\Me)- u_{(\widetilde{\se}^{\textrm{up}},\widetilde{\se}^{\textrm{do}})}\Delta_i\widetilde{S}^{\textrm{do}}+ \delta
\le \Uup_{i+1}(\widetilde{\se}^{\textrm{do}},Z,\Me)- u\Delta_i\widetilde{S}^{\textrm{do}} + \delta \\
&\le& \sup_{\se'\in\Swe_{(\se,i)}} [ \Uup_{i+1}(\se',Z,\Me)- u \Delta_iS' ] + \delta.
\end{eqnarray*}
Then
\[ L_i(\se,Z,\Me) \le \sup_{\se'\in \Swe_{(\se,i)}} [ \Uup_{i+1}(\se',Z,\Me)- u\Delta_iS' ] + \delta,\]
for all $u \in I^{i}_{\se}$ and for all $\delta > 0$. Therefore
\begin{equation}\nonumber
L_i(\se,Z,\Me) \le \inf_{u \in I^{i}_{\se}} \sup_{\se' \in \Swe_{(\se,i)}} [ \Uup_{i+1}(\se',Z,\Me)-
u\Delta_iS' ] = \Uup_i(\se,Z,\Me).
\end{equation}
Assume now $L_i(\se,Z,\Me) = \infty$. Then, for an arbitrary constant $B \in \mathbb{R}$ there exist $\seup \in \Sup_{(\se,i)}$ and $\sedo \in \Sdo_{(\se,i)}$ such that
\[ B \le \Uup_{i+1}(\seup,Z,\Me)- u_{(\seup,\sedo)}\Delta_iS^{\textrm{up}}.\]
A similar reasoning as above shows that
$ \displaystyle B \le \sup_{\se'\in\Swe_{(\se,i)}}[ \Uup_{i+1}(\se',Z,\Me)- uS^{\textrm{up}} ],$
for all $u \in I^{i}_{\se}$.
Therefore $L_i(\se,Z,\Me) = \Uup_i(\se,Z,\Me)=\infty.$ 
\qed \end{proof}

The next Theorem, which requires extra assumptions, gives an easier way to solve the optimization problem for
the case $I^{i}_{\se}=\mathbb{R}$ while allowing for  a more efficient algorithm. We remark that the assumption $I^{i}_{\se}=\mathbb{R}$ is a convenient way of guaranteeing  $u_{(\seup,\sedo)}\in I^{i}_{\se}$.

\begin{theorem} \label{ConvexHullThm}
Let $\Me = \Swe \times \He$ be an $n$-bounded discrete market. If for any $\se \in \Swe$ 
$I^{i}_{\se}=\mathbb{R}$ assume at least one the two following conditions for $\se \in \Swe$ below hold,
\begin{enumerate}
\item $L_i(\se,Z,\Me) = \Uup_{i+1}(\se^{\bullet},Z,\Me)- u_{(\se^{\bullet},\se^{\circ})}\Delta_iS^{\bullet}$ for some
$\se^{\bullet}\in \Sup_{(\se,i)}$ and $\se^{\circ} \in \Sdo_{(\se,i)}$.
\item For any $\se' \in \Swe_{(\se,i)}$, $0<a \le |S'_{i+1}-S_i|\le b$ ($a$ and $b$ may depend on $\se$).
\end{enumerate}
Then,
\begin{equation}\nonumber
\Uup_i(\se,Z,\Me)= L_i(\se,Z,\Me).
\end{equation}
\end{theorem}
\begin{proof}
It is enough to prove $\Uup_i(\se,Z,\Me)\le L_i(\se,Z,\Me)$ as the reverse inequality follows immediately from Proposition \ref{prop:2}. We need only to consider the case when $L_i(\se,Z,\Me)<\infty$.  Let $\delta > 0 $, then there exist $\se^{\bullet} \in
\Sup_{(\se,i)}$ and $\se^{\circ} \in \Sdo_{(\se,i)}$ such that
\[ L_i(\se,Z,\Me) \le \Uup_{i+1}(\se^{\bullet},Z,\Me)- u_{(\se^{\bullet},\se^{\circ})}(S_{i+1}^{\bullet}-S_i)+ \delta.\]
Observe that, in case  $1$, the above equation holds for $\se^{\bullet}$ and $\se^{\circ}$ that appear in the statement in case $1$. Also, since
\[\Uup_i(\se,Z,\Me) \le \sup_{\se'\in \Swe_{(\se,i)}} [\Uup_{i+1}(\se',Z,\Me)- u_{(\se^{\bullet},\se^{\circ})}\Delta_iS'], \]
there exists $\se^{*} \in \Swe_{(\se,i)}$ such that
\[\Uup_i(\se,Z,\Me) \le \Uup_{i+1}(\se^{*},Z,\Me)- u_{(\se^{\bullet},\se^{\circ})}\Delta_iS^{*} + \delta. \]
Consider first the case when the hypothesis $1$ holds. If $\se^{*} \in \Sdo_{(\se,i)}$, one should have
\[ \Uup_{i+1}(\se^{*},Z,\Me) \le \Uup_{i+1}(\se^{\bullet},Z,\Me) - u_{(\se^{\bullet},\se^{\circ})}(S_{i+1}^{\bullet}-S_{i+1}^*),\]
otherwise $-u_{(\se^{\bullet},\se^{*})}>-u_{(\se^{\bullet},\se^{\circ})}$ which leads to the contradiction
$\Uup_{i+1}(\se^{\bullet},Z,\Me)-u_{(\se^{\bullet},\se^{*})}\Delta_iS^{\bullet}>L_i(\se,Z,\Me).$
Therefore,
\[ \Uup_i(\se,Z,\Me)-\delta \le \Uup_{i+1}(\se^{\bullet},Z,\Me)- u_{(\se^{\bullet},\se^{\circ})}(S_{i+1}^{\bullet}-S_{i+1}^*)-
u_{(\se^{\bullet},\se^{\circ})}\Delta_iS^{*} =L_i(\se,Z,\Me).\]
On the other hand, if $\se^{*} \in \Sup_{(\se,i)}$, in a similar way results
\[ \Uup_{i+1}(\se^{*},Z,\Me) \le \Uup_{i+1}(\se^{\circ},Z,\Me)-
u_{(\se^{\bullet},\se^{\circ})}(S_{i+1}^{\circ}-S_{i+1}^*).\]
Thus
\begin{eqnarray*}
\Uup_i(\se,Z,\Me)-\delta &\le& \Uup_{i+1}(\se^{\circ},Z,\Me)-u_{(\se^{\bullet},\se^{\circ})}
(S_{i+1}^{\circ}-S_{i+1}^*)-u_{(\se^{\bullet},\se^{\circ})}(S^{*}_{i+1}-S_i)= \\
&=& \Uup_{i+1}(S^{\circ},Z,\Me)-u_{(S^{\bullet},S^{\circ})}\Delta_iS^{\circ}=L_i(S,Z,\Me).
\end{eqnarray*}
Then, the proof for the case when hypothesis  $1$ applies is complete.

In the case when hypothesis $2$ holds, assume first that $S_{i+1}^{*} \le S_i$ and define $r=\frac{S_{i+1}^{\bullet}-S_{i+1}^{*}}{S_{i+1}^{\bullet}-S_i} \delta >0$.
We are going to show, by contradiction, that
\begin{equation}\label{convexHull}
\Uup_{i+1}(S^{*},Z,\Me) \le \Uup_{i+1}(S^{\bullet},Z,\Me)-
u_{(S^{\bullet},S^{\circ})}(S_{i+1}^{\bullet}-S_{i+1}^*)+ r.
\end{equation}
Towards this end assume
\begin{equation}\nonumber
\Uup_{i+1}(S^*,Z,\Me) > \Uup_{i+1}(S^{\bullet},Z,\Me)-
u_{(S^{\bullet},S^{\circ})}(S_{i+1}^{\bullet}-S_{i+1}^*)+r,
\label{eqn:92}
\end{equation}
then
\[ -u_{(S^{\bullet},S^{*})}> -u_{(S^{\bullet},S^{\circ})}+ \frac{r}{S_{i+1}^{\bullet}-S_{i+1}^{*}} \]
which leads to
\begin{eqnarray*}
\Uup_{i+1}(S^{\bullet},Z,\Me)-u_{(S^{\bullet},S^{*})}(S_{i+1}^{\bullet}-S_i)
&>& \Uup_{i+1}(S^{\bullet},Z,\Me)-u_{(S^{\bullet},S^{\circ})}(S_{i+1}^{\bullet}-S_i)+
r~\frac{S_{i+1}^{\bullet}-S_i}{S_{i+1}^{\bullet}-S_{i+1}^{*}} \\
&>& L_i(S,Z,\Me).
\end{eqnarray*}
The latter is a contradiction with the definition of $L_i$. Then, since (\ref{convexHull}) holds,
\[\Uup_i(S,Z,\Me)-\delta\le \Uup_{i+1}(S^{\bullet},Z,\Me)-u_{(S^{\bullet},S^{\circ})}(S_{i+1}^{\bullet}-S_{i+1}^*)-u_{(S^{\bullet},S^{\circ})}(S^{*}_{i+1}-S_i) + r,\]
now, since  $r \le \frac{2b}{a} \delta$, it follows
\[\Uup_i(S,Z,\Me)\le L_i(S,Z,\Me) + \delta + r \le L_i(S,Z,\Me) + \left(1+\frac{2b}a\right)\delta. \]
While if $S_{i+1}^{*} > S_i$, in a similar way results
\begin{equation}  \nonumber
\Uup_{i+1}(S^{*},Z,\Me) \le \Uup_{i+1}(S^{\circ},Z,\Me)-
u_{(S^{\bullet},S^{\circ})}(S_{i+1}^{\circ}-S_{i+1}^*)+ r
\label{eqn:44}
\end{equation}
for $r=\frac{S_{i+1}^{*}-S_{i+1}^{\circ}}{S_{i+1}^{\circ}-S_i} \delta < 0$. Since $r \le -2\delta$, it follows from \eqref{eqn:44}
\[ \Uup_i(S,Z,\Me)\le L_i(S,Z,\Me) + \delta + r \le L_i(S,Z,\Me)  -\delta.\]
Then, the proof for the case when hypothesis  $2$ applies is complete.
\qed \end{proof}

Below we obtain some simplifications that apply to arbitrage $0$-neutral nodes, towards this end, we refine Definition \ref{localDefinitions}.
\begin{definition}
Consider a discrete market  $\Me=\Swe\times\He$, and a $0$-neutral node $(\se,k)$.
\begin{enumerate}
\item We call $(\se, k)$ a \emph{positive} arbitrage node if
\[ \sup_{\se' \in \Swe_{(\se,k)}} (S'_{k+1} - S_k) > 0 \mbox{ and } \inf_{\se' \in \Swe_{(\se,k)}} (S'_{k+1} - S_k) =0. \]
\item We call $(\se, k)$ a \emph{negative} arbitrage node if
\[ \sup_{\se \in \Swe_{(\se,k)}} (S'_{k+1} - S_k) =0 \mbox{ and } \inf_{ \se' \in \Swe_{(\se,k)}} (S'_{k+1} - S_k)<0. \]
\item We call $(\se, k)$ a \emph{flat} arbitrage node if
\[ \sup_{\se \in \Swe_{(\se,k)}} (S'_{k+1} - S_k) =0 = \inf_{ \se' \in \Swe_{(\se,k)}} (S'_{k+1} - S_k). \]
\end{enumerate}
\end{definition}

\noindent
Observe that in a negative arbitrage node
\[\Sup_{(\se,k)}=\{\se' \in \Swe_{(\se,k)}:S'_{k+1}=S_k\}\equiv \Se^=_{(\se,k)} \subseteq \Swe_{(\se,k)},\]
while in a positive arbitrage node
\[\Sdo_{(\se,k)}=\{\se' \in \Swe_{(\se,k)}:S'_{k+1}=S_k\}\equiv \Se^=_{(\se,k)} \subseteq \Swe_{(\se,k)}.\]

\begin{corollary}\label{convexityOfDynamicBounds}
Let $\Me = \Swe \times \He$ be an $n$-bounded
discrete market and assume the hypothesis of Theorem
\ref{ConvexHullThm} item $1$ holds. For any node $(\se, i)$, $0\le i <M(\se)$, which is either a negative arbitrage node or a positive arbitrage node  and $\Se^=_{(\se,i)}$ is nonempty, it holds that
\begin{equation}\nonumber
\Uup_i(\se,Z,\Me)= \sup_{\se' \in \Se^=_{(\se,i)}} \Uup_{i+1}(\se',Z,\Me).
\end{equation}
\end{corollary}
\begin{proof}
Assume $(\se, i)$ is a positive arbitrage node and
$\se' \in \Se^=_{(S,i)}=\Sdo_{(\se,i)}$. It follows that
for any $\seup \in \Sup_{(\se,i)}$
\[ \Uup_{i+1}(\se',Z,\Me)=\Uup_{i+1}(\se',Z,\Me)-u_{(\seup,\se')}(S'_{i+1}-S_i)=\Uup_{i+1}(\seup,Z,\Me)-u_{(\seup,\se')}(S^{\textrm{up}}_{i+1}-S_i).\]
Thus
$L_i(\se,Z,\Me)=\sup\{\Uup_{i+1}(\se',Z,\Me):\se'\in\Se^=_{(\se,i)}\}$
and the result follows from Theorem \ref{ConvexHullThm}.

For a negative arbitrage node, the proof is similar by making use of Theorem \ref{ConvexHullThm}
with the alternative definitions for $\Sdo_{(\se,i)},~\Sup_{(\se,i)}$ (see Remark \ref{obs:2}).
\qed \end{proof}

\section{Examples: Trajectory Sets Via Another Source of Uncertainty}  \label{chartsAndTrajectories}

This section provides examples of trajectory sets defined by means of an additional source of uncertainty, denoted by $W$, besides the stock. A general class of models and a discretized version of them are developed as well as  concrete examples: the classical binomial and trinomial models and a model based on sampled quadratic variation.

\subsection{Interval Markets}

One should not develop the wrong impression that there is a small possible collection of models supported by the formalism described in the paper, on the contrary, the approach allows for trajectory sets that could be constructed from historical data, random samples from large collection of trajectories, etc. We refer to \cite{deganoI} where trajectory sets are constructed by sampling paths of continuous time martingales and to \cite{fleck} where, the so called, operational models are introduced and compared to market data.

The general principle guiding the constructions presented in this section is to isolate an observable quantity (representing a variable of interest) and proceed to define a trajectory space by imposing constraints relating the trajectories and a free variable representing this observable. In some cases, this process allows to impose natural constraints that
follow from the discrete nature of the financial transactions. In the present examples, for simplicity, $W$ is chosen to be one dimensional and in applications is meant to be associated to the values taken by  an observable quantity which unfolds along the stock {\it chart} $x(t)$. This latter quantity could unfold in  continuous time and its future values be influenced by a source of uncertainty encoded in $W$.

There is no essential result in our paper that requires $S_i \geq 0$, but, doing so makes it easier to connect with the usual models. The definition below assumes given: $w_0=0$,  $s_0$, and sets $\Sigma_i \subseteq \mathbb{R}$ and $\Omega_i \subset(0,\infty)$.

\begin{definition} \label{implicitDefinitionThroughConstraints}
We will say that a trajectory set $\Swe \subseteq \mathcal{S}^{\mathcal W}_{\infty}(s_0, w_0)$ is an \emph{interval trajectory set} if for real numbers $c>0$ and $0<d<1<u$ and a subset $Q \subseteq \cup_{i=0}^{\infty} \Omega_i$ each $\se \in \Swe$ verifies:
\begin{enumerate}  \label{generalConstraints}
\item $\displaystyle d \le \frac{S_{i+1}}{S_i} \le u$ for all $i \geq 0$,
\item $0 < W_{i+1} - W_i \leq c$ for all $0 \le i < M(\se),$
\item $W_{M({\bf S})}  \in Q$.
\end{enumerate}
For a set of portfolios $\He$, we set $\Me= \mathcal{S}^{\mathcal W} \times \He$  and call $\Me$  an \emph{interval market}.
\end{definition}
Given $\Swe$ an interval trajectory set, recall that if we have two trajectories $\se^1, \se^2 \in \mathcal{S}^{\mathcal W}_{\infty}(s_0, w_0)$ such that $(S_i, W_i) = (S_i', W_i')$ for all $i \in \mathbb{N}$, it does not follow that $M(\se)=M(\se')$. In particular, it could be the case that $W^1_{M(\se^1)} \in Q$ and $W^2_{M(\se^2)} \notin Q$ and, therefore, $\se^1 \in \Swe$ and $\se^2 \notin \Swe$.

\begin{remark} \label{directlyModelingTheStock}
We can consider the special case $d=e^{-\alpha}$ and $u=e^{\alpha}$ for an $\alpha>0$. Then, condition $1$ in the above Definition could equally be replaced by
\[ \left\vert \log \left( \frac{S_{i+1}}{S_i} \right) \right\vert \leq \alpha,\]
we return to this case later.
\end{remark}

An interval trajectory set $\mathcal{S}^{\mathcal W}$, as characterized above, does not need to be, in general, the set of {\it all} trajectories ${\bf S}$ satisfying the listed constraints in Definition \ref{generalConstraints}. Interval trajectory sets can be used to model the unfolding of a data chart $x(t_i)$ by mapping $\{(x(t_i),F(x,t_0,t_i))\}$, one index $i$ at a time (i.e. as the chart unfolds), to its closest path $\{(S_i, W_i,m)\}_{i\geq 0}$. Here $F(x,t_0,t_i)$ is an observable quantity that changes as the path unfolds; it can represent any variable of interest such as number or volume of transactions, time, quadratic variation, etc. In the context of an option contract expiring at time $T$, $S_{M({\bf S})}$ will be a possible value being taken by $x(T)$. The introduction of $W_i$ as an independent variable allows to widen the scope of applicability of the model given by Definition \ref{implicitDefinitionThroughConstraints} and it allows to incorporate arbitrage $0$-neutral nodes
(see Section \ref{arbitrageNodes}).

Specific instances of interval sets or their finite versions (that we present below) will in fact impose further constraints on admissible trajectories. Once these further specifications are established, the resulting trajectory sets are defined in a combinatorial way i.e. by allowing membership to $\mathcal{S}^{\mathcal{W}}$ to {\it all} possible $\{(S_i, W_i,m)\}$ satisfying the constraints. This way of defining trajectory sets will make it easy to check if the local properties of $0$-neutral or up-down are satisfied. For example, assume an interval model such that for all $\se \in \Swe$, $S_{i+1} \in [dS_i,uS_i]$ for all $0 \le i \le n$ and fix a node $(\se,i)$. Clearly, there exists the possibility of choosing trajectories ${\bf S}^1,{\bf S}^2 \in \Swe_{({\bf S},i)}$ such that $S^1_{i+1}>S_i$, and $S^2_{i+1}<S_i$ respectively, so any node is up-down, and in that case the market results locally arbitrage-free (see Definition \ref{localDefinitions}).

The figure below illustrates a typical step of a trajectory in an interval market.

\[ \xymatrix@R-1pc{            & u S_i \\
                               & S_{i+1}\ar@{.}[u]\\
                   S_i \ar[ur] & \ar@{.}[u]\\
                               & \ar@{.}[u]\\
                               & d S_i \ar@{.}[u] \\     }\]
The next two subsections provide concrete examples of interval markets and some of their properties. At first, we do not assume that the interval markets contain all the trajectories satisfying the constrains.

\subsubsection{Fixed Time Partition}

Consider a fixed time partition $\Pi$, that is, for the time interval $[0,T]$, we fix $\Pi:0=t_0<t_1<\dots<t_n=T$ being the only times at which a portfolio could be rebalanced. Set $\Omega_i=\{ t_i \}$ for all $1 \le i \le n$, then $W_i=t_i$ for all $1 \le i \le n$. Also, since the option expires at $t_n=T$, we need to impose $\Theta=\{n\}$. Therefore a trajectory $\se \in \Swe_{\infty}(s_0,w_0)$, under the above restriction, has the form $\se=\{ (S_i,t_i,n)\}_{i=0}^n$.

 \begin{remark}
For any portfolio set $\He$, the discrete market $\Me= \Swe \times \He$ with $\Swe \subseteq \Swe_{\infty}(s_0,w_0)$ under the above constrains is an $n$-bounded discrete market. Note that in the general formalism the trajectories are infinite sequences of real numbers, as $\Me$ is an $n$-bounded market, it is inconsequential to define the values of $\se_i$ for $i>n$.
 \end{remark}

The condition $M(\se)=n$ for all $\se \in \Swe_{\infty}(s_0,w_0)$ implies that $W_{M(\se)}=t_n=T$. Then, in order to define a subset of $\Swe_{\infty}(s_0,w_0)$ in the terms of Definition \ref{implicitDefinitionThroughConstraints}, the set $Q$ only must contain the element $T$, namely $Q=\{ T\}$. Also, we define $c=\max \{ t_{i+1} - t_i : 0 \le i \le n-1\}$, and then condition $2$ in Definition \ref{generalConstraints} holds. Finally, given $0<d<1<u$, we denote by $\mathcal{S}^{\mathcal W}(s_0, d,u,\Pi)$ a subset of $\Swe_{\infty}(s_0,w_0)$ satisfying the remaining conditions Definition  \ref{generalConstraints} for $u,d$ and the set $Q=\{ T\}$. For any portfolio set $\He$, we will call the associated market $\Me=\Swe(s_0,d,u,\Pi) \times \He$ a \emph{fixed time} interval market.

Note that if for each node $(\se,i)$ condition \eqref{eqn:sets} holds, then $\Swe(s_0,d,u,\Pi)$ is locally $0$-neutral, independently of the intermediate values between $d$ and $u$, and then, the associated market $\Me=\Swe(s_0,d,u,\Pi) \times \He$ is $0$-neutral. Therefore, by Theorem \ref{havingAnIntervalTheorem},
$[\Vdo(s_0,Z,\Me),\Vup(s_0,Z,\Me)]$ is a fair price interval for the option $Z$ and the bounds can be evaluated with the methods developed in the paper.

For the next result we need to define a particular kind of derivative in general markets. Indeed an European option defined on a trajectory set $\Swe$ will be called convex if its payoff function $Z$ is given by a convex real variable function $Z^f$ as follows: $Z(\se)=Z^f(S_{M(\se)})$ for any $\se \in \Swe$. The next proposition shows that, in an interval market, the dynamic bounds for a convex European option are convex. This result is proven in \cite{roorda}, we provide an alternative proof using Theorem \ref{ConvexHullThm}.

Notice that the parameters $u$ and $d$ appearing in the next Proposition could depend on $S_0, \ldots, S_i$.

\begin{proposition}  \label{convexPayoffCorollary}
Let $0<d<1<u$. Consider a fixed time interval market $\Me=\Swe(s_0, d,u,\Pi) \times \He$. Let $Z(\se) = Z^f(S_{M(\se)})$ be the payoff function of an European derivative.
\begin{enumerate}
  \item Assume that $Z^f$ is convex and for all $0 \le i \le n-1$ and $\se \in \Swe(s_0, d,u,\Pi)$ there exists $\se^u,\se^d \in \Swe_{(\se,i)}$ such that $S^u_{i+1}= u~S_i$ and $S^d_{i+1}= d~S_i$. Then, the dynamic bounds are convex and given by
\begin{equation}
\Uup_i(\se,Z,\Me)=\frac{1-d}{u-d}~~\Uup_{i+1}(\se^u,Z,\Me)+\frac{u-1}{u-d}~~\Uup_{i+1}(\se^d,Z,\Me).
\label{eqn:3}
\end{equation}
  \item Assume that $Z^f$ is concave and that for all $\se \in \Swe(s_0, d,u, \Pi)$ and $0 \le i < n$ condition \eqref{eqn:sets} holds and there exists $\se' \in \Swe_{(\se,i)}$ such that $S'_{i+1}=S_i$. Then, the dynamic bounds are concave and given by
\begin{equation} \nonumber
\Uup_i(\se,Z,\Me)=Z^f(S_i).
\end{equation}
\end{enumerate}
\end{proposition}
\begin{proof}
Let $\se \in \Swe(s_0, d,u,\Pi)$; since $S^u_n \ge S^{\textrm{up}}_n$ for any $\seup \in \Sup_{(\se,n-1)}$ and $Z^f$ convex,
\[ Z^f(S^u_n)-u_{(\se^u,\se^d)}(S^u_{n}-S^+_{\textrm{up}}) \ge Z^f\left( S^u_n \left( 1-\frac{S^u_n-S^{\textrm{up}}_n}{S^u_n-S^d_n}\right) + S^d_n \left( \frac{S^u_n-S^{\textrm{up}}_n}{S^u_n-S^d_n}\right)\right) = Z^f(S^{\textrm{up}}_n). \]
Similarly, since $S^u_n\ge S^{\textrm{do}}_n$ for any $\sedo \in \Sdo_{(\se,n-1)}$ and $Z$ convex, it follows
\[ Z^f(S^u_n)-u_{(S^u,S^d)}(S^u_{n}-S^{\textrm{do}}_{n}) \ge Z^f \left( S^u_n \left( 1-\frac{S^u_n-S^-_n}{S^u_n-S^d_n}\right) + S^d_n \left( \frac{S^u_n-S^{\textrm{do}}_n}{S^u_n-S^d_n}\right)\right) =Z^f(S^do_n).\]
Then by Lemma \ref{lem:1} in Appendix \ref{sec:auxiliary},
\[ Z^f(S^{\textrm{up}}_n)-u_{(\seup,\sedo)}(S^{\textrm{up}}_n-S_{n-1}) \le Z^f(S^u_n)-u_{(S^u,S^d)}(S^u_n-S_{n-1}),\]
for all $\seup \in \Sup_{(\se,n-1)}$ and $\sedo \in \Sdo_{(\se,n-1)}$. Therefore, hypothesis $1$ of Theorem \ref{ConvexHullThm} holds and so,
\begin{equation} \nonumber
\Uup_{n-1}(\se,Z,\Me)= Z^f(S^u_n)-u_{(\se^u,\se^d)}(S^u_n-S_{n-1})= \frac{1-d}{u-d}Z^f(uS_{n-1})+\frac{u-1}{u-d}Z^f(dS_{n-1}).
\end{equation}
Since the property of convexity is preserved under scaling and under taking positive linear combinations, it is seen from the above that $\Uup_{n-1}(\cdot,Z,\Me)$ is convex and only depends on the value of $S_{n-1}$. We proceed now by backward induction; let $0\le i < n$ and suppose that $\Uup_{i+1}(\cdot,Z,\Me)$ is convex and given by \eqref{eqn:3}. Then, with the same calculations that we use for the case $n-1$, we can prove that $\Uup_{i}(\se,Z,\Me)$ is convex and given by \eqref{eqn:3} for all $\se \in \Swe$. This concludes the proof of \eqref{eqn:3}.

Consider now the statement and assumptions in the case $2$ of our theorem and take $\se \in \Swe(s_0, d,u,\Pi)$. Since $Z^f$ is concave, it follows that
\[ Z^f(S^{\textrm{up}}_n)-u_{(\seup,\sedo)}(S^{\textrm{up}}_n-S_{n-1}) \le Z^f\left( S^{\textrm{up}}_n \left( 1-\frac{S^{\textrm{up}}_n-S_{n-1}}{S^{\textrm{up}}_n-S^{\textrm{do}}_n}\right) + S^{\textrm{do}}_n \left( \frac{S^{\textrm{up}}_n-S_{n-1}}{S^{\textrm{up}}_n-S^{\textrm{do}}_n}\right)\right)=Z^f(S_{n-1})\]
for all $\seup \in \Sup_{(\se,n-1)}$ and $\sedo \in \Sdo_{(\se,n-1)}$. In particular
\[Z^f(S^{\textrm{up}}_n)-u_{(\seup,\se')}(S^{\textrm{up}}_n-S_{n-1})=Z(S_{n-1}). \]
Therefore, hypothesis $1$ of Theorem \ref{ConvexHullThm} holds, and then $\Uup_{n-1}(\se,Z,\Me)=Z^f(S_{n-1})$. Furthermore, $\Uup_{n-1}(\cdot,Z,\Me)$ is concave. Finally, by backward induction we obtain the desired result.
\qed \end{proof}

The standard binomial tree model presented in \cite{cox} is a particular situation of fixed time interval market. In a typical node of this model, the value of $S_{i+1}$ can only be $uS_i$ or $dS_i$ for each $0 \le i <n$.
%

Binomial models are important because they can be used to approximate continuous time models by letting the time step tend to zero. The next Proposition shows that $\Vup(s_0,Z,\Me)$ for a binomial model coincides with the Cox-Ross-Rubinstein price of the derivative. This can be seen to be a special case of the general result (\cite[Theorem 8]{deganoI}) showing the equality of the risk neutral price  with  the price bounds of an associated trajectory based discrete market. As a complement, note that the binomial model is a complete market (\cite[Theorem 6.8]{cutland}), then by Corollary \ref{onevalue} we will have a unique fair price.
\begin{proposition}
Consider $\Me=\Swe(s_0, d,u,\Pi) \times \He$ a binomial market with parameters $u$ and $d$, where $0<d < 1 < u$. Let $Z=Z^f$ be the payoff function of a European derivative. Then, $\Vup(s_0,Z,\Me)=\Vdo(s_0,Z,\Me)$ and are given by the Cox-Ross-Rubinstein price:
   \[\Vup(\se_0,Z,\Me)=\sum_{i=0}^{n} \left(\begin{array}{c} n \\ j \end{array}\right) \left( \frac{1-d}{u-d}\right)^{j} \left( \frac{u-1}{u-d}\right)^{n-j}Z(S_0u^{j+1}d^{n-j}).\]
\end{proposition}
\begin{proof}
  We will prove it by induction over $n$. Let $n=1$, then by Proposition \ref{convexPayoffCorollary},
  \[\Vup(\se_0,Z,\Me)=\Uup(\se_0,Z,\Me)=\frac{1-d}{u-d}Z^f(uS_{n-1})+\frac{u-1}{u-d}Z^f(dS_{n-1})\]
  which is the price given by Cox-Ross-Rubinstein for a $1$-step binomial model. Suppose now that $\Vup(s_0,Z,\widetilde{\Me})$ is the Cox-Ross-Rubinstein price for all binomial $n$-bounded market $\widetilde{\Me}$ and let $\Me$ a binomial $n+1$-bounded market. It follows by Theorem \ref{fullDynamicProgramming} and Proposition \ref{convexPayoffCorollary},
   \[\Vup(\se_0,Z,\Me)=\Uup(\se_0,Z,\Me)=\frac{1-d}{u-d}~~\Uup_{1}(\se^u,Z,\Me)+\frac{u-1}{u-d}~~\Uup_{1}(\se^d,Z,\Me),\]
where $\se^u,\se^d \in \Swe_{(\se,0)}$ such that $S^u_{1}= u~S_0$ and $S^d_{1}= d~S_0$.
   Then, by Lemma \ref{inductionModelproof} and Theorem \ref{fullDynamicProgramming},
   \begin{eqnarray*}
     \Uup_{1}(\se^u,Z,\Me)&=&\Uup(\widehat{\se^u}_0,Z,\widehat{\Me})=\Vup(\widehat{\se^u_0},Z,\widehat{\Me})\\
     \Uup_{1}(\se^d,Z,\Me)&=&\Uup(\widehat{\se^d}_0,Z,\widehat{\Me})=\Vup(\widehat{\se^d_0},Z,\widehat{\Me})
   \end{eqnarray*}
   where $\widehat{\Me}$ is a binomial $n$-bounded market, and $\widehat{S^u_0}=S^u_1$ and $\widehat{S^d_0}=S^d_1$. Then, by inductive hypothesis,
   \begin{eqnarray*}
     \Vup(\widehat{\se^u_0},Z,\widehat{\Me})&=& \sum_{i=0}^{n} \left(\begin{array}{c} n \\ j \end{array}\right) \left( \frac{1-d}{u-d}\right)^{j+1} \left( \frac{u-1}{u-d}\right)^{n-j}Z(S_1u^{j+1}d^{n-j})\\
     \Vup(\widehat{\se^d_0},Z,\widehat{\Me})&=&\sum_{i=0}^{n} \left(\begin{array}{c} n \\ j \end{array}\right) \left( \frac{1-d}{u-d}\right)^j \left( \frac{u-1}{u-d}\right)^{n+1-j}Z(S_1u^jd^{n+1-j}).
   \end{eqnarray*}
   Finally, replacing and changing variable, we obtain
   \[\Vup(\se_0,Z,\Me)=\sum_{i=0}^{n+1} \left(\begin{array}{c} n+1 \\ j \end{array}\right) \left( \frac{1-d}{u-d}\right)^{j} \left( \frac{u-1}{u-d}\right)^{n+1-j}Z(S_0u^{j+1}d^{n+1-j}),\]
   which is the Cox-Ross-Rubinstein for a $n+1$-step binomial model.
\qed \end{proof}

The trinomial tree model was originally presented in \cite{boyle} and offers more flexibility than binomial trees. The stock price can move up, down or can also take an intermediate price between $uS_i$ and $dS_i$ at each node, as shown in the diagram below.

\[ \xymatrix@R-1pc{                           & uS_i\\
                   S_i \ar[ur] \ar[dr] \ar[r] & bS_i \\
                                              & dS_i\\    }\]

Hence, $0<d < b < u$, and it is not necessary that $b = 1$. Such market model is incomplete and, then, the technique of determining the value of an option via a replicating portfolio does not work. We can however find upper and lower bounds for the option values.

The next Theorem characterizes the minmax bounds $\Vup(s_0,Z,\Me)$ and $\Vdo(s_0,Z,\Me)$ for general incomplete fixed time interval markets $\Me=\Swe(s_0, d,u,\Pi) \times \He$. It shows that the bounds are completely determined for an European convex payoff $Z$. The result can also be found in \cite{kolokoltsov} and \cite{roorda}.

\begin{theorem}
  Consider $\Me=\Swe(S_0, u,d,\Pi) \times \He$ a fixed time interval market where $0<d < 1 < u$. Let $Z=Z^f$ be the payoff function of an European derivative and assume it is convex.
  \begin{enumerate}
    \item If for all $0 \le i \le n-1$ and $\se \in \Swe(s_0, d,u,\Pi)$ there exists $\se^u,\se^d \in \Swe_{(\se,i)}$ such that $S^u_{i+1}= u~S_i$ and $S^d_{i+1}= d~S_i$, then $\Vup(S_0,Z,\Me)$ are given by the Cox-Ross-Rubinstein price of the derivative in the binomial tree model with the same parameters as the interval model.
    \item If for all $\se \in \Swe(s_0, d,u, \Pi)$ and $0 \le i < n$ condition \eqref{eqn:sets} holds and there exists $\se' \in \Swe_{(\se,i)}$ such that $S'_{i+1}=S_i$, then $\Vdo(S_0,Z,\Me)=Z^f(S_0)$.
  \end{enumerate}
  \label{coxtheo}
\end{theorem}
\begin{proof}
  For a proof of $(1)$ see \cite[Theorem 1]{roorda}. For $(2)$, recall that $\Vdo(s_0,Z,\Me)=-\Vup(s_0,-Z,\Me)$. Then, since $Z$ is convex, $-Z$ is concave. Thus, by Proposition \ref{convexPayoffCorollary}, $\Uup(S_0,-Z,\Me)=-Z(S_0)$. Then,
  \[\Vdo(S_0,Z,\Me)=-\Vup(S_0,-Z,\Me)=-\Uup(S_0,-Z,\Me)=Z(S_0).\]
\qed \end{proof}

The Theorem assumes that the constant trajectory belongs to $\Swe(s_0, d,u,\Pi)$, namely, by Proposition \ref{convexPayoffCorollary}, item $2$, for each node $(\se,i)$, there exists a trajectory $\se' \in \Swe_{(\se,i)}$ such that $S'_{i+1}=S_i$. If this condition does not hold, then part $2$ of the above Theorem is not true. For example, if we consider a trinomial market with $b \neq 1$, it is easy to see that $\Udo_i(\se,Z,\Me)=\frac{1-c}{u-c}Z^f(uS_{i})+\frac{u-1}{u-c}Z^f(cS_{i})$ if $c<1$ and, clearly, $\Vdo(s_0,Z,\Me)$ tends to $\Vup(s_0,z,\Me)$ when $b$ tends to $d$.

\subsubsection{Sampled Quadratic Variation (SQV)}  \label{quadraticVariationCase}

This section introduces a discrete market model where $S_i$ is intended to model $e^{x(t)}$ with $x(t)$ the chart stock and $W$ represents the sampled quadratic variation of the trajectories, that is
\begin{equation}  \label{bJNQV}
W_i = \sum_{k=0}^{i-1} (\log S_{k+1} - \log S_k)^2.
\end{equation}
Notice that when using $S_i = e^{x(t)}$ one should use the word charts for the data $e^{x(t)}$, instead of $x(t)$ as we do, but we allow ourselves some abuse of language at this point.

Set for all $i \in \mathbb{N}$,
\[ \Sigma_i=\{e^x, x \in \mathbb{R}_+ \} \mbox{ and } \Omega_i=\{\sum_{k=0}^{i-1} (\log S_{k+1} - \log S_k)^2, S_k \in \Sigma_k \}.\]
We consider $m \in \Theta \equiv \mathbb{N}$. Therefore a trajectory $\se \in \Swe_{\infty}(s_0,w_0)$ has the form $\se=\{ (S_i,W_i,m)\}_{i \in \mathbb{N}}$, where $W_i$ is given by \eqref{bJNQV}.

The general constraints defining interval models in Definition \ref{implicitDefinitionThroughConstraints} can, in the present case, be interpreted as imposing constraints on the consumed quadratic variation and on the absolute value of the change on chart values, both in between consecutive trading instances.
Let $\alpha>0$, as indicated in Remark \ref{directlyModelingTheStock}
, we can restrict \begin{equation} \label{eqn:alpha}
\left| \log \frac{S_{i+1}}{S_i} \right| \le \alpha.
\end{equation}
The condition $W_{M({\bf S})}  \in Q$ means we deal with trajectories whose total sampled quadratic variation in the interval $[0, T]$ belongs to the a-priori given subset $Q$ and taking $c \equiv \alpha^2$, the constraint $W_{i+1} - W_i \leq c$ in Definition \ref{implicitDefinitionThroughConstraints} holds. We denote by $\Swe(s_0,\alpha, Q)$ a subset of $\Swe_{\infty}(s_0,w_0)$ satisfying the condition of Definition \ref{implicitDefinitionThroughConstraints} for $d=e^{-\alpha}$, $u=e^{\alpha}$, $c=\alpha^2$ and $Q$. For any portfolio set $\He$, we will call the associated market $\Me=\Swe(s_0,\alpha, Q) \times \He$ as \emph{sampled quadratic variation} interval market (or SQV market for short).

A typical node is shown below,
\[ \xymatrix@R-2pc{ & & (S_ie^{c_2},W_{i}+c_2^2)& \\
                    & & &\\
                    & & & \\
                    & (S_ie^{c_1},W^1_{i}+c_1^2) &&  \\
             (S_i,W_i) \ar[ur] \ar[dr] \ar@/^/[uuuurr] \ar@/_/[ddddrr] &&& c_i \le \alpha \\
                    & (S_ie^{-c_1},W^1_{i}+c_1^2) & &   \\
                    & & &\\
                    & & &\\
                    & & (S_ie^{-c_2},W^2_{i}+c_2^2)&    } \]
The trajectory set introduced in \cite{BJN} can be recovered as a special case of Definition \ref{implicitDefinitionThroughConstraints} by taking $Q = \{v_0\}$.

In the next section we will study how to evaluate the interval price for a finite version of intervals markets, in particular SQV markets. We will consider a finite set $Q$, which does not necessary contain a unique element. We present next an appropriate discretization for this kind of trajectories, as well as a grid data structure which will allow us to calculate the dynamic bounds for these examples.

\section{Discretization and Grid Data Structure}  \label{gridDataStructure}

\subsection{Finite Interval Markets} \label{discretization}

A natural finite discretization leading to an implementation of interval markets defined in Section \ref{chartsAndTrajectories} is obtained by introducing real numbers $\delta, \beta >0$ and natural numbers $N_1, N_2 \in \mathbb{N}$. We assume in this section that the coordinates $S_i$ are associated to chart values by $e^{x(t_i)} \rightarrow S_i$, using the exponential function makes it easier to connect with the
usual geometric stochastic models. Then, $S_i$ and $W_i$ are restricted to belong to the sets
\begin{eqnarray} \label{eqn:discqv}
\nonumber \Sigma_i \equiv \Sigma(\delta, N_1) &=& \{s_0e^{k \delta}, k \in \{-N_1, -N_1 +1, \ldots, N_1\}\}\\
\Omega_i \equiv \Omega(\beta, N_2) &=& \{j \beta^2, j\in \{ 0, \dots ,N_2 \}\}.
\end{eqnarray}
The parameters $\delta$ and $\beta$ provide a natural discretization of the chart exponentials.

\begin{remark}
If the variable $W_i$ is directly related to the samples $S_i$, for instance, in a SQV market from Section \ref{quadraticVariationCase}, it is natural to have a unique discretization parameter $\delta$ for $\Sigma_i$ and $\Omega_i$. On the other hand, if the sets $\Omega_i$ are discrete apriori, there is no need to implement a discretization. This is the case, for example, of a fixed time interval market where $\Omega_i$ has a unique element.
\end{remark}

Note that for any trajectory $\se=\{(S_i,W_i,m)\}_{i \in \mathbb{N}}$, in an interval market always holds that $w_0<W_1<\dots < W_{m}$. Therefore, if there exists $k \in \mathbb{N}$ such that $W_k=N_2\beta^2$, $k$ must be equal to $m$. Then, a trajectory $\se \in \Swe_{\infty}(s_0,w_0)$ with $\Sigma_i$ and $\Omega_i$ defined by \eqref{eqn:discqv} necessarily have $M(\se) \le N_2$. Therefore, the coordinate $m$ are restricted to belong to the set
\begin{equation} \label{eqn:theta}
\Theta=\{ 1, \dots ,N_2 \},
\end{equation}
and so, by Definition \ref{nBoundedDefinition}, the corresponding markets will be $N_2$-bounded.

In order to define a subset of $\Swe_{\infty}(s_0,w_0)$ satisfying the properties listed in Definition \ref{implicitDefinitionThroughConstraints}, let $\Lambda=\{n_1,\ldots, n_{\theta}\} \subseteq \Theta$ be a collection of positive integers and define $Q_{\Lambda}=\{n_1\beta^2,\dots,n_{\theta}\beta^2 \}$. Without loss of generality, we can assume that $n_{\theta}=N_2$. For positive integers $p$ and $q$, we denote by $\Swe(s_0,\delta,\beta,p,q,N_1,\Lambda)$ a subset of $\Swe_{\infty}(s_0,w_0)$ with $\Sigma$, $\Omega$ and $\Theta$ defined by \eqref{eqn:discqv} and \eqref{eqn:theta} satisfying the conditions of Definition \ref{implicitDefinitionThroughConstraints} (in the terms of Remark \ref{directlyModelingTheStock}) for $\alpha=p\delta$, $c=q\beta^2$ and $Q_{\Lambda}$. We will refer to this class of trajectories as \emph{finite interval trajectory sets} and as \emph{finite interval markets} for the associated markets
\[ \Me=\Swe(s_0,\delta,\beta,p,q,N_1,\Lambda) \times \He \]
where $\He$ is a portfolio set. It is clear that finite trajectory sets will have finite cardinality.

The parameters $N_1$ and $N_2$ play a key role in the local behavior of a finite discrete market. Assume the trajectory ${\bf S}=\{(S_i,W_i,N_2) \}_{i \in \mathbb{N}}$ belongs to a finite trajectory set $\Se^{\mathcal{W}}(s_0, \delta,\beta, p,q,N_1,\Lambda)$. Taking into account the constraint
\[p\delta=\alpha \ge |\log S_{i+1}-\log S_i|=|k_{i+1}-k_i|\delta, \]
the largest value that $S_{N_2}$ can attain corresponds to the value $S_{N_2}=s_0e^{N_2\,p\delta}$. Then, in order to allow for this kind of trajectory, we must take $N_1\le p~N_2$. In the case that $N_1\le (N_2-1)\,p$, there could exist trajectories with arbitrage nodes, in the sense of Definition \ref{localDefinitions}. For example, assume the trajectory $\se$ defined by
\[ \se_i=\left\lbrace \begin{array}{cc}
                      (s_0e^{ip\delta},i\beta^2,N_2) & \mbox{if } i \le \frac{N_1}{p} \\
                      (s_0e^{N_1\delta},i\beta^2,N_2) & \mbox{if } i > \frac{N_1}{p}.
                    \end{array} \right. \]
belongs to $\Se^{\mathcal{W}}(s_0, \delta,\beta p,q,N_1,\Lambda)$ with $N_1=(N_2-1)p$, it satisfies $\frac{N_1}{p}=N_2-1<N_2$ and so, one more step is available. Then, for any trajectory ${\bf S}'\in \Swe_{({\bf S},N_2-1)}$ it follows that
\[ S'_{N_2} \le S'_{N_2-1}=S_{N_2-1}=s_0e^{N_1\delta},\]
and, therefore, $({\bf S},N_2-1)$ is an arbitrage node.

Figure \ref{fig:25} displays random trajectories in a finite trajectory set $\Se^{\mathcal{W}}(s_0, \delta, \beta, p, q, N_1,\Lambda)$ with $s_0=1$, $w_0=0$, $\beta=\delta=0.0058$, $p=3$, $q=9$, $\Lambda= \{100, 200\}$, $N_1=300$ and $N_2=200$. It shows  $100$ random trajectories in each display. The first one corresponds to trajectories with $W_{M(\se)} = 0.0034= 100 ~\delta^2$, then they must have $M  \leq 100$; while the second one corresponds to trajectories with $W_{M(\se)}= 0.0067 = 200~\delta^2$ and then, they must have $M  \leq 200$. The trajectories are shown in different displays for convenience but they belong to the same trajectory set $\Se^{\mathcal{W}}(s_0, \delta, \beta ,p,q, N_1, \Lambda)$.

\begin{figure}
\centering
\includegraphics[width=0.8\textwidth]{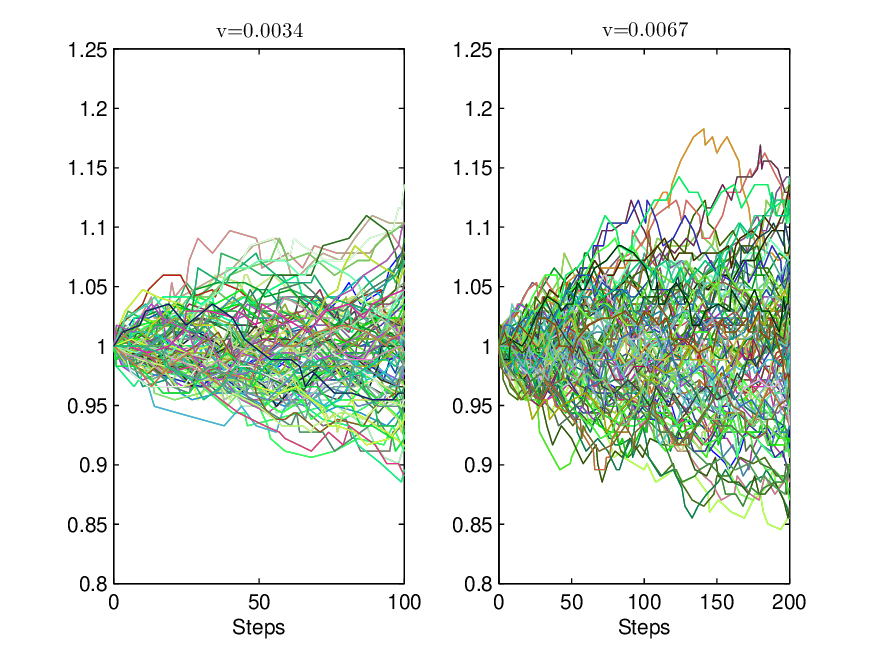}
\caption{Trajectory sets with different quadratic variation for $s_0=1$, $w_0=0$, $\delta=0.0058$, $N_1=300$, $N_2=200$ and p=3.}
\label{fig:25}
\end{figure}

We refer to Appendix \ref{gridCharacterization} for a description of a data structure and an algorithm implementing finite interval markets.

\subsection{Numerical Results}
\label{numerical}

This section provides numerical results illustrating some characteristics of the model described in Section \ref{quadraticVariationCase}.  We compute the minmax option bound prices using the finite models from Section \ref{discretization} and data structure and algorithm from Appendix \ref{gridCharacterization}. The output illustrates the super-replication price for call options with respect to the maximum number of steps and different jump sizes $p$ and its variation on the presence of arbitrage nodes. Finally, some superhedging and underhedging approximations and the effect of variable volatility are presented. For reasons of space we do not provide details related to the software implementation. Other numerical results, for a different class of models, and based on market, data can be found in \cite{fleck}.

Consider a two-month European call option with strike of \$1 on a stock that pays no dividends, with current price \$1 and the volatility of the stock is taken to be equal to $\sigma=20$\%. The Black \& Scholes price for this assumptions is \$0.0326 when $s_0 = K = 1$. Define
\[ v_0=\sigma^2 \cdot T=0.04 \cdot \frac{2}{12}=0.0067.\]
We build a sampled quadratic variation trajectory set by taking $Q = \{v_0\}$ and defining $W_i$ by \eqref{bJNQV}. Recall that $N_2$ is the maximum number of steps for a trajectory in the model, therefore

\begin{equation} \nonumber
N_2\beta^2=v_0=0.0067
\end{equation}

Then, for a given $N_2 \in \mathbb{N}$, we have a unique value for $\beta$. Since $W$ is defined in term of $S$, we only need a unique parameter in order to build a finite version of a SQV market. Then, we assume in the following $\delta=\beta$ and, in consecuence, $q=p^2$. Thus, for $p \in \mathbb{N}$, $\Lambda=\{ N_2\}$ and $N_1=pN_2$, we will consider the finite SQV trajectory set $\Swe(s_0, \alpha, Q)$, where $\alpha=p \delta$ and $Q=\{ N_2 \delta^2\}$. In this part we will consider all the trajectories in the sets \eqref{eqn:discqv} satisfying conditions \eqref{eqn:alpha}.

Let $\Me=\Swe(s_0, \alpha, Q) \times \He$ the associated market for a full set of portfolios $\He$. Figure \ref{fig:3} shows the convergence behavior of $\Vup(s_0,Z,\Me)$ and $\Vdo(s_0,Z,\Me)$ when $p=2,3,5$ with $N_2$ ranging from $10$ to $200$. When the jump unit, $p$, is greater than one, clearly,  the interval price range is more narrow as $N_2$ increases and the Black \& Scholes price belongs to the interval. Also, we can see that the interval becomes wider as $p$ increases. The reason for this is that if $p<p'$ are two jump sizes then the set of trajectories with with parameter $p$ is included in the set of trajectories with parameter $p'$. Therefore, when we calculate the bounds, the maximum over the set with parameter $p'$ is higher than the maximum over the set with parameter $p$.

\begin{figure}
\begin{center}
\includegraphics[width=0.6\textwidth]{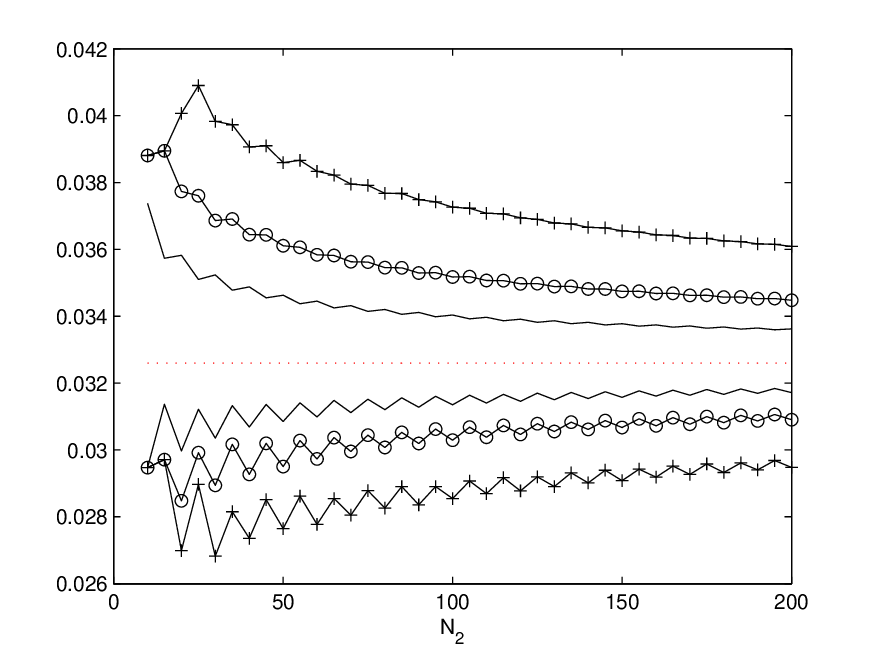}
\caption{Convergence of $\Uup(0,0)$ and $\Udo(0,0)$ as a function $N_2$ for different values of $p$.}
\label{fig:3}
\end{center}
\end{figure}

Notice a detail, when $N_2 = 5$ and $N_2=10$, in the case of the jump units are 3 and 5, the upper bounds are equal. When $N_2=10$ the maximum jump that the algorithm can take is $\sqrt{N_2}\approx 3$. Therefore, although we can run the program for the jump unit 5, this jump does not really take into account and thus does not affect the price of the option in the algorithm. Similarly for the lower bounds.

Now we fix $N_2 = 100$ and we will calculate the interval price for different starting levels of the stock $s_0$. Let $\Me=\Swe(s_0, \alpha, Q) \times \He$ the associated market for a full set of portfolios $\He$. Figure \ref{fig:5} displays $\Vup(s_0,Z,\Me)$ and $\Vdo(s_0,Z,\Me)$ for different jump units $p=1,3,5,7$. We can see that the price increase when $s_0$ increases. We notice that the minmax bounds are very narrow for the higher starting level $s_0$. Therefore, jumps have less of an effect on the bound of the option prices for higher values of the stock.

\begin{figure}
\begin{center}
\includegraphics[width=0.6\textwidth]{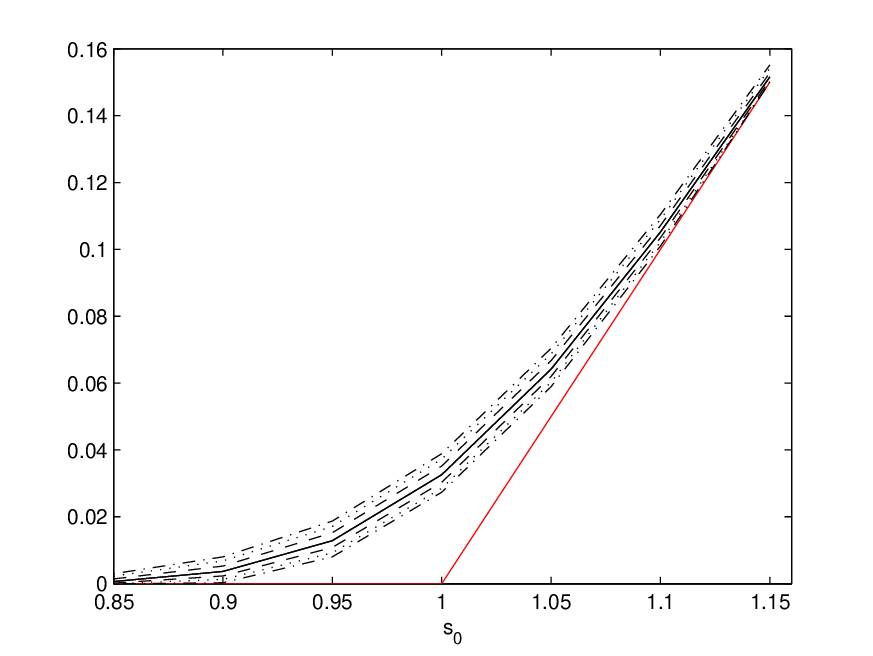}
\caption{Minmax upper and lower bounds price as a function of $s_0$ for different values of $p$.}
\label{fig:5}
\end{center}
\end{figure}

\subsubsection{Effect of arbitrage nodes on minmax bounds}  \label{arbitrageNodes}

It is interesting to see the effect of arbitrage nodes (in the sense of Definition \ref{localDefinitions}) on the model proposed above. We consider again the finite SQV trajectory set $\Swe(s_0, \alpha,Q)$ with the same parameters as above, but now the coordinates $W_i$ are not defined by \eqref{bJNQV}. Namely, $W_i$ does not depend on the stock values. Such trajectory set is now modified in order to incorporate arbitrage nodes: let $\Gamma$ the trajectory grid corresponding to $\Swe(s_0, \alpha,Q)$ (as per Section \ref{gridCharacterization}). Nodes $(k,j)$ are selected randomly and we change its reachable nodes $(k',j')$ as follows:
\begin{itemize}
\item If $k\ge 0$, the reachable nodes are $(k',j')$ where
\[-p \le k'-k \le 0 \quad \mbox{and} \quad 0<j'-j \leq p^2. \]
\item If $k<0$, the reachable nodes are $(k',j')$ where
\[ 0 \le k'-k \le p  \quad \mbox{and} \quad 0<j'-j \leq p^2.\]
\end{itemize}
These definitions give new trajectory sets which we denote by $\Swe_{\texttt{arb}}(s_0, \alpha , Q)$, where {\tt arb} refers to {\it arbitrage}. Observe that the modified trajectory set has trajectories with $S_{i+1}=S_i$ passing through an arbitrage node.

Figure \ref{fig:19} displays the upper and lower bound as a function of $s_0$ for the market $\Me=\Swe \times \He$ and the modified $\Me_{\texttt{arb}}=\Swe_{\texttt{arb}} \times \He$ with $p=1$ adding different percentages of arbitrage nodes. In the same way, Figure \ref{fig:7} displays the upper and lower bound as a function of $s_0$ with $p=3$.

\begin{figure}
\centering
\includegraphics[width=0.9\textwidth]{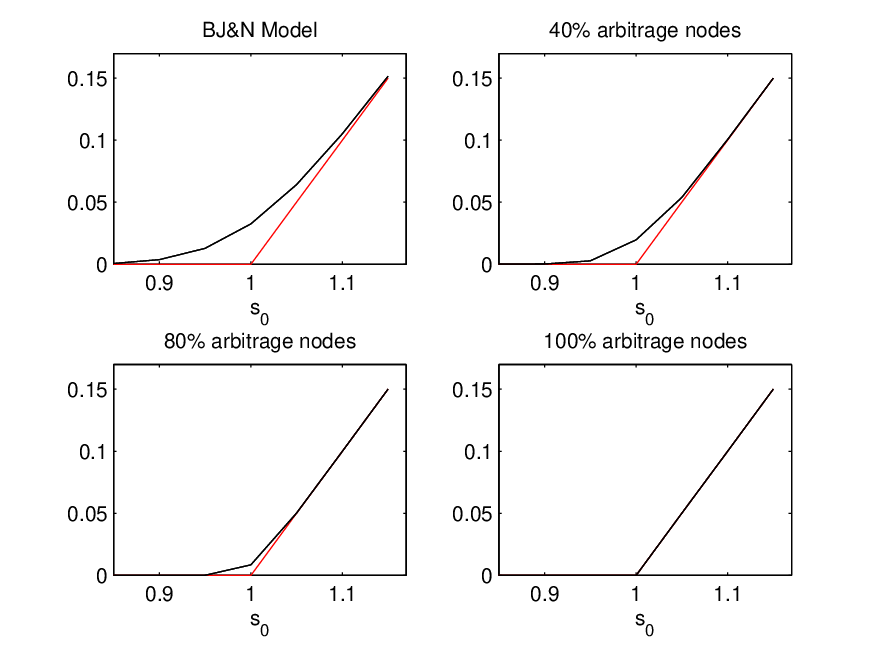}
\caption{Minmax upper and lower bounds price (which are the same in this case) as a function of $s_0$ for $p=1$ in the presence of arbitrage nodes in relation to the lower Merton bound (red line)}
\label{fig:19}
\end{figure}

\begin{figure}
\centering
\includegraphics[width=0.9\textwidth]{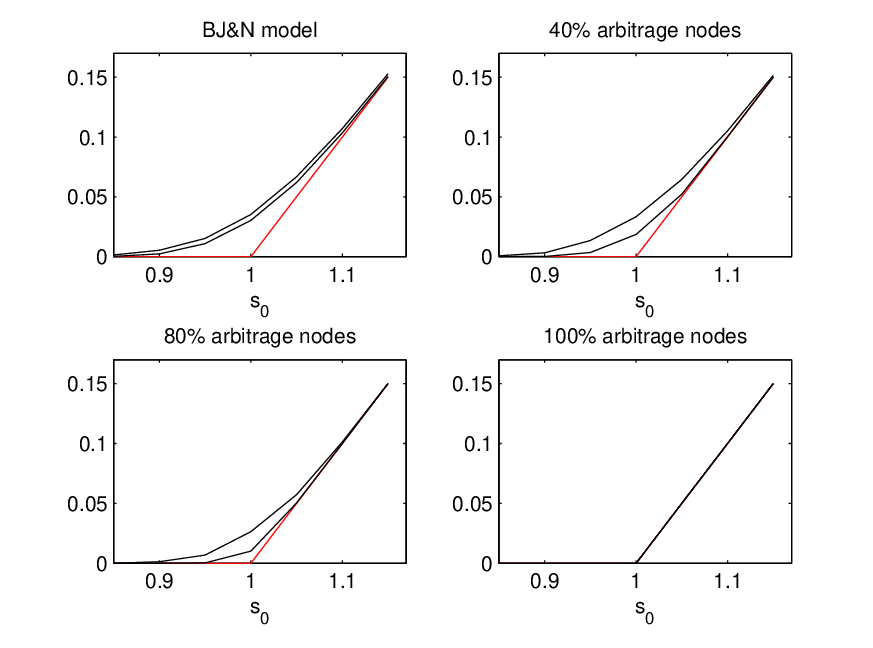}
\caption{Minmax upper and lower bounds price (which are the same in this case) as a function of $s_0$ for $p=3$ in the presence of arbitrage nodes
in relation to the lower Merton bound (red line).}
\label{fig:7}
\end{figure}

\subsubsection{Hedging}

The algorithm presented allow us to calculate not only the value of $\Vup(s_0,Z,\Me)$ but also the optimal portfolio $H$ providing the investments along each possible trajectory in $\Swe$. On each vertex $(k,j)$ of the data grid $\Gamma$ given in Section \ref{sec:computationGrid}, the dynamic upper bound $\Uup(k,j)$ is available and corresponds to an optimal value $u(k,j)=\Delta^{\pm}$ given by equation (\ref{+-reachable}). Recall that $\Uup(k,j)$ and so $u(k,j)$ give a unique value for any trajectory passing through that vertex. Therefore, we can define an optimal strategy $\{H^{\uparrow}_i\}_{i \in \mathbb{N}}$ on $\se \in \Swe$ by:
\begin{eqnarray*}
H^{\uparrow}_i(\se)=u(k,j)\quad \mbox{ if } \quad (S_i,W_i)=(s_k,w_j).
\end{eqnarray*}
This optimal strategy is non-anticipative.

It is interesting to study how $H^{\uparrow}$ actually approximates $Z$, as function of a trajectory $\se \in \Swe$, for an initial portfolio value $X$. In a short position the hedging values are given by
\begin{equation}  \label{shorthedge}
X + \sum_{i=0}^{N_{H^{\uparrow}}(\se)-1} H^{\uparrow}_i(\se) \Delta_iS
\end{equation}
with $X \in \mathbb{R}$ the initial portfolio value.

Figure \ref{fig:11}  shows the hedging values \eqref{shorthedge} with $ X=\Vup(s_0,Z,\Me)+0.01$ and $X=\Vup(s_0,Z,\Me)-0.03$, for random trajectories with $s_0=1$, $p=3$ and $N_2=100$ with respect to an European call $Z$ for the model $\mathcal{M}=\Swe(s_0, \alpha,Q) \times \He$ studied at the begin of the Subsection. We can see that the values from \eqref{shorthedge} superhedge the payoff value in the first case. For the case $X=\Vup(s_0,Z,\Me)-0.03$, the values tightly approximate the payoff values.

\begin{figure}
\centering
\includegraphics[width=0.9\textwidth]{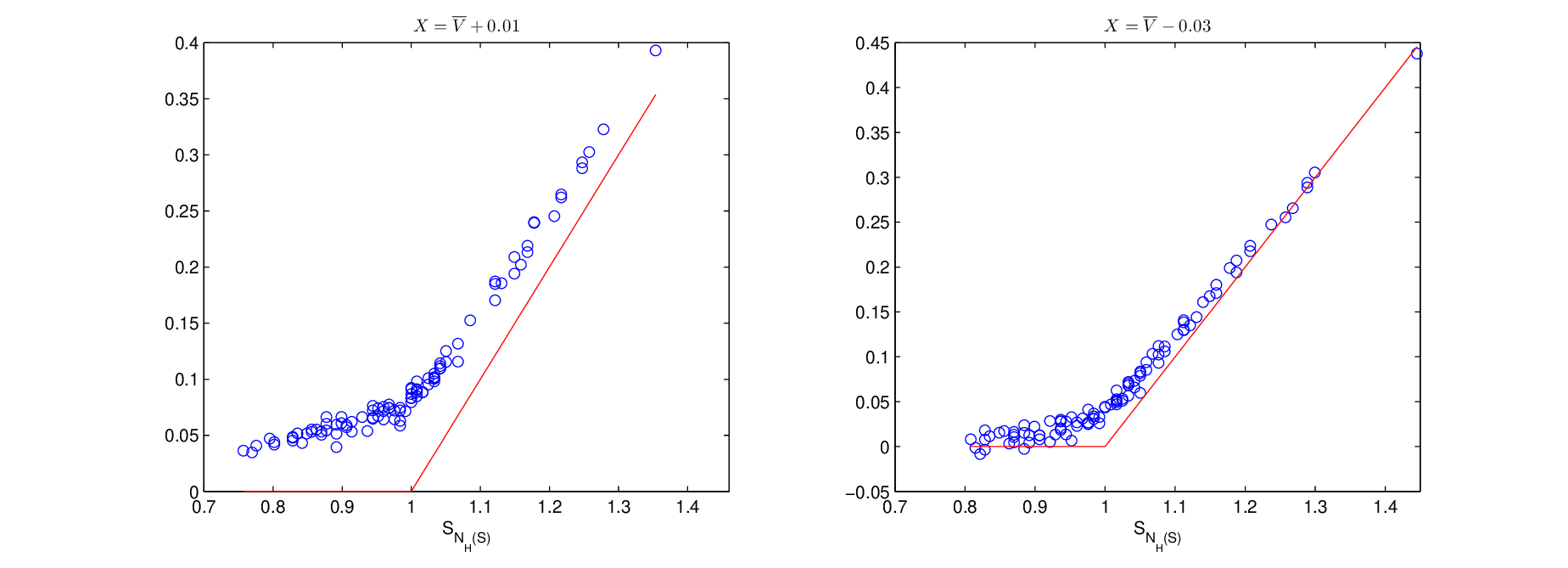}
\caption{Comparison between the hedging values for $X=\Vup(s_0,Z,\Me)+0.01$ and $X=\Vup(s_0,Z,\Me)-0.03$ and the payoff values.}
\label{fig:11}
\end{figure}

In a long position, the hedging values are given by
\begin{equation} \label{longhedge}
X - \sum_{i=0}^{N_{H^{\downarrow}}(\se)-1} H^{\downarrow}_i(\se) \Delta_iS,
\end{equation}
with $X \in \mathbb{R}$ the initial portfolio value. The underhedging portfolio $H^{\downarrow}$ is computed in a similar way than the upperhedging one $H^{\uparrow}$, but using the values which gives the lower bounds $\Udo(k,j)$ instead of the upper bounds. Figure \ref{fig:12} displays the values from equation \eqref{longhedge} with $X=\Vdo(s_0,Z,\Me)-0.01$ and $X=\Vdo(s_0,Z,\Me)+0.03$, for random trajectories with respect an European call $Z$. In this case, we can see that the values from \eqref{longhedge} underhedge the payoff value for $X=\Vdo(s_0,Z,\Me)-0.01$. For $X=\Vdo(s_0,Z,\Me)+0.03$, the values better approximate the payoff.

\begin{figure}
\centering
\includegraphics[width=0.9\textwidth]{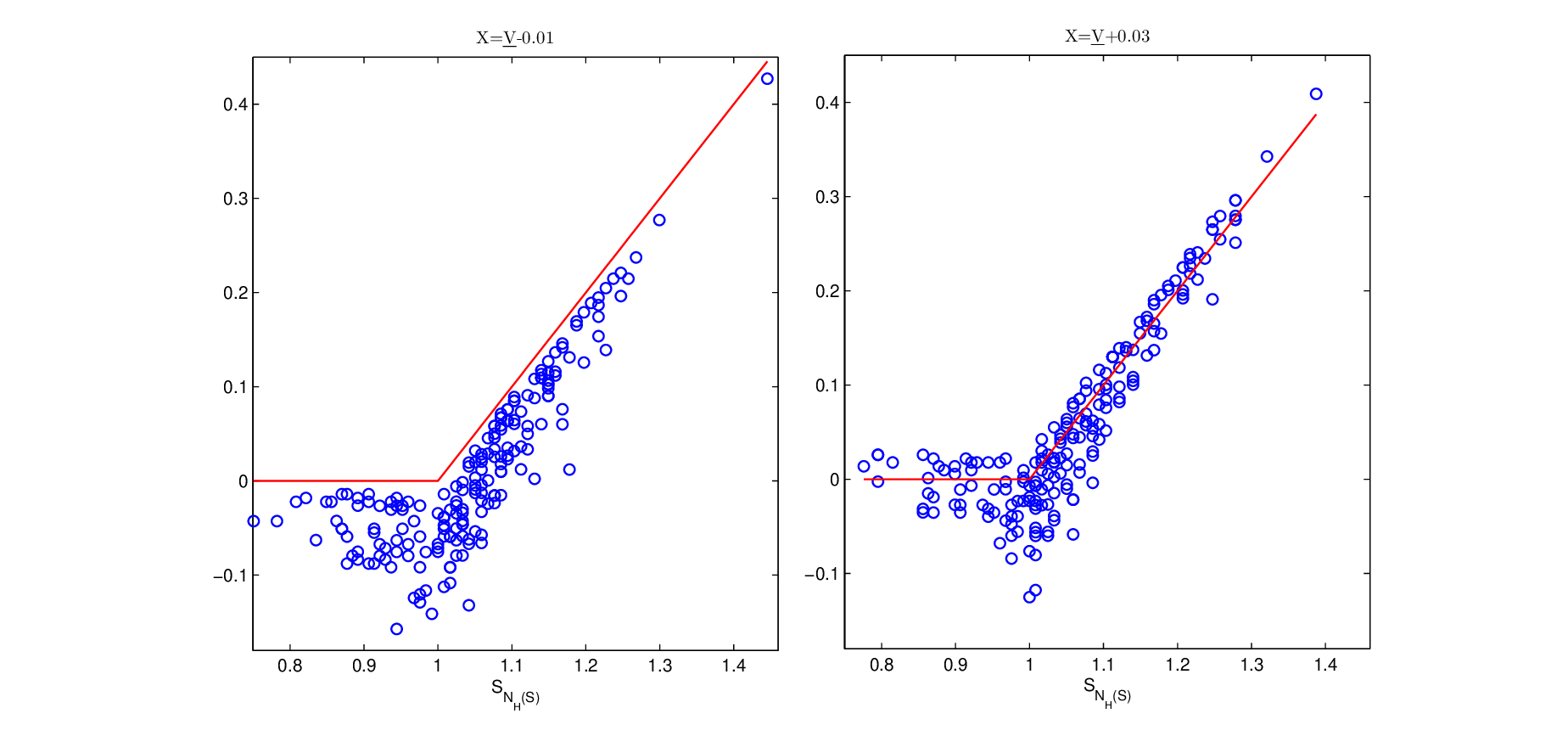}
\caption{Comparison between the hedging values for $X=\Vdo(s_0,Z,\Me)-0.01$ and $X=\Vdo(s_0,Z,\Me)+0.03$ and the payoff values.}
\label{fig:12}
\end{figure}

Finally, it is of interest to superimpose the upperhedging and lower hedging using $X = \overline{V}(s_0, Z)$ and
$X = \underline{V}(s_0, Z)$ respectively. Figure \ref{fig:hedgingA}
does this for $\mathcal{M}=\Swe(s_0, \alpha,Q) \times \He$ with $s_0=1$, $N_2=100$ and $p=3$.

\begin{figure}
\centering
\includegraphics[width=0.6\textwidth]{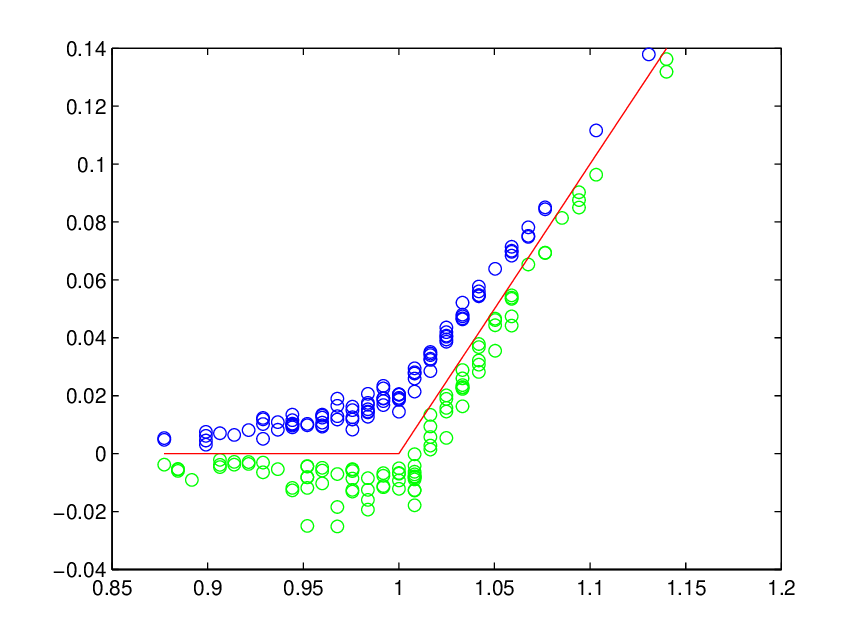}
\caption{Superhedging and underhedging comparison between the hedging values for $X=\Vup(s_0,Z,\Me)$ and $X=\underline{V}(s_0,Z,\Me)$ and the payoff values.}
\label{fig:hedgingA}
\end{figure}

\subsubsection{Effect of Variable Volatility}

This section illustrates the minmax bounds for several finite SQV markets (introduced in Section \ref{quadraticVariationCase}) related to the selection of the set $\Lambda$. Recall that $\Lambda$ gives the possible values of quadratic variation of the trajectories in the market.

We consider first markets where $\Lambda$ is a singleton set $\{n_{\theta}\}$ with $1\le \theta \le l$ and $n_{\theta} < n_{\theta+1}$. The corresponding markets are denoted by $\Me_{\theta}=\Swe(s_0, \alpha, Q_{\theta}) \times \He$, $1\le \theta \le l$, where $Q_{\theta}=\{n_{\theta} \delta^2 \}$. Figure \ref{fig:8}  shows the lower and upper bound as function of increasing values of the quadratic variation for two different options. A European call and a butterfly call option with strikes $K_1<K_2$ is defined by
\[ Z^f(X)=\left\lbrace
\begin{array}{ccc}
(X-K_1)^+ & \mbox{if} & X \le \frac{K_1+K_2}{2}\\
(K_2-X)^+ & \mbox{if} & X > \frac{K_1+K_2}{2}
\end{array}.
\right. \]
We will consider  $s_0=1$, $\alpha = 3 \cdot \sqrt{0.0067/200}$, $N_1=N_2$ and $n_{\theta}$ ranging from $25$ to $200$. So we build eight finite SQV markets $\Me_{\theta}=\Swe(s_0, \alpha, Q_{\theta}) \times \He$.

\begin{figure}
\centering
\includegraphics[width=0.9\textwidth]{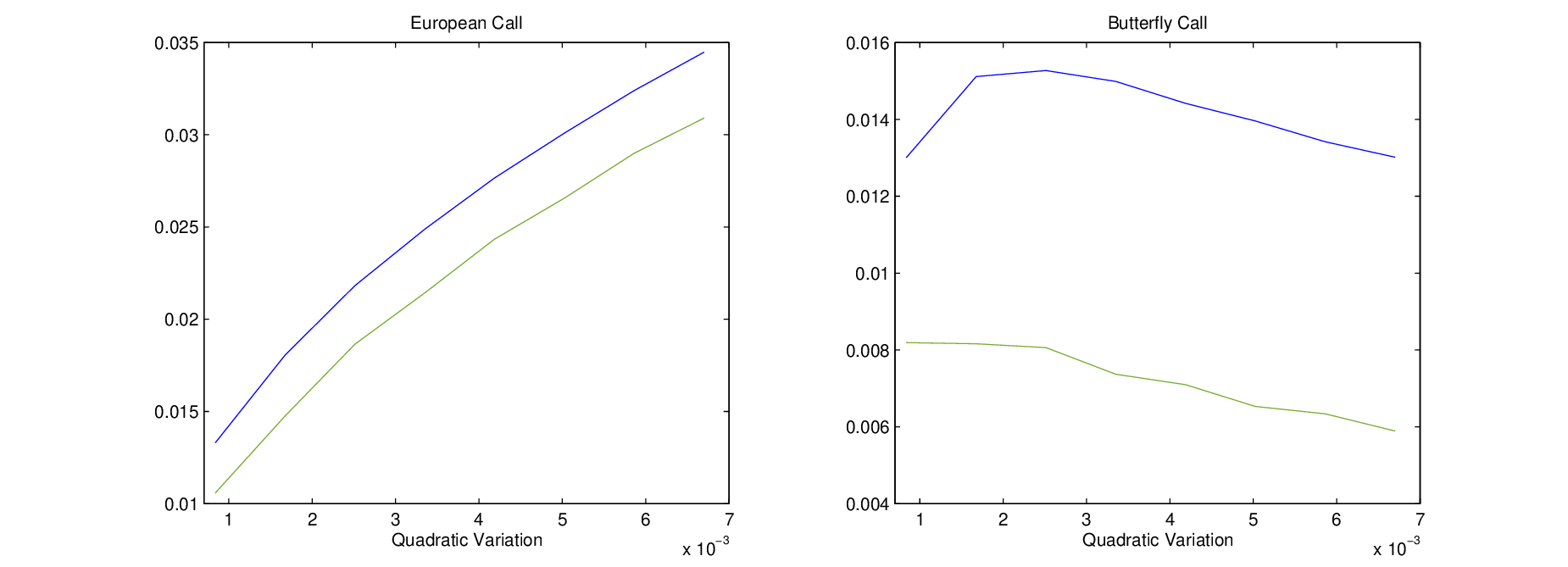}
\caption{Minmax upper and lower bounds price as a function of $v_{\theta}=n_{\theta}\delta^2$ in the $\Me_{\theta}$ for a European Call with $K=1$ and a butterfly Call with $K_1=1$ and $K_2=1.1$.}
\label{fig:8}
\end{figure}

It is observed that the bounds increase monotonically with respect to the quadratic variation for the case of an European Call but, for the case of a butterfly Call, the behavior is not monotonic. It is important to remark that the payoff of an European call is a convex function and the butterfly call is neither convex nor concave.

We now incorporate several possible quadratic variation values to the set $Q$. To this end, we build the finite SQV market $\Me_{\theta}=\Swe(s_0, \alpha, Q_{\theta}) \times \He$ where, in this case, $Q_{\theta}=\{n_1\delta^2, \dots, n_{\theta}\delta^2 \}$. Figure \ref{fig:13} shows the lower and upper bound as function of the sets $Q_{\theta}$ for a European call and a butterfly call.

\begin{figure}
\centering
\includegraphics[width=0.9\textwidth]{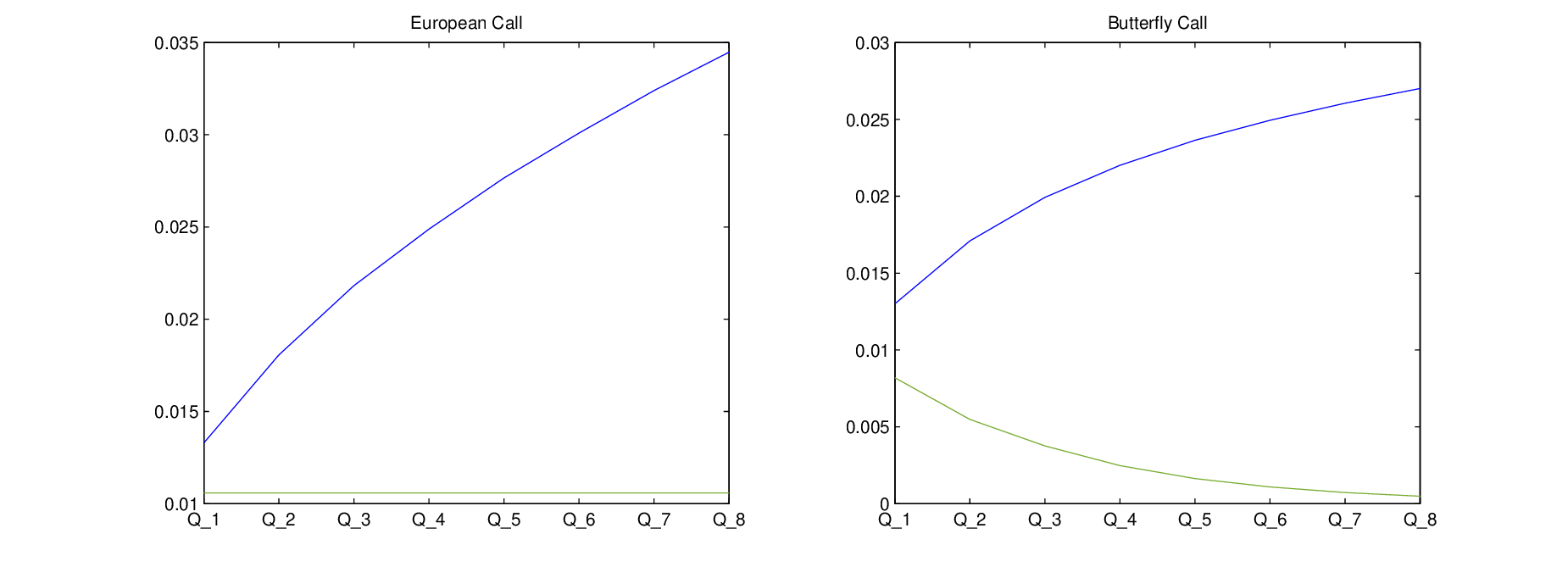}
\caption{Minmax upper and lower bounds price as a function of $Q_{\theta}=\{n_{1}\delta^2, \dots, n_{\theta}\delta^2 \}$ in the $\Me_{\theta}$ market for a European call with $k=1$ and a butterfly Call with $K_1=1$ and $K_2=1.1$.}
\label{fig:13}
\end{figure}

Note that for the European call, the upper bound graph coincide with the upper bound graph in Figure \ref{fig:8}. It means that the upper bound only depends on the maximum value of the set $Q$. Instead, the lower bound is constant for all $Q_{\theta}$ given that the lower bound only depends on the minimum value of the set $Q$. In the case of the butterfly call the upper bound increases monotonically and the lower bound decrease monotonically  as the size of  $Q_{\theta}$ increases and, so, reflecting a general feature of minmax pricing.   

\section{Conclusion}  \label{sec:conclusions}

General results are obtained to evaluate minmax bounds in an effective way and for general classes of trajectory markets assuming a bound on the number of possible trades. We perform explicit computations for the usual options, covering a new model where trajectories have different values of (sampled) quadratic variation. The numerical experiments indicate some of the phenomena that may occur in a trajectory based approach for the examples introduced. In particular, the effect of arbitrage nodes on prices is illustrated. Testing with different trajectory sets, we obtain narrower price intervals for European options. We conclude that designing suitable trajectory sets for different setups is a relevant task. The reference \cite{fleck} introduces models reflecting practical constraints with parameters estimated to market data.

\appendix
\section{Minmax Functions Results}   \label{minMaxFunctionsProofs}

This Appendix provides the main results on minmax function and the relation with the boundedness of $\overline{V}$ and $\underline{V}$. We will need the following definition.

\begin{definition}[Stopping Time] \label{stoppingtime}
Given a trajectory space $\Swe$ a {\it trajectory based stopping time} (or {\it stopping time} for short) is a function
$\nu: \Swe \rightarrow \mathbb{N}$ such that if $\se, \se' \in \Swe$ with $S_k=S'_k$ and $W_k=W'_k$ for $0 \leq k \leq \nu(\se)$ then $\nu(\se')
= \nu(\se)$.
\end{definition}

The integrability conditions, required for payoffs in a
probabilistic setting, are replaced in the proposed framework by
the, so called, minmax functions (introduced in \cite[Definition 14]{deganoI}). In what follows consider a
discrete market $\Me = \Swe \times \He$, and
a function $Z$ defined on $\Swe$.

\begin{definition}(Upper and Lower Minmax Functions) \label{minmaxFunctions}
Given a finite sequence of stopping times $(\nu_i)_{i=1}^{N}$ with
$\nu_i < \nu_{i+1}$ for $1 \le i < n$, a  real sequence
$(a_i)_{i=1}^{N}$, and $b\in \mathbb{R}$, we call $Z$ an \emph{upper
minmax function} if
\begin{equation}  \nonumber
Z(\se) \leq \sum_{i=1}^{N}a_i~\se_{\nu_i(\se)} + b,~\forall \se \in
\Swe.
\end{equation}
Similarly, $Z$ is called a \emph{lower minmax function} if
\begin{equation} \nonumber
Z(\se) \geq \sum_{i=1}^{N}a_i~\se_{\nu_i(\se)} + b,~\forall \se \in
\Swe.
\end{equation}
\end{definition}

Given a finite sequence of stopping times $(\nu_i)_{i=1}^{N}$ with
$\nu_i < \nu_{i+1}$ for $1 \le i < N$, and a  real sequence
$(a_j)_{j=1}^{N}$, set (set $\nu_0=0$ for convenience), define
\begin{equation}\label{coefficientDefinition}
 A_l(\se) = \sum_{j=i}^N a_j \quad \mbox{if} \quad \nu_{i-1}(\se) \le l < \nu_i(\se) \quad \mbox{for} \quad 1\le i
 \le N, \quad \mbox{and} \quad A_l(\se) =0 \quad \mbox{for} \quad l\ge\nu_N(\se).
\end{equation}
Also, for $H \in \He$, define the functions $H_i^{(A)}:\Swe \rightarrow \mathbb{R}$, for $\se \in
\Swe$, by
\begin{equation}
\label{defport}
H_i^{(A)}(\se) = H_i(\se)+A_i(\se)\quad \mbox{if}\quad 0\le i < \nu_N(S) \quad \mbox{with} \quad  V_{H^{(A)}}(0,S_0)=A_0 \quad \mbox{and}\quad N_{H^{(A)}}(\se) =\max\{N_H(\se),\nu_N(\se)\}.
\end{equation}

The fact that $H^{(A)}=(H_i^{(A)})_{i\ge 0}$, in the above definition,  is a portfolio on $\Swe$, for any  $H \in \He$,
is proven in the next Lemma. Observe first that, for a fixed $1 \le i < N$ and $\se\in \Swe$ we have
$a_i~S_{\nu_i(\se)} = a_i~S_0 + \sum_{l=0}^{\nu_i(\se)-1}a_i~\Delta_lS$.
Then
\begin{equation} \label{minmaxFunctionBound}
\sum_{i=1}^N a_i~S_{\nu_i(\se)}+b =
\sum_{l=0}^{\nu_N(\se)-1}A_l(\se)~\Delta_lS+ A_0~S_0 + b.
\end{equation}

\begin{lemma}
Assume $\nu_N(\se) \le M(\se)$ for each $\se \in \Swe$. For $H \in \He$, 
$H^{(A)}$ defined by \eqref{defport} is a portfolio on $\Swe$.
\label{lem:2}
\end{lemma}
\begin{proof}
It is enough to prove that the functions $A_l$, defined by (\ref{coefficientDefinition}) for $0  \le l < \nu_{N}$, are non anticipative. Hence, assume for $\se,\se'\in\Swe$, that
$S_j=S'_j$ for $0\le j \le l$ with $l < \min\{N_{H^{(A)}}(\se),N_{H^{(A)}}(\se') \}$. It follows from (\ref{coefficientDefinition}) that there exists $1\le i_0 \le N$
such that
\begin{equation} \label{coefficientIndex}
A_l(\se) = \sum_{j=i_0}^N a_j,  \quad \mbox{with}\quad \nu_{i_0-1}(\se)
\le l < \nu_{i_0}(\se).
\end{equation}
By hypothesis $S_j=S'_j$ for $0\le j \le \nu_{i_0-1}(\se)$, then
$\nu_{i_0-1}(\se)=\nu_{i_0-1}(\se')$. Also it must be $l <
\nu_{i_0}(\se')$, if not $l \ge \nu_{i_0}(\se')=\nu_{i_0}(\se)$ in
contradiction with (\ref{coefficientIndex}). Thus $A_l(\se)=A_l(\se')$.
\qed \end{proof}

Follows trivially from the above Lemma that for $\se \in \Swe$, and any $\se' \in \Swe_{(\se,k)}$,
\begin{equation}
\label{firstCoordinates}
\sum_{l=0}^{k-1}A_l(\se')~\Delta_lS' = \sum_{l=0}^{k-1}A_l(\se)~\Delta_lS.
\end{equation}

Next natural Proposition gives the key statements for the boundless
of $\Vup(Z)$ and $\Vdo(Z)$.

\begin{proposition}\label{finitenessCondition}
Let $\se \in \Swe$ be fixed, and $k \geq 0$, then
\begin{enumerate}
\item $\Vup_k(\se, Z, \Me) < \infty$ if and only if there exists $b\in \mathbb{R}$ and $H^b\in\He$ such that
\begin{equation}\label{superhedgingTrategy}
Z(\se')\le \sum_{i=k}^{N_{H^b}(\se')-1}H^b_i(\se')\Delta_iS' + b, \quad
\mbox{for all} \quad \se' \in \Swe_{(\se,k)}.
\end{equation}
In any case $\Vup_k(\se, Z, \Me) \le b$.
\item If there exists $b\in \mathbb{R}$ and $H^b\in \He$ such that
\begin{equation}\label{underhedgingTrategy}
Z(\se')\ge \sum_{i=k}^{N_{H^b}(\se')-1}H^b_i(\se')\Delta_iS' + b, \quad
\mbox{for all} \quad \se' \in \Swe_{(\se,k)},
\end{equation}
and either of the two statements below hold:
    \begin{enumerate}
    \item $\Me$ is conditionally $0$-neutral at $(\se,k)$ and for any $H \in \He$, $\tilde{H}$ defined by
    $\tilde{H}_i=H_i$ if $i\le k$ and $\tilde{H}_i=H_i-H_i^b$ if $i>k$, with  $N_{\tilde{H}}=\max\{N_{H},N_{H^b} \}$ and $V_{\tilde{H}}(s_0,0)=V_{H}(s_0,0)$, belongs to $\He$.
    \item $\Me$ is $n$-bounded such that $\Swe$ satisfies the local $0$-neutral property.
    \end{enumerate}
Then $\Vup_k(\se, Z, \Me) > -\infty$.
\end{enumerate}
\end{proposition}
\begin{proof}
Proof of part (1). Since $\Vup_k(\se, Z, \Me) < \infty$, there exist $H^b \in \He$ and $b \in \mathbb{R}$
such that
\[ \sup_{\se' \in \Swe_{(\se,k)}}[Z(\se')- \sum_{i=k}^{N_{H^b}(\se')}H^b_i(\se')\Delta_iS'] \le b. \]
From where (\ref{superhedgingTrategy}) holds. Conversely, if
(\ref{superhedgingTrategy}) is valid, it is clear that
\[
\overline{V}_k(\se, Z, \Me) \le \sup_{\se' \in
\Swe_{(\se,k)}}[Z(\se')- \sum_{i=k}^{N_{H^b}(\se')}H^b_i(\se')\Delta_iS']
\le b.
\]

Proof of part (2): $\tilde{H}$ is a non-anticipative function, then by the general hypothesis
\begin{eqnarray}\nonumber \label{lowerBoundV}
\Vup_k(\se, Z, \Me)&\ge& \inf_{H \in \He} \sup_{\se' \in \Swe_{(\se,k)}}[\sum_{i=k}^{N_H^b(\se')-1}H^b_i(\se')\,\Delta_iS' - \sum_{i=k}^{N_H(\se')-1}H_i(\se')\,\Delta_iS'] +b = \\
&=& \inf_{H \in \He} \sup_{\se' \in \Swe_{(\se,k)}}[-
\sum_{i=k}^{N_{\tilde{H}}(\se')-1}\tilde{H}_i(\se')\,\Delta_iS'] +b
\end{eqnarray}
Using now the condidtions in $(a)$, it follows that
\begin{eqnarray*}
\Vup_k(\se, Z, \Me)&\ge& \inf_{H \in \He} \sup_{\se'\in \Swe_{(\se,k)}}[- \sum_{i=k}^{N_H(\se')-1}H_i(\se')\,\Delta_iS'] +b = b.
\end{eqnarray*}
For the hypothesis $(b)$, consider the set of portfolios
$\mathcal{\widetilde{H}}$ consisting of all $\tilde{H}$ with $H \in
\He$ defined in $(a)$, then the market
$\Swe \times \mathcal{\widetilde{H}}$ is $n$-bounded and local $0$-neutral, and
then, by Proposition
\ref{0-neutral}, conditionally 0-neutral. Thus the right hand side of
\eqref{lowerBoundV} is equal to $b$.
\qed \end{proof}

Proposition \ref{finitenessCondition} holds in a more general scenario. The $n$-bounded condition in the second part can be replaced by the initially bounded condition defined in \cite{deganoI}, as follows.

\begin{definition}\label{initiallyBoundedDef}
Given a discrete  market $\Me=\Swe\times\He$ and $H \in \mathcal{H}$; we will call $N_H$ \emph{initially bounded} if there exists a bounded function $\rho:\Swe \rightarrow\mathbb{N}$ (which may depend on $H$) such that 
for all $\se \in \Swe$:
\begin{equation}  \label{initiallyBounded}
N_H  \mbox{ is bounded on } \Swe_{(\se, \rho(S))}.
\end{equation}
\end{definition}

Under this hypothesis, Theorem \ref{0-neutral} keep holding and then, we could formulate the next Proposition in this terms. But as the present work focus in $n$-bounded markets, we present the proof of Proposition \ref{finiteBoundsFromMinMaxFunctions} for this kind of markets. Observe that if $N_H$ is bounded, it is initially bounded, $\rho=N_H$ satisfies the definition.

\begin{proposition} \label{finiteBoundsFromMinMaxFunctions}
Let $\Me = \Swe \times \He$ be a discrete
market and $Z$ a function defined on $\Swe$. Consider a
finite sequence of stopping times $(\nu_i)_{i=1}^{N}$ with $\nu_i
< \nu_{i+1}$ for $1 \le i < N$ with $\nu_N(\se) \le M(\se)$ for all $\se \in \Swe$, a  real sequence
$(a_j)_{j=1}^{N}$, and $b\in \mathbb{R}$. Fix $\se \in
\Swe$ and an integer $k \geq 0$.  Then the following
statements are valid:
\begin{enumerate}
\item If $Z$ is an upper minmax function and $0^{(A^1)} \in \He$, then:
\begin{equation} \nonumber
\overline{V}_k(\se, Z, \Me) \leq A^1_0~s_0 + B^1.
\end{equation}
\item If $Z$ is a lower minmax function and $0^{(-A^2)} \in \He$, then:
\begin{equation} \nonumber
\underline{V}_k(\se, Z, \Me) \geq A^2_0~s_0 + B^2.
\end{equation}
\noindent Furthermore:
\item If $Z$ is a lower minmax function and either of the two statements below hold:
    \begin{enumerate}
    \item $\Me$ is conditionally $0$-neutral at $(\se,k)$ and for any $H \in \He$, $H^{(-A^3)}\in \He$.
    \item $\Me$ is $n$-bounded such that $\Swe$ satisfies the local $0$-neutral property and $\nu_N$ is bounded.
    \end{enumerate}
    Then:
    \begin{equation} \label{boundFromBelow}
    \overline{V}_k(\se, Z, \Me) \geq  A^3_0~s_0 + B^3.
    \end{equation}
\item If $Z$ is an upper minmax function and either of the two statements below hold:
     \begin{enumerate}
     \item $\Me$ is conditionally $0$-neutral at $(\se,k)$ and for any $H \in \He$, $H^{(A^4)}\in \He$.
     \item $\Me$ is $n$-bounded such that $\Swe$ satisfies the local $0$-neutral property and $\nu_N$ is bounded.
     \end{enumerate}
     Then:
     \begin{equation} \nonumber
      \underline{V}_k(\se, Z, \Me) \le A^4_0~s_0 + B^4.
     \end{equation}
\end{enumerate}
Where for $1\le i\le 4$ the sequences $(A^i_l)_{l\ge 0}$ are given by (\ref{coefficientDefinition}), and $B^i=
\sum_{l=0}^{k-1}A^i_l(\se)~\Delta_lS+b$ respectively, for each item.
\end{proposition}
\begin{proof}
Fix $\se' \in\Swe_{(\se,k)}$.
Proof of item (1). By (\ref{minmaxFunctionBound}) and \eqref{firstCoordinates}
\[ Z(\se') \le \sum_{i=0}^{k-1}A^1_i(\se)\Delta_iS + \sum_{i=k}^{\nu_{N}(\se')-1}A^1_i(\se')\Delta_iS' +  A^1_0s_0 + b = \sum_{i=k}^{N_{0^{(A^1)}}(\se')-1}0^{(A^1)}_i(\se')\Delta_iS' +  A^1_0s_0 + B^1.\]
Since $0^{(A^1)} \in \He$, Proposition
\ref{finitenessCondition}, part 1, gives
\[ \overline{V}_k(\se, Z, \Me) \le A^1_0s_0 + B^1. \]

Proof of item (2). From hypothesis $-Z(\se') \leq
\sum_{i=1}^{N}-a_iS'_{\nu_i(\se')} - b$, and $0^{(-A^2)} \in \He$,
it follows from (1) that
\[\underline{V}_k(\se, Z, \Me) = -\overline{V}_k(\se, -Z, \Me) \geq A^2_0s_0 + B^2.\]

Proof of item (3). For any $H \in \He$ it follows from
\eqref{minmaxFunctionBound}, and similar computation as in the proof
of part (1), that
\begin{eqnarray}
\nonumber Z(\se') - \sum_{i=k}^{N_H(\se')-1}H_i(\se')\,\Delta_iS' &\ge& \sum_{i=k}^{\nu_{N}(\se')-1}A^3_i(\se')\,\Delta_iS' ~ - \sum_{i=k}^{N_H(\se')-1}H_i(\se')\,\Delta_iS' +  A^3_0s_0 + B = \\
&=& - \sum_{i=k}^{N_{H^{(-A^3)}}(\se')-1}H_i^{(-A^3)}(\se)\,\Delta_iS' +
A^3_0s_0 + B^3. \label{eqn:1}
\end{eqnarray}
Under assumption (a) in item (3), we know that $H^{(-A^3)} \in
\He$, therefore by $0$-conditional property,
\begin{equation} \nonumber
\sup_{\se' \in \Swe_{(\se, k)}}[Z(\se') - \sum_{i=k}^{N_H(\se')-1}H_i(\se')\,\Delta_iS'] \ge A^3_0s_0 + B^3 + \sup_{\se' \in \Swe_{(\se, k)}}[ - \sum_{i=k}^{N_{H^{(-A^3)}}(\se')-1}H_i^{(-A)}(\se')\,\Delta_iS']
\ge
\end{equation}
\begin{equation} \nonumber
A^3_0s_0 + B^3 + \inf_{\tilde{H} \in \He}\sup_{\se' \in \Swe_{(\se, k)}}[ -\sum_{i=k}^{N_{\tilde{H}}(\se')-1}\tilde{H}_i(\se')\,\Delta_iS']
= A^3_0s_0 + B^3,
\end{equation}
Assume now $(b)$ in item (3) and $K \in \mathbb{R}$ such that $\nu_N(\se) \le K$ for all $\se \in \Swe$. Let $\tilde{\He}$ be any set containing the portfolios $H$ and $H^{(-A^3)}$ for each $H \in
\He$. Then, the market $\Swe \times \mathcal{\tilde{H}}$ is $N_0$-bounded, where $N_0=\max\{n,K\}$. Thus Theorem
\ref{0-neutral} shows that $\Swe \times \tilde{\He}$ is conditionally $0$-neutral at all
nodes, in particular at $(\se,k)$; therefore, taking the
supremum over $\Swe_{(\se, k)}$ in both sides of
\eqref{eqn:1}, evaluating the infimum over $H \in \He$ in
the right hand side, and using the conditional $0$-neutral property
of $\Swe \times \tilde{\He}$ we obtain
(\ref{boundFromBelow}).

The proof of item (4) follows from (3) in a similar way than (2)
from (1).
\qed \end{proof}

\section{Some Technical Results}  \label{sec:technicalresults}
Here are located some definitions and auxiliary lemmas required for results in subsections \ref{sec:fullPortfolios} and \ref{sec:u-completePortfolios}. Recall that at this section we assume $N_H(\se)=M(\se)$ for all $H \in \He$.

\subsection{Auxiliary Results for Subsection \ref{sec:fullPortfolios}}

\begin{definition}\label{inductionModel}
\noindent Consider a discrete market model $\Me =
\Swe \times \He$, and a function $Z$ defined on
$\Swe$. Fix $k\ge 0$, and $\se^k \in\Swe$ such that $M(\se^k)>k$. Set
$\hat{\emph s}_0=S^k_k$ and $\hat{\emph w}_0=W^k_k$. For $\se=\{S_i,W_i,m\}_{i\ge 0} \in \Se_{(\se^k,k)}$ and $H \in \He$ define
\begin{itemize}
\item $\hat{S}_i=S_{i+k}$, $\hat{W}=W_{i+k}$ and $\hat{m}=m-k$. Then $\hat{\se} \equiv (\hat{S},\hat{W},\hat{m})$.
\item $\hat{H}\equiv (\hat{H}_{i})_{i\geq 0}$ where $\hat{H}_i(\hat{\se})\equiv H_{i+k}(\se)$ and $V_{\hat{H}}(0,\hat{S}_0) = V_{H}(k,\se^k)$ (recall $N_{\hat{H}}=\hat{m}$).
\end{itemize}
Also define \[\widehat{\Swe} \equiv \{\hat{\se}: \se \in
\Se_{(\se^k,k)}\},\quad \mathcal{\widehat{H}}\equiv\{\hat{H}:
H\in\He\},\quad \mathcal{\widehat{M}}_k \equiv
\widehat{\Swe} \times \mathcal{\widehat{H}},\] and for any
$\hat{\se}\in \widehat{\Swe}$,
\[\hat{Z}(\hat{\se})\equiv Z(\se). \]
\end{definition}

\begin{lemma}\label{inductionModelproof}
Under the conditions of Definition \ref{inductionModel}, for any $k \ge 1$ and $\se^k \in\Se$ with $M(\se^k)>k$,
\begin{enumerate}
\item  $\mathcal{\widehat{M}}_k\equiv \widehat{\Swe} \times \mathcal{\widehat{H}}$ is a discrete market model, with initial value $\hat{\emph s}_0=S^k_k$ and $\hat{\emph w}_0=W^k_k$. Moreover it is $n$-bounded if $\Me$ is $n+k$-bounded.
\item Assuming $\Me$ is $n+k$-bounded, for any $\se \in \Swe_{(\se^k,k)}$,
\[\overline{U}_i(\hat{\se}, \hat{Z}, \mathcal{\widehat{M}}_k)= \overline{U}_{i+k}(\se, Z, \Me)\mbox{ for } 0\le i\le n.\]
 \item $\Vup_0(\hat{S}_0,  \hat{Z}, \mathcal{\widehat{M}}_k) = \Vup_k(S^k, Z, \Me)$.
\end{enumerate}
\end{lemma}
\begin{proof} By definition, $\widehat{\Swe}$ consist of sequences in $\mathbb{R}^{N} \times \mathbb{R}^{N} \times \mathbb{R}$, with $\hat{S}_0=S_k=S^k_k=\hat{s}_0$ and $\hat{W}_0=W_k=W^k_k=\hat{w}_0$ for any $\hat{\se}\in \widehat{\Swe}$. $\mathcal{\widehat{H}}$ is a family of sequences of functions
$(\hat{H})_{i\ge 0}$ with $\hat{H}_i: \widehat{\Swe} \rightarrow \mathbb{R}$. Lets see $\hat{H}$ is non-anticipative. Set $\hat{\se},\hat{\se'} \in \widehat{\Swe}$ such that $\hat{S'}_j=\hat{S}_j$ and $\hat{W'}_j=\hat{W}_j$, for $0\le j\le
i$ with $i<\min\{N_{\hat{H}}(\hat{\se}),N_{\hat{H}}(\hat{\se'})\}=\min\{M(\hat{\se}),M(\hat{\se'})\}$. Then by definition $S'_{j}=S_j$ and $W'_j=W_j$, for $0\le j\le i+k$ with $i+k<\min\{M(\se),M(\se')\}$. Therefore
\[ \hat{H}_i(\hat{\se'})=H_{i+k}(\se')=H_{i+k}(\se)=\hat{H}_i(\hat{\se}) \]
since $H$ is non-anticipative. Thus $\mathcal{\widehat{H}}$ is a set of portfolios on
$\widehat{\Swe}$. Furthermore, if $\Me$ is $n+k$-bounded, for each $\hat{\se} \in \widehat{\Swe}$, we have $M(\hat{\se})=M(\se)-k<n+k-k=n$. This proves $(1)$.

For $(2)$, we proceed by induction backwards over $i$. Let $i=n$ and $\hat{\se} \in \hat{\Swe}$, then $M(\hat{\se}) \le n$ since $\hat{\Me}_k$ is $n$-bounded. If $n = M(\hat{\se})$, then $n+k = M(\se)$ and
\[ \overline{U}_{n}(\hat{\se}, \hat{Z}, \mathcal{\widehat{M}}_k)= \hat{Z}(\hat{\se}) = Z(\se) = \overline{U}_{n+k}(\se, Z, \Me).\]
But if $M(\hat{\se})<n$, then $M(\se)<n+k$ and $\overline{U}_n(\hat{\se}, \hat{Z}, \mathcal{\widehat{M}}_k)= 0 = \overline{U}_{n+k}(\se, Z, \Me)$. Therefore $\overline{U}_n(\hat{\se}, \hat{Z}, \mathcal{\widehat{M}}_k)= \overline{U}_{n+k}(\se, Z, \Me)$.
Now assume $(2)$ is valid for $0<i\le n$. Set $\hat{\se} \in \mathcal{\widehat{\Swe}}$ and suppose first $M(\hat{\se}) \le i-1$. Similar analysis shows that $\overline{U}_{i-1}(\hat{\se}, \hat{Z}, \mathcal{\widehat{M}}_k)= \overline{U}_{i+k-1}(\se, Z, \Me)$. Suppose now $M(\hat{\se})>i-1$, then
\begin{eqnarray*}
\overline{U}_{i-1}(\hat{\se}, \hat{Z}, \mathcal{\widehat{M}}_k) & = & \inf_{\hat{H} \in \mathcal{\widehat{H}}} \sup_{\hat{\se'} \in \widehat{\Swe}_{(\hat{\se},i-1)}} ~[\overline{U}_{i}(\hat{\se'}, \hat{Z}, \mathcal{\widehat{M}}_k) - \hat{H}_{i-1}(\hat{\se}) \Delta_{i-1}\hat{\se'}] \\
& = & \inf_{H\in \He} \sup_{\se' \in \Se_{(\se,i+k-1)}} ~~[\overline{U}_{i+k}(\se', Z, \Me) - H_{i+k-1}(\se) \Delta_{i+k-1}\se'] =\\
& = & \overline{U}_{i+k-1}(\se, Z, \Me).
\end{eqnarray*}
by inductive hypothesis and Definition \ref{inductionModel}. Then we get $(2)$.

Now we will prove $(3)$. Since $\widehat{\Swe}=\Se_{(\se^k,k)}$, it follows that
\begin{eqnarray*}
\Vup(\hat{S}_0, \hat{Z}, \mathcal{\widehat{M}}_k) & = & \inf_{\hat{H} \in \mathcal{\widehat{H}}} \sup_{\hat{\se'} \in \widehat{\Swe}}[\hat{Z}(\hat{\se})-\sum_{i=0}^{M(\hat{\se'})-1}\hat{H}_i(\hat{\se'})\Delta_i\hat{\se'}] = \\
& = & \inf_{H \in \mathcal{H}} \sup_{\se' \in \Swe_{(\se^k,k)}}[Z(\se')-\sum_{i=0}^{M(\se')-k-1}H_{i+k}(\se')\Delta_{i+k}\se'] = \\
& = & \inf_{H \in \mathcal{H}} \sup_{\se' \in \Swe_{(\se^k,k)}}[Z(\se')-\sum_{i=k}^{M(\se')-1}H_i(\se')\Delta_i\se'] = \Vup(\se^k,Z,\Me).
\end{eqnarray*}
\qed \end{proof}

\begin{lemma}\label{FullinductionModelproof} Consider the $n$-bounded market
$\mathcal{\widehat{M}}_1\equiv \mathcal{\widehat{S}}\times
\mathcal{\widehat{H}}$, given in definition \ref{inductionModel},
for $k=1$ and some $S^1 \in\Se$ in an $(n+1)$-bounded market
$\Me = \Se\times \He$. 
Then $\mathcal{\widehat{H}}$ is {\small FULL} if so is also
$\He$.
\end{lemma}
\begin{proof} Assume $\He$ is {\small
FULL}. Let $1\le k\le n-1$, $\hat{H'}\in \mathcal{\widehat{H}}$, and
$\hat{S'}\in \mathcal{\widehat{S}}$. We are going to prove that for
$~k\le j\le n-1$, any function
\[h:\mathcal
{\widehat{S}}_{(\hat{S'},k)} \rightarrow I^j_{ \widehat{\Se}_{(\hat{S'},k)} }\] non-anticipative with respect to $j$, is the $j$-coordinate of a portfolio $\hat{H}\in \mathcal{\widehat{H}}$.
For it, we will find $H\in \He$ such that $\hat{H}_{j}=h$ on
$\widehat{\Se}_{(\hat{S'},k)}$.

\vspace{.1in}\noindent We need to show that $\hat{S}\in \mathcal
{\widehat{S}}_{(\hat{S'},k)}$ if and only if $S\in \mathcal
{S}_{(S',k+1)}$. Let $0\le i\le k$, and $\hat{S}\in \mathcal
{\widehat{S}}_{(\hat{S'},k)}$,  then

$S_{i+1} = \hat{S}_i = \hat{S'}_i = S'_{i+1}.$ On the other hand
$\hat{S}_i = S_{i+1} = S'_{i+1}= \hat{S'}_i.$

\vspace{.1in}\noindent If $\hat{H}\in\mathcal{\widehat{H}}$, and
$\hat{S}\in \mathcal {\widehat{S}}$, $\hat{H}_{j}(\hat{S})=
H_{j+1}(S)$, it means that $I^j_{ \widehat{\Se}_{(\hat{S'},k)} }\subset
I^j_{\Se_{(S',k+1)}}$. Since $\He$ is {\small\emph {FULL}} it then
follows that exists $H\in \He$ such that $H_{j+1}:\mathcal
{S}_{(S',k+1)} \rightarrow I^j_{\Se_{(S',k+1)}}$ is given by
$H_{j+1}(S)\equiv h(\hat{S})$.
\qed \end{proof}

\subsection{Proofs for u-Complete Markets Section}\label{sec:u-completeProofs}
Consider a discrete $(n+1)$-bounded market $\Me
= \Swe \times \He$. For any $\se \in \Swe$
define $\tilde{\se}$ by $(\tilde{S}_i,\tilde{W}_i)=(S_i,W_i)$ for $0\le i$ and
\[ M(\tilde{\se})=\left\{\begin{array}{ccc} n & \mbox{if} & M(\se)=n+1. \\ \\
M(\se) & \mbox{if} & M(\se)\le n.
\end{array}\right.\]
Set $\widetilde{\Swe} \equiv \{\tilde{\se}:\se \in \Swe\}$ and define $\mathcal{\widetilde{M}} \equiv \widetilde{\Swe} \times \He$. $\mathcal{\widetilde{M}}$ results an $n$-bounded discrete market.

If $Z$ is a derivative function defined on $\Swe$, then
$\tilde{Z}$ is defined on $\widetilde{\Swe}$ by
\[\tilde{Z}(\tilde{\se})=\left\{\begin{array}{ccc} \overline{U}_n(\se, Z, \Me) & \mbox{if} & M(\se)=n+1. \\ \\
Z(\se) & \mbox{if} & M(\se)\le n.
\end{array}\right.\]
for any $\se \in \widetilde{\se}$.

Moreover
\begin{lemma}\label{u-completeInductionLemma}
Let $\Me = \Swe \times \He$ an
$(n+1)$-bounded discrete market. Then
\begin{enumerate}
\item For any $0\le k \le n$, and $\se \in \Swe$,
\[\overline{U}_k(\tilde{\se}, \tilde{Z}, \mathcal{\widetilde{M}})=\overline{U}_k(\se, Z,
\Me).\]
\item If $\Me$ is \emph{u-complete} for $Z$, so
is $\mathcal{\tilde{M}}$ for $\tilde{Z}$.
\end{enumerate}
\end{lemma}
\begin{proof}
Reasoning by induction backwards, for
$k=n$, and $\se \in \Swe$,
\begin{eqnarray*}
\overline{U}_n(\tilde{\se}, \tilde{Z},\mathcal{\widetilde{M}})=0=\overline{U}_n(\se, Z, \Me) &\mbox{if}& n>M(\tilde{\se}) \\
\overline{U}_n(\tilde{\se}, \tilde{Z},\mathcal{\widetilde{M}})=\tilde{Z}(\tilde{\se})=\overline{U}_n(\se, Z,\Me) &\mbox{if}& M(\tilde{\se})=n.
\end{eqnarray*}
Since if $M(\se)=n+1$, $\tilde{Z}(\tilde{\se})=\overline{U}_n(\se, Z,\Me)$, and if $M(\se)=n$, $\tilde{Z}(\tilde{\se})=Z(\se)=\overline{U}_n(\se, Z, \Me)$. Assume
$(1)$ is valid for some $0<k\le n$. If $M(\tilde{\se})\le k-1$, then
$M(\tilde{\se})=M(\se)$ and, $\overline{U}_{k-1}(\tilde{\se}, \tilde{Z},
\mathcal{\widetilde{M}}) = \overline{U}_{k-1}(\se, Z, \Me)$,
since its common value is $0$ or $\tilde{Z}(\tilde{\se})=Z(\se)$. If $k-1 < M(\se)$, then $k-1 < M(\tilde{\se})$ ($k-1 \ge \tilde{M}$ implies $M(\se)>M(\tilde{\se})$, then
$M(\tilde{\se})=n >k-1$ !), and by inductive hypothesis and
definition of $ \mathcal{\widetilde{M}}$,
\begin{eqnarray*}
\overline{U}_{k-1}(\tilde{\se}, \tilde{Z}, \mathcal{\widetilde{M}})&=&
\inf_{H \in \He}~~\sup_{\tilde{\se'} \in
\widetilde{\Swe}_{(\tilde{\se},k-1)}} ~~[\overline{U}_{k}(\tilde{\se'}, \tilde{Z},
\mathcal{\widetilde{M}}) - H_{k-1}(\tilde{\se'})
\Delta_{k-1}\tilde{S'}] \\
&=& \inf_{H \in \He}~~\sup_{\se'
\in \Swe_{(\se,k-1)}} ~~[\overline{U}_{k}(\se', Z, \Me)
- H_{k-1}(\se') \Delta_{k-1}S'] = \overline{U}_{k-1}(\se, Z,
\Me).
\end{eqnarray*}
For $(2)$, let $\tilde{\se^*}\in \widetilde{\Swe}$, $1\le k\le n-1$ and a
derivative function $Z$. Since $\Me$ is \emph{u-complete}
there exists $H^*\in \He\;$, such that
\[ \sup_{\tilde{\se} \in \widetilde{\Swe}_{(\tilde{\se^*},k)}}[\Uup_{k+1}(\tilde{\se}, \tilde{Z}, \mathcal{\widetilde{M}})-H_k^*(\tilde{\se})\Delta_k\tilde{S}]=\sup_{\se \in\Swe_{(\se^*,k)}}[\Uup_{k+1}(\se, Z, \Me)-H^*_k(\se)\Delta_kS]=
\overline{U}_k(\se^*, Z, \Me)= \overline{U}_k(\tilde{\se^*},\tilde{Z}, \mathcal{\widetilde{M}}).\]
Last equalities hold for $(1)$.
\qed \end{proof}

\begin{proof} of {\bf Proposition \ref{u-completionOfMarkets}.}  Define $G:\mathbb{R}\rightarrow \mathbb{R}$, by
\[G(u)= \sup_{\se \in \Swe_{(\se^*,k)}} ~\{\overline{U}_{k+1}(\se, Z, \Me) - u\Delta_kS\},\]
assuming that $\overline{U}_{k+1}(\se, Z, \Me^{n})<\infty$.
Since for any $\se \in \Swe_{(\se,k)}$, the functions given by
$G_{\se}(u)=\overline{U}_{k+1}(\se, Z, \Me^{n}) - u ~
\Delta_k~S$ are affine, then its supremum $G$ is lower
semicontinuous, and convex.

If $I^k_{\se^*}$ is compact, by lower semicontinuity, there exists
$u^*\in I^k_{\se^*}$ verifying $G(u^*)= \inf_{u \in I^k_{\se^*}}
G(u).$

If $I^k_{\se}=\mathbb{R}$ and $\Swe$ satisfies the local
up-down property at $\se^*$ and $k$, $G$ is also coercive. Indeed,
there exist $S^+,S^-\in \Swe_{(\se^*,k)}$ such that
$S^+_{k+1}-S_k=r^+>0$ and $S^-_{k+1}-S_k=r^-<0$. Let $m\in
\mathbb{N}$ and
\[K = \max \left\{|\frac{m-\overline{U}_{k+1}(\se^+, Z, \Me^{n})}{r^+}|, |\frac{\overline{U}_{k+1}(\se^-, Z, \Me^{n})-m}{r^-}|\right\}.\]
If $u>K$, $u=|u|>\frac{\overline{U}_{k+1}(\se^+, Z, \Me^{n})-m}{r^-}$, then
$\;m < \overline{U}_{k+1}(\se^-, Z, \Me^{n}) - u ~
\Delta_kS^- \le G(u)$. On the other hand, if $u<-K$, since $-u=|u|
> \frac{m-\overline{U}_{k+1}(S^+, Z, \Me^{n})}{r^+}$, then
$G(u)\ge \overline{U}_{k+1}(S^+, Z, \Me^{n}) - u ~
\Delta_kS^+ > m$.

\noindent Thus, by Corollary 4.3 in \cite{agheksanterian}, from
\cite[Thm 7.3.1]{kurdila} $G$ attains a minimizer.

\noindent Finally, by coercivity, there exists $R>0$ such that,
$G(u)>|G(0)|\ge G(0)$ if $|u|>R$. Then
\[\inf\{G(u) : |u|\le R\} \le G(0) \le \inf\{G(u) : |u|>R\}.\]
\qed \end{proof}

\begin{proof} of {\bf Proposition \ref{usefulThForUcomplete}.}
First it is necessary to show that $H^*$ defined by
(\ref{u-completionPortfolio}) is non-anticipative. Let
$\se,\se'\in \Swe$ with $(S_i,W_i)=(S'_i,W'_i)$ for $0\le i \le k$ with $k \le \min \{N_{H*}(\se),N_{H*}(\se') \}=\min \{M(\se),M(\se')\}$, then
$\Swe_{(\se,k)}=\Swe_{(\se',k)}$ and
$I^k_{\se}=I^k_{\se'}$ since $N_H$ is a stopping time for all $H \in \He$, so $H^*_k(S')=u^*=H^*_k(S)$.

For the u-completion of $\Me^*$, we first prove by backward induction that $\Uup_i(\se, Z, \Me) = \Uup_i(\se, Z, \Me^*)$ for any $0\le i \le n$. It is clear that $\Uup_{i}(\se, Z, \Me) = \Uup_{i}(\se, Z, \Me^*)$, for all $i \ge M(\se)$. Let $\se \in \Swe$ such that $M(\se')=n$ for all $\se' \in \Swe_{(\se,n-1)}$. Then
\begin{eqnarray*}
\Uup_{n-1}(\se, Z, \Me^*) & = & \inf_{u \in I^{*\,n-1}_{\se}}~\{~\sup_{\se'\in\Swe_{(\se,n-1)}}[\Uup_{n}(\se',Z,\Me^*)-u~\Delta_{n-1}S']\}= \\
& = & \inf_{u \in I^{n-1}_{\se}}~\{~\sup_{\se'\in\Swe_{(\se,n-1)}}[\Uup_{n}(\se',Z,\Me)-u~\Delta_{n-1}S']=  \Uup_{n-1}(\se, Z, \Me).
\end{eqnarray*}
since $I^{*\,n-1}_{\se} =\{H_{n-1}(\se):H\in\He\}\cup\{H^*_{n-1}(\se)\}=I^{n-1}_{\se}$. Assume now $\se \in \Swe$ such that $M(\se')\ge i+1$ for all $\se' \in \Swe_{(\se,i)}$ and suppose $\Uup_{i+1}(\se', Z, \Me) = \Uup_{i+1}(\se', Z, \Me^*)$ for all $\se'$. Then
\begin{eqnarray*}
\Uup_{i}(\se, Z, \Me^*) & = & \inf_{u \in I^{*\,i}_{\se}}~\{~\sup_{\se'\in\Swe_{(\se,i)}}[\Uup_{i+1}(\se',Z,\Me^*)-u~\Delta_{i}S']\}= \\
& = & \inf_{u \in I^{i}_{\se}}~\{~\sup_{\se'\in\Swe_{(\se,i)}}[\Uup_{i+1}(\se',Z,\Me)-u~\Delta_{i}S']= \Uup_{i}(\se, Z, \Me).
\end{eqnarray*}
since $I^{*\,i}_{\se} =\{H_{i}(\se):H\in\He\}\cup\{H^*_{i}(\se)\}=I^{i}_{\se}$. Finally for (\ref{u-completionEquation}), for any $i\ge 0$,
\[\Uup_i(\se, Z, \Me^*) = \Uup_i(\se, Z, \Me) = \sup_{\se' \in \Swe_{(\se,i)}}~[\Uup_{i+1}(\se', Z,\Me)- H^*_i(\se) \Delta_iS'],\]
with $H^*\in\He^*$.
\qed \end{proof}

\section{Auxiliary results}\label{sec:auxiliary}
The next geometric Lemma is used in section \ref{sec:convexEnvelope}.
\begin{lemma}
Let $A, B, C, D, s_1, s_2, s \in \mathbb{R}$, with $s_1<s_2$ and $s_1 \le s \le s_2$. If $A>B$ and
$C>D$, then
\begin{equation}
B-\left( \frac{B-D}{s_2-s_1} \right)(s_2-s) \le A-\left( \frac{A-C}{s_2-s_1}
\right)(s_2-s) \label{eqn:57}
\end{equation}
\label{lem:1}
\end{lemma}
\begin{proof}
Let
\[ \lambda= \frac{s_2-s}{s_2-s_1} \]
Since $s_1 \le s \le s_2$, it follows that $0 \le \lambda \le 1$. Then
\[ \lambda(A-B-(C-D)) \le A-B, \]
re-arranging the last inequality we obtain \eqref{eqn:57}.
\qed \end{proof}

\section{Computational Grid}  \label{gridCharacterization}

Here we are going to introduce a grid of pairs of integer numbers $\Gamma$, which will be used to represent the trajectories of a finite discrete market. The purpose of the grid $\Gamma$ is to give a combinatorial way to build finite trajectory sets and implement an efficient algorithm in order to evaluate the dynamic bounds $\Uup_i(\se,Z,\Me)$ for a finite discrete market. Consequently, under appropriate conditions, we will obtain also the global bound $\Vup_0(s_0,Z,\Me)$.

Given the discretization parameters $\delta,\beta>0$ and $p,q,N_1, N_2 \in \mathbb{N}$, we call \emph{trajectory grid} to
\[\Gamma =\{(k,j): |k|\le N_1, 0\le j \le N_2, -pj\le k\le pj\}.\]
For any $i\ge 0$, each node $\se_i=(S_i,W_i,m)$ of a trajectory $\se \in \Swe(s_0, \delta, \beta, p,q, N_1,\Lambda)$ can be represented by a vertex $(k_i,j_i) \in \Gamma$, such that
\begin{equation}\label{gridCorrespondence}
S_i=s_0\,e^{k_i \delta}\quad \mbox{and}\quad W_i=j_i \beta^2.
\end{equation}
It has shown in Section \ref{discretization} that it is enough that $N_1\le p\,N_2$. Also observe that the constrains of Definition \ref{implicitDefinitionThroughConstraints} are translated to the grid information: if $\se \in \Swe(s_0, \delta, p, \Lambda, N_1, N_2)$ then
\begin{eqnarray} \label{eqn:gridconst}
\nonumber |\log S_{i+1} - \log S_i|\le p \delta &\Leftrightarrow & |k_{i+1}-k_i| \le p \\
\nonumber 0< W_{i+1}-W_i \le q \beta^2  &\Leftrightarrow & 0<j_{i+1}-j_i \le q \\
W_{M(\se)} \in Q_{\Lambda} &\Leftrightarrow& j_{M(\se)} \in \Lambda.
\end{eqnarray}

\begin{remark}
Note that if $\se^1,\se^2 \in \Swe(s_0, \delta, \beta,p, q,N_1,\Lambda)$ such that $S_i^1=S_i^2$ and $W_i^1=W_i^2$ for all $i \in \mathbb{N}$, and $M(\se^1) \neq M(\se^2)$, then $\se^1$ and $\se^2$ are associated with the same vertex in $\Gamma$, but $j_{M(\se^1)}=n_{\theta^1} \in \Lambda$ and $j_{M(\se^2)}=n_{\theta^2} \in \Lambda$ with $\theta^1 \neq \theta^2$.
\end{remark}

On the other hand, any sequence $\{(k_i,j_i)\}_{i \ge 0}$, with the constrains listed on the left side of \eqref{eqn:gridconst}, corresponds by the same association (\ref{gridCorrespondence}), to a trajectory ${\bf S}$ satisfying the constrains of Definition \ref{implicitDefinitionThroughConstraints}. Then, given a trajectory grid $\Gamma$ with parameters $p,q,N_1$ and $\Lambda$, we can build a finite trajectory set $\Swe_{\Gamma}(s_0, \delta, \beta,p, q,N_1,\Lambda)$ for appropiate $\delta$ and $\beta$, in such way that any possible path in $\Gamma$ with the constrains listed on $\eqref{eqn:gridconst}$ corresponds to a trajectory in $\Swe_{\Gamma}(s_0, \delta, \beta,p, q,N_1,\Lambda)$, and the inverse implication also holds.

\begin{remark}
Note that a grid $\Gamma$ does not contain necessarily all the path satisfying the constrains listed in \eqref{eqn:gridconst}. For example, the next grid satisfies the conditions for $p=q=N_1=1$ and $\Lambda=\{ 1\}$ and do not contain the path $(k_0,j_0),(k_1,j_1)$ such that $k_1-k_0=-1$.
\[ \xymatrix@R-1pc{                           & \bullet \\
                   \bullet \ar[ur] \ar[r] & \bullet                                   }\]

\end{remark}

\subsection{Computation of Prices in the Grid}\label{sec:computationGrid}

The trajectory grid $\Gamma$ presented above  will be used to compute the dynamic bounds $\Uup_i(\bf{S},Z,\Me)$ where $\Me=\Swe_{\Gamma}(s_0, \delta,\beta,p,q, N_1,\Lambda)\times\He$ is the finite discrete market associated to the grid $\Gamma$ with parameters $p,q,N_1$ and $\Lambda$. To this end, we will using Theorem \ref{ConvexHullThm}. For reasons of space, we will use the abbreviated notation $\Swe_{\Gamma}=\Swe_{\Gamma}(s_0, \delta,\beta,p,q, N_1,\Lambda)$.

Let $Z$ an European option defined on $\Swe_{\Gamma}$. The option is assumed independent of the trajectory history,  namely  $Z({\bf S})=Z^f(S_{M({\bf S})})$ for a real variable function $Z^f$. This condition on $Z$ allows to compute the dynamic bounds on the vertices of $\Gamma$ as follows. For simplicity we will use the notation $s_k = s_0\,e^{k\,\delta}$. Also assume that the set of portfolios $\He$ is composed for sequences $H=\{H_i\}_{i\ge 0}$ including any function from $\Swe_{\Gamma}$ to $\mathbb{R}$, non anticipative with respect to $i$, thus $\He$ is {\small FULL}.

Now we describe an algorithm that works for the case $\Lambda=\{ N_2 \}$. The dynamic bounds $\Uup_i({\bf S},Z,\Me)$ for $0\le i\le N_2$, can be associated to the vertices of $\Gamma$. Indeed since $W_{M(\bf S)} = N_2\beta^2$ the node $\se_{M(\se)}=(S_{M(\bf S)},W_{M(\bf S)},M(\se))$ corresponds by (\ref{gridCorrespondence}) to some $(k_M,N_2)\in\Gamma$ ($M\equiv M(\bf S)$), then $\Uup_M({\bf S},Z,\Me) = Z^f(s_{k_M})$. Moreover whenever the trajectory ${\bf S}$ has a node $(s_{k_M},N_2\beta^2)$, will have
\[\Uup_M({\bf S},Z,\Me)=Z^f(s_{k_M}).\]
Now, the grid node $(k_{i^0},j_{i^0})$ correspond to a trajectory $\se \in \Swe_{\Gamma}$ at stage $i^0$. We know from Definition \ref{dynamicBounds} and Theorem \ref{ConvexHullThm} that $\Uup_{i^0}({\bf S},Z,\Me)$ only depends on $\Uup_{i^0+1}({\bf S}',Z,\Me)$, $S'_{i^0+1}$ and $S_{i^0}$, where ${\bf S}' \in \Swe_{({\bf S},i^0)}$. Then, by \eqref{eqn:gridconst}, those quantities are associated to the vertices $(k,j)\in\Gamma$ with
\begin{equation}\label{reachableVertex}
-p\le k-k_{i^0}\le p,\quad \mbox{and}\quad 0<j-j_{i^0}\le q.
\end{equation}
Vertices $(k,j)\in\Gamma$ verifying (\ref{reachableVertex}) are called \emph{reachable} from $(k_{i^0},j_{i^0})$.

$\Uup_{i^0}({\bf S},Z,\Me)$ can be associated with the vertex $(k_{i^0},j_{i^0})$, via a function $\Uup$ with domain $\Gamma$ in such way that $\Uup(k_{i^0},j_{i^0})=\Uup_{i^0}({\bf S},Z,\Me)$. Thus, for each vertex $(k,j)\in \Gamma$ we define $\Uup$ by the following procedure. Since any vertex $(k,N_2)\in \Gamma$ corresponds to a trajectory $\se \in \Swe_{\Gamma}$, with $\se_{M(\se)}=(s_k,N_2\beta^2,m)$, define
\[\Uup(k,N_2)= Z^f(s_k),\quad\mbox{for any}\quad k~:~|k|\le N_1.\]
Now assume, for fixed $j<N_2$, $\Uup(k^*,j^*)$ was defined for any $j^*~:~j <j^*\le N_2$, and
any $k^*~:~|k^*|\le p\,j^*$. For fixed $(k,j)\in\Gamma$ and any pairs $(k^+,j^+)$, $(k^-,j^-)$ verifying
\begin{eqnarray}
\nonumber 0 <k^+-k \le p &\mbox{and}& 0<j^+-j\le q \\
-p\le k^--k\le 0 &\mbox{and}& 0<j^--j\le q,
\label{+-reachable}
\end{eqnarray}
set
\[\Delta^{\pm} \equiv \frac{\Uup(k^+,j^+)-\Uup(k^-,j^-)}{s_{k^+}-s_{k^-}}. \]
Being $\se \in \Swe_{\Gamma}$ a trajectory such that $\se_i$ corresponds by (\ref{gridCorrespondence}) to $(k,j)$, it is important to notice that the pairs $(k^+,j^+)$ and $(k^-,j^-)$ verifying (\ref{+-reachable}) are reachable from $(k,j)$, if ${\bf S}^+,{\bf S}^-\in \Swe_{({\bf S},k)}$ verify that $\se^+_{i+1}$ and $\se^-_{i+1}$ corresponds respectively to $(k^+,j^+)$ and $(k^-,j^-)$, then ${\bf S}^+ \in \Sup_{({\bf S},i)}$ and ${\bf S}^-\in \Sdo_{({\bf S},i)}$. Consequently Theorem \ref{ConvexHullThm} is applicable and $\Uup(k,j)$ is defined according to it, by
\begin{equation} \label{Udefinition}
\Uup(k,j)\equiv C(k,j)=\sup~\{\Uup(k^+,j^+)-\Delta^{\pm}(s_{k^+}-s_k)\},\quad \emph{for}\quad 0 \le j < N_2\;\; \emph{and}\;\; |k|\le p\,j,
\end{equation}
\[\emph{where the supremum is taken over the pairs}~~(k^+,j^+), (k^-,j^-)~~ \emph{verifying (\ref{+-reachable})}.\]

Therefore, the above recursive procedure allows to obtain $~\Uup(0,0)=\Uup_0(s_0,Z,\Me)=\Vup(s_0,Z,\Me)$, since the hypothesis of Theorem \ref{fullDynamicProgramming} are satisfied.

We now extend the procedure to an strictly increasing $l$-tuple $\Lambda=\{n_1,\dots,n_l\}$ with $n_l=N_2$. Now $W_{M(\bf S)} = n_{\theta}\delta^2$ for some ${\theta}=1, \dots,l$, then the node $\se_{M(\se)}$ of some trajectory $\se \in \Swe_{\Gamma}$ corresponds by (\ref{gridCorrespondence}) to some $(k_M,n_{\theta})\in\Gamma$, and $\Uup_{M(\se)}({\bf S},Z,\Me) = Z^f(s_{k_{M(\se)}})$. But observe that if $(k_{M(\se)},n_{\theta})$ also corresponds to a node $\widehat{\se}_i$ of a trajectory $\widehat{{\bf S}}$ with $i \neq M(\widehat{{\bf S}})$, by Definition \ref{dynamicBounds} and Theorem \ref{ConvexHullThm}
\[ \Uup_i({\bf \widehat{S}},Z,\Me)=\sup_{\substack {\seup \in
\Sup_{({\bf \widehat{S}},i)} \\ \sedo\in \Sdo_{({\bf \widehat{S}},i)}}}
[\Uup_{i+1}(\seup,Z,\Me)-u_{(\seup,\sedo)}(S_{i+1}^{\textrm{up}}-\widehat{S}_i)] \]

We start the analysis from column $j=n_l=N_2$. Any vertex $(k,N_2)$ corresponds to the node $\se_{M(\se)}$ of a trajectory in $\Swe_{\Gamma}$, then define
\[ \Uup(k,N_2)=Z^f(s_k). \]
For a vertex $(k,j)\in\Gamma$ with $n_{m-1}< j < n_m$, and $k \in [-p\,j,p\,j]$, $\Uup(k,j)$ is given by (\ref{Udefinition}). The vertices on the column $n_{l-1}$ in $\Gamma$, correspond by (\ref{gridCorrespondence}) to trajectories $\bf{S}$ that could have $M({\bf S})=n_{l-1}$ at that node, it is $W_M = n_{l-1}\beta^2$, or continue to get $W_M = n_{l}\beta^2$. Thus for $k^* \in [-pn_{l-1},pn_{l-1}]$, $\Uup(k^*,n_{l-1})$ should take the value $Z^f(s_{k^*})$ in the first case, while in the second case 
its value at that vertex, should be given by (\ref{Udefinition}).
Both situations must be contemplated to compute $\Uup(k,j)$ for $j<n_{l-1}$, by mean of (\ref{Udefinition}), when any of the vertices $(k^*,n_{l-1})$ is reachable from $(k,j)$. Then, observing that the maximum of these two values is the one which contributes to (\ref{CHdynamicBound}), in the referred computation, and by Theorem \ref{ConvexHullThm}, we have
\[ \Uup(k^*,n_{l-1}) = \max\{Z^f(s_{k^*}), C(k^*,n_{l-1})\}.\]
Following the same considerations, $\Uup(k,j)\equiv C(k,j)$ is defined by \eqref{Udefinition} for all $n_{\theta} < j < n_{\theta+1}$ 
with $1\le \theta<l-1$ and $k \in [-p\,j,p\,j]$. For $j=n_{\theta}$ with $1\le \theta<m-1$ and $k \in [-p\,n_{\theta},p\,n_{\theta}]$
\[\Uup(k,n_{\theta})=\max \{Z(s_{k}), C(k,n_{\theta}) \}\]
where $C(k,n_{\theta})$ is given by \eqref{Udefinition}.

To summarize,  $\Uup(k,j)$ for $0 \le j \le N_2$ and $k \in [-pj,pj]$ is given by
\begin{equation} \nonumber
\Uup(k,j)= \left\lbrace \begin{array}{ll}
                         Z^f(s_{k}) & \textrm{if $j=n_l$}\\ \\
                         \max \{Z^f(s_{k}), C(k,j) \} & \textrm{if $j=n_1,n_2\dots,n_{l-1}$}\\ \\
                         C(k,j) & \textrm{in the other case,}
                         \end{array} \right.
\label{eqn:81}
\end{equation}
where $C(k,j)$ is given by (\ref{Udefinition}). With this recursive procedure we can calculate the value of $\Uup(0,0)=\Uup(s_0,Z,\Me)=\Vup(s_0,Z,\Me)$.

Recalling that $\Udo_i({\bf S},Z,\Me)=-\Uup_i({\bf S},-Z,\Me)$, the lower dynamic bounds $\Udo_i(.)$ are computed by a similar procedure.

\end{document}